\documentclass[a4paper, 11pt]{article}

\usepackage{authblk}

\usepackage{yhmath}
\usepackage[most]{tcolorbox}

\usepackage{tikz}
\usetikzlibrary{automata, arrows.meta, positioning}
\usetikzlibrary{decorations.pathmorphing}

\usepackage{mathrsfs}

\usepackage{graphicx} 

\usepackage{hyperref}
\usepackage{verbatim}

\usepackage[utf8]{inputenc}
\usepackage[T1]{fontenc}
\usepackage{textcomp}
 
\usepackage{amsmath, amssymb, amsthm}
\usepackage{tikz}
\usepackage{lipsum}
\usepackage{pdfpages}
\usepackage{multicol} 
\usepackage{stmaryrd}

\usepackage[textwidth=147mm, hcentering, vmargin={24mm,32.5mm}, foot=20mm]{geometry}

\usepackage{booktabs}

\usepackage{makecell}
\usepackage{mathtools}
\usepackage{colonequals}

\usepackage{marginnote}

\usepackage{bm}

\theoremstyle{plain}
\newtheorem{theorem}{Theorem}[section]
\newtheorem{lemma}[theorem]{Lemma}
\newtheorem{proposition}[theorem]{Proposition}

\theoremstyle{definition}
\newtheorem{remark}[theorem]{Remark}
\newtheorem{example}[theorem]{Example}

\newcommand{\N}{\mathbb N}
\newcommand{\Z}{\mathbb Z}

\newcommand{\primes}{\mathscr P}

\newcommand{\bb}{\mathbf{b}}
\newcommand{\bc}{\mathbf{c}}
\newcommand{\bd}{\mathbf{d}}

\newcommand{\bi}{\mathbf{i}}

\newcommand{\bl}{\mathbf{l}}

\newcommand{\br}{\mathbf{r}}
\newcommand{\bs}{\mathbf{s}}
\newcommand{\bt}{\mathbf{t}}
\newcommand{\bu}{\mathbf{u}}

\newcommand{\bw}{\mathbf{w}}

\newcommand{\bV}{\mathbf{V}}

\newcommand{\cA}{\mathcal{A}}
\newcommand{\cB}{\mathcal{B}}

\newcommand{\cF}{\mathcal{F}}
\newcommand{\cG}{\mathcal{G}}

\newcommand{\cL}{\mathcal{L}}
\newcommand{\cM}{\mathcal{M}}

\newcommand{\cR}{\mathcal{R}}

\newcommand{\cT}{\mathcal{T}}

\newcommand{\cV}{\mathcal{V}}

\newcommand{\bbA}{\mathbb{A}}
\newcommand{\bbB}{\mathbb{B}}
\newcommand{\bbC}{\mathbb{C}}
\newcommand{\bbD}{\mathbb{D}}

\newcommand{\bbP}{\mathbb{P}}

\newcommand{\fm}{\mathfrak{m}}

\newcommand*{\abs}[1]{\lvert#1\rvert}

\newcommand{\MSO}{\mathrm{MSO}}

\newcommand{\dom}{\mathrm{dom}}
\newcommand{\MSC}{\mathrm{MSC}}
\newcommand{\SC}{\mathrm{SC}}
\newcommand{\GMSC}{\mathrm{GMSC}}
\newcommand{\GGMSC}{\mathrm{GGMSC}}

\newcommand{\TM}{\mathrm{TM}}

\newcommand{\ML}{\mathrm{ML}}
\newcommand{\GML}{\mathrm{GML}}
\newcommand{\GGML}{\mathrm{GGML}}
\newcommand{\PROP}{\mathrm{PROP}}

\newcommand{\DiamondG}{\langle E \rangle}
\newcommand{\BoxG}{\lbrack E \rbrack}

\newcommand{\El}{\mathrm{Eloise}}
\newcommand{\Ab}{\mathrm{Abelard}}

\newcommand{\leftf}{\mathrm{left}}
\newcommand{\rightf}{\mathrm{right}}

\newcommand{\MCL}{\mathrm{MCL}}
\newcommand{\GMCL}{\mathrm{GMCL}}
\newcommand{\GGMCL}{\mathrm{GGMCL}}
\newcommand{\Rf}{\mathrm{Rf}}

\newcommand{\I}{\mathrm{\mathbf{I}}}
\newcommand{\II}{\mathrm{\mathbf{II}}}

\newcommand{\ordo}{\mathcal{O}}

\newcommand{\CM}{\mathrm{CM}}

\usepackage{enumitem}
\setitemize{noitemsep,topsep=0pt,parsep=3pt,partopsep=0pt}
\setenumerate{noitemsep,topsep=0pt,parsep=3pt,partopsep=0pt}

\newcommand{\tb}[1]{\textcolor{blue}{#1}}
\newcommand{\tre}[1]{\textcolor{red}{#1}}

\usepackage{array,multirow}

\title{Formula size game and model checking for modal substitution calculus}
\author{Veeti Ahvonen, Reijo Jaakkola, Antti Kuusisto}
\affil{Mathematics Research Centre, Tampere University}
\date{\today}

\begin{document}

\maketitle

\begin{abstract}
\noindent
Recent research has applied modal substitution calculus (MSC) and its variants to characterize various computational frameworks such as graph neural networks (GNNs) and distributed computing systems. 
For example, it has been shown that floating-point recurrent graph neural networks have the same expressive power as graded modal substitution calculus GMSC, which is the extension of MSC with counting modalities.
GMSC can be further extended with the counting global modality, resulting in the logic GGMSC, which corresponds to floating-point GNNs with global readout mechanisms. In this paper, we introduce a formula size game that characterizes the expressive power of MSC, GMSC, GGMSC, and related logics. Furthermore, we study the expressiveness and model checking of logics in this family. We prove that MSC and its extensions (GMSC, GGMSC) are as expressive as linear tape-bounded Turing machines, while asynchronous variants are linked to the modal $\mu$-calculus and modal computation logic MCL. We establish that for MSC, GMSC and GGMSC, both the combined and data complexity of model checking are PSPACE-complete, and for their asynchronous variants, both complexities are PTIME-complete. We also establish that for the propositional fragment SC of MSC, the combined complexity of model checking is PSPACE-complete, while for asynchronous SC it is PTIME-complete, and in both cases, the data complexity is constant. As a corollary, we observe that SC satisfiability is PSPACE-complete and NP-complete for its asynchronous variant. Finally, we construct a universal reduction from all recursively enumerable problems to MSC model checking.

\end{abstract}

\section{Introduction}

The article \cite{Kuusisto13} introduced a rule-based bisimulation-invariant logic called modal substitution calculus (MSC), which was used to capture the expressive power of distributed automata (without identifiers). 
Recently, variants of $\MSC$ have been used to characterize various other computing frameworks: in \cite{10.1093/logcom/exae087}, distributed automata with circuits and identifiers were characterized via $\MSC$, recurrent neural networks over floating-point numbers were characterized via the diamond-free fragment of $\MSC$ (or $\mathrm{SC}$) in \cite{10.1093/logcom/exag011, ahvonen_et_al:LIPIcs.CSL.2024.9, ahvonen2025descriptivecomplexityneuralnetworks}, and a logical characterization of recurrent graph neural networks over floating-point numbers was provided using the graded extension of $\MSC$ (or $\GMSC$) in
 \cite{ahvonen2024logical}.
Thus, MSC and its variants have proven to be highly useful in the field of descriptive complexity; see the Related work section for more details.

The logic MSC consists of programs. Each program is defined over a finite set of Boolean variables $\cT$, and each such variable is associated with a base rule and an induction rule defined as follows. 
A base rule for a variable is simply a formula of modal logic and an induction rule for a variable is a formula of modal logic that may also contain variables from $\cT$ as atomic subformulae.
Moreover, each program of MSC is associated with a set of accepting variables.
Informally, each program $\Lambda$ of $\MSC$ is run over a Kripke model $M$. The run of $\Lambda$ over $M$ is an $\omega$-sequence of rounds, where each round determines an interpretation of the variables in $\cT$ over $M$, as follows.
In round zero, a variable from \(\cT\) is true at a node $v$ in $M$ if and only if its base rule is true at $v$. In each subsequent round, a variable from \(\cT\) is true if and only if its induction rule is true when the variables in its induction rule are interpreted based on the truth values of the variables from the previous round.
A program of MSC accepts a node if one of its accepting variables becomes true at that node in some round.

In this paper, we study MSC and its variants SC and GMSC, as well as an extension of GMSC with the (counting) global modality (denoted by $\GGMSC$), concentrating on expressive power and model checking complexity.
We next describe the main contributions in detail; Table \ref{tab:contributions} summarizes our contributions.

We begin by defining two game-theoretic semantics for MSC and its variants GGMSC, GMSC, and SC via two types of semantic games.
Intuitively, each semantic game has two players, Abelard and Eloise; Eloise aims to show that for a given program $\Lambda$ of GGMSC (resp. GMSC, MSC and SC), a Kripke model $M$ and a node $v$ of $M$, $\Lambda$ accepts $v$.
We informally describe two new (co-inductive) game-theoretic semantics for GGMSC below.
\begin{enumerate}
    \item The first game-theoretic semantics is an extension of the standard game-theoretic semantics for graded modal logic with the counting global modality. The game-theoretic semantics is based on games that are played as follows: at the beginning of the game, Eloise chooses an iteration number $t \in \N$, and the players start either from the base rule or the induction rule of an accepting predicate based on $t$. The players play similarly to the standard semantic game for graded modal logic with the counting global modality, and whenever the players face a variable $X$, the game continues from the induction rule of $X$ and $t$ is decreased by one. However, if $t = 0$, then $t$ is not decreased and the game continues from the base rule of $X$. Ultimately, the game ends at a base rule.
    \item The second game-theoretic semantics is based on \emph{global} semantic games, played globally over the entire model $(M, v)$ as follows.
    Informally, Eloise attempts to simulate interpretations over the model $M$ by ``tracing backwards'' from an accepting round to the initial round. In a bit more detail, Eloise begins from an interpretation in which the node $v$ is accepted, and then provides one interpretation for each preceding round, one round at a time, eventually declaring when she has reached the interpretation for the initial round. Abelard may challenge any interpretation during the game. If a challenge occurs, then the players verify that the challenged interpretation is defined correctly w.r.t. the rules of the studied program.
    Intuitively, while the standard compositional semantics builds interpretations inductively, global semantic games build interpretations co-inductively.
\end{enumerate}
The correctness of these games is proved in Theorem \ref{thrm: semantic game} and Theorem \ref{thrm: global semantic game}.
Both systems of game-theoretic semantics are straightforward to modify for GMSC, MSC and SC based on the semantic games for GGMSC.
We also obtain an asynchronous semantics for GGMSC and its fragments by omitting ``clocks'' from the first semantic games (see the section on semantic games for the formal details).

From the first game-theoretic semantics, it is easy to define a formula size game for GGMSC and its fragments GMSC, MSC and SC. 
Intuitively, the formula size game verifies all the semantic games over the given classes of pointed Kripke models simultaneously.
Theorem \ref{thrm: uniform fs-game characterization} shows that the formula size game characterizes the equivalence of classes of pointed Kripke models up to programs of a given size. 
In a bit more detail, given two classes $\bbA$ and $\bbB$ of pointed Kripke models and $k \in \N$, the game is played by two players, Samson and Delilah, and Samson tries to show that there is a program $\Lambda$ of size at most $k$ such that every pointed model in $\bbA$ is accepted by $\Lambda$ and no pointed model in $\bbB$ is accepted by $\Lambda$. Intuitively, our formula size game verifies all possible semantic games at once over the models in the input classes.
The formula size game is inspired by the formula size game defined in \cite{mu-calculus-formula-size}, which in turn builds on \cite{arksiivi, kandalffi, HELLA2022104882}. 
Unlike in \cite{mu-calculus-formula-size}, our characterization is also obtained w.r.t. non-uniform winning strategies over finite Kripke models (cf. Theorem \ref{thrm: finite uniform non-uniform characterization}).
We demonstrate how the formula size game works by providing an inexpressibility result for a fragment of GGMSC (cf. Proposition \ref{prop: formulasize game demo}).

We also study the model checking problem and the expressive power of variants of MSC.
Theorem \ref{thrm: AMSC=MCL=mu-fragment} proves that (over finite models) MSC (resp. GMSC and GGMSC) with asynchronous semantics has the same expressive power as the $\mu$-fragment of the modal $\mu$-calculus (resp. the $\mu$-fragment of the graded modal $\mu$-calculus and the $\mu$-fragment of the global graded modal $\mu$-calculus) and the corresponding modal variant of computation logic CL \cite{HELLA2022104882, jaakkola2022firstorderlogicselfreference, Kuusisto_2014}. We also provide a new logical characterization for deterministic linear tape-bounded Turing machines, over words, via GGMSC, GMSC and MSC (cf. Theorem \ref{thrm: lbas equiv 2-way GMSC}). Informally, a deterministic linear tape-bounded Turing machine is a restricted deterministic Turing machine in which the amount of space used by the Turing machine is bounded by a linear function of the size of the input string.
Furthermore, we obtain that both the combined and data complexity of the model checking problem for MSC, GMSC and GGMSC are PSPACE-complete, while for asynchronous variants of these logics, both complexities are PTIME-complete (cf. Theorem \ref{thrm: ptime model checking} and Theorem \ref{thrm: pspace model checking}).
We also show that the combined complexity of the model checking problem for SC is PSPACE-complete, while for the asynchronous variant of SC it is PTIME-complete (cf. Theorem \ref{thrm: SC async PTIME} and Theorem \ref{thrm: SC PSPACE}). In contrast, the data complexity of the model checking problem for both SC and asynchronous SC is constant, because if a program of SC is fixed, there is only a fixed number of possible inputs. 
The model checking problem for SC is closely related to the well-known result that the reachability problem for Boolean networks is PSPACE-complete; see Section~\ref{sec: model checking} for further discussion.
As a corollary, we obtain that the satisfiability problem for SC is PSPACE-complete, and NP-complete for asynchronous SC.
In contrast, the satisfiability problem for $\MSC$, $\GMSC$ and $\GGMSC$ is undecidable.

Ultimately, we identify a uniform way of giving a computable reduction from any recursively enumerable problem to the model checking problem for MSC by allowing extensions of input words. In other words, we identify a method that can be seen as a ``universal reduction'' from all recursively enumerable problems to the model checking problem for MSC; see Section \ref{sec: meta reduction} for more details.
This is related to the fact that there is a computable reduction from any recursively enumerable problem to the membership problem for linear tape-bounded Turing machines.

\begin{table}[t]
\centering
\caption{Summary of contributions}
\begin{tabular}{p{0.96\linewidth}}
\toprule
\textbf{1.} We introduce two new game-theoretic semantics and a formula size game for MSC, its variants, and related logics. \\
\midrule
\textbf{2.} We show that MSC, GMSC, and GGMSC are as expressive as linear tape-bounded Turing machines, while their asynchronous variants are linked to modal computation logic (MCL) and the modal $\mu$-calculus. \\
\midrule
\textbf{3.} We prove PSPACE-completeness for both the combined and data complexity of the model checking problem for MSC, GMSC, and GGMSC, whereas for asynchronous variants, both complexities are shown to be PTIME-complete. \\
\midrule
\textbf{4.} For SC, we show that the combined complexity of its model checking is PSPACE-complete, and for asynchronous SC, it is PTIME-complete; in both cases, the data complexity is constant. From these results, we conclude that SC satisfiability is PSPACE-complete and NP-complete for its asynchronous variant. \\
\midrule
\textbf{5.} We provide a ``universal reduction'' from all recursively enumerable problems to MSC model checking. \\
\bottomrule
\end{tabular}
\label{tab:contributions}
\end{table}

\subsubsection*{Related work}

The work by Hella et al. \cite{hella2012weak, weak_models} provided pioneering logical characterizations of distributed systems in a constant-iteration setting.
Later, $\MSC$ was introduced by Kuusisto in \cite{Kuusisto13}, where the expressive power of distributed automata (which can compute unboundedly many rounds instead of being limited to the constant-iteration setting) was captured via $\MSC$. The paper also explicitly envisioned a research program of descriptive complexity for distributed computing. Moreover, the article showed that over finite models the $\mu$-fragment of the modal $\mu$-calculus is contained in $\MSC$ and that $\MSC$ and the full modal $\mu$-calculus have orthogonal expressive power (and an analogous result applies to $\GMSC$ and the graded modal $\mu$-calculus). In more detail, it was shown that $\MSC$ cannot express the non-reachability property and the modal $\mu$-calculus cannot express the center-point property, which is a node property stating that all the walks starting from the given node lead to a dead end in the same number of steps.

Recently, in \cite{10.1093/logcom/exae087, dist_circ_mfcs}, it was shown that distributed automata with Boolean circuits and identifiers can be translated into programs of $\MSC$, and vice versa, with a small blow-up in the size in both directions. The diamond-free fragment of $\MSC$ was used to characterize recurrent neural networks and ordinary feedforward neural networks over floating-point numbers in \cite{10.1093/logcom/exag011, ahvonen_et_al:LIPIcs.CSL.2024.9, ahvonen2025descriptivecomplexityneuralnetworks}, again with only a small blow-up in the sizes in both directions. The expressive power of recurrent graph neural networks over floating-point numbers was captured via $\GMSC$ in \cite{ahvonen2024logical}. In the same article, recurrent graph neural networks over reals were captured via graded modal logic extended with infinite disjunctions, and it was also shown that when restricted to the properties definable in monadic second-order logic, recurrent graph neural networks over reals and floats have the same expressive power. The complexity of the model checking problem for $\MSC$ and its variants has not been studied before, although it was already noted in \cite{ahvonen2024logical} that $\GMSC$ is contained in the partial fixed-point logic with choice, which captures PSPACE.

Hella and Vilander \cite{mu-calculus-formula-size} defined a formula size game for the modal $\mu$-calculus that was inspired by the game-theoretic semantics for the modal $\mu$-calculus introduced in \cite{arksiivi, kandalffi, HELLA2022104882}. In \cite{mu-calculus-formula-size}, it was shown that first-order logic (FO) is non-elementarily more succinct than the modal $\mu$-calculus, meaning that there is a sequence $(\varphi_i)_{i\in \N}$ of FO-formulae equivalent to a sequence $(\psi_i)_{i \in \N}$ of formulae of the modal $\mu$-calculus such that the size of $\psi_i$ is non-elementary in the size of $\varphi_i$. Moreover, the formula size game in that article characterized the equivalence of classes of pointed Kripke models up to the formulae of the modal $\mu$-calculus, but required uniform strategies for Samson (one of the players), who tries to show that the classes are not equivalent. Informally, a uniform strategy means that Samson plays according to a formula of the modal $\mu$-calculus. 
However, we prove that uniform strategies are unnecessary over finite models in our formula size game, and we believe that they are unnecessary for all models as well.

There exist other similar formula size games; for example, Adler and Immerman defined a formula size game \cite{adler2003n} (or the Adler--Immerman game), which is highly similar to the formula size game defined by Hella and Vilander \cite{mu-calculus-formula-size}. Our game and the Adler--Immerman game are similar in that, at each position in the game, both maintain a set of multiple subformulae that are currently being verified.
Hella and Väänänen also defined a formula size game in \cite{HellaVäänänen+2015+193+214} for propositional logic and first-order logic, which was another inspiration for the formula size game defined by Hella and Vilander in \cite{mu-calculus-formula-size}. Recently, formula size games have been used to study the link between entropy and formula size \cite{JAAKKOLA2025103615}. For other work related to formula size in logics, see, e.g., \cite{grohe2005succinctness, otto2006bisimulation}.

Regarding related work on modal logics with recursion, there are several relevant systems. These include the modal iteration calculus (MIC) \cite{dawar2004inflationary}, which is essentially a variant of the modal $\mu$-calculus with simultaneous inflationary fixed-point operators; the modal iterative substitution logic (MISL) \cite{tu_modal_approach_substitution}, which extends modal logic with single-step and iterated substitutions of proposition symbols; and the iterated relativization operator in the setting of epistemic programs \cite{miller2005undecidability}. These logics are closely related to each other. For instance, the logic $\mathcal{L}(\mathsf{rel},\mathsf{rel}^*)$ of \cite{miller2005undecidability} can be translated into both MIC \cite{miller2005undecidability} and MISL \cite{tu_modal_approach_substitution}. Unlike MSC, all of these three logics are closed under negation and lack the finite model property \cite{dawar2004inflationary, tu_modal_approach_substitution, miller2005undecidability}, and since MSC has the finite model property (as each program of MSC is equivalent to an infinite disjunction of formulae of modal logic) and is not closed under negation, none of the three logics is contained in MSC. 
Conversely, in MIC, every iteration stabilizes after polynomially many steps (as the fixed-points are inflationary), and thus the data complexity of its model checking is PTIME, whereas it is PSPACE-complete for MSC (Theorem~\ref{thrm: pspace model checking}). Thus, provided PTIME is not PSPACE, we conclude that MSC is not contained in MIC. Since MIC contains $\mathcal{L}(\mathsf{rel},\mathsf{rel}^*)$, the same holds for $\mathcal{L}(\mathsf{rel},\mathsf{rel}^*)$ and MSC. On the other hand, to our knowledge, the data complexity of the model checking problem for MISL has not been determined, and thus the same reasoning does not apply. Currently, it is not known whether MSC can express properties not expressible in MISL.

Concerning further related work on logics with iteration, the paper \cite{balbiani2013dynamic} studies the dynamic logic of propositional assignments (DL-PA), which updates the truth values of proposition symbols by assignments $p \leftarrow \varphi$ that can be combined and iterated by the program constructs of propositional dynamic logic. DL-PA is closely related to SC, i.e., the diamond-free fragment of MSC. Since both logics have the same expressive power as propositional logic, being expressively complete over Boolean assignments, it is not obvious how to compare them. One possibility would be to compare the time series of assignments that they produce. Regarding this theme, the paper \cite{10.1093/logcom/exag011} compared SC and Boolean networks and showed that SC can produce time series with repetitions within minimal periods, which is not possible for Boolean networks. It remains unexplored how DL-PA stands in this respect. 
Finally, \cite{tang2024substitution} studies a variant of first-order logic that allows the substitution of first-order variables by terms, i.e., they introduce operations $[t/x] \varphi$, where a term $t$ substitutes a first-order variable $x$ in a formula $\varphi$. However, this logic has no recursion mechanism, and operates on the level of first-order logic, so despite some similarities to our framework, a direct comparison here is not possible.

\section{Preliminaries}

We let $\Z_+$ denote the set of positive integers and $\N$ the set of non-negative integers. For $n \in \Z_+$, define $[n] \colonequals \{1, \ldots, n\}$, and for $n \in \N$, define $[0;n] \colonequals \{0, 1, \ldots, n\}$.
For a set $S$, we let $\cM(S)$ denote the set of multisets over~$S$, i.e., the set of functions \mbox{$S \to \N$}.

We let $\PROP$ denote the countably infinite set of \textbf{proposition symbols} and we let $\mathrm{VAR}$ denote the countably infinite set of \textbf{schema variables}. Given a $\Pi \subseteq \PROP$, a \textbf{Kripke model over $\Pi$} (or simply a $\Pi$-model) is a tuple $(W, R, V)$, where $W$ is a non-empty domain (or a set of nodes), $R \subseteq W \times W$ is an accessibility relation and $V \colon W \to \wp(\Pi)$ is a valuation function, where $\wp(\Pi)$ denotes the \textbf{power set} of $\Pi$. A \textbf{pointed Kripke model} is a pair $((W, R, V), w)$, where $w \in W$. Moreover, the set of \textbf{out-neighbors} (or \textbf{successors}) of a node $v$ is $\{\, u \in W \mid (v, u) \in R \,\}$. 

Let $X$ and $Y$ be sets. If $f \colon X \rightharpoonup Y$ is a partial function, $x \in X$ and $y \in Y$, then we let $f' = f[y/x]$ denote the partial function $f' \colon X \rightharpoonup Y$ defined by $f'(z) = f(z)$ when $z \neq x$, and $f'(x) = y$.
We let $Y^X$ denote the set of functions from $X$ to $Y$. Given a function $g \in Y^X$, the \textbf{range} of $g$ is the set $\{\, g(x) \in Y \mid x \in X\,\}$. Given a relation $R$ over $X$, if $(x, y) \in R$, then we say that $y$ is an \textbf{$R$-successor} of $x$. 

\subsection{Variants of substitution calculus}

Given a $\Pi \subseteq \PROP$ and a $\cT \subseteq \mathrm{VAR}$, 
the set of \textbf{$(\Pi, \cT)$-schemata of graded modal substitution calculus with the (counting) global modality} 
(or $\GGMSC$) is defined by the following grammar
\[
\varphi \coloncolonequals \bot \mid \top \mid p  \mid X \mid \neg \varphi \mid \varphi \lor \varphi \mid \varphi \land \varphi \mid \Diamond_{\geq k} \varphi \mid \Box_{< k} \varphi \mid \langle E \rangle_{\geq k} \varphi \mid \BoxG_{< k} \varphi,
\]
where $k \in \Z_+$, $p \in \Pi$ and $X \in \cT$.\footnote{The connectives $\rightarrow$, $\leftrightarrow$ are considered abbreviations in the
usual way. 
We also use the abbreviations $\Diamond\varphi \coloneqq \Diamond_{\geq 1}\varphi$, $\Box\varphi \coloneqq \Box_{<1}\varphi$, $\Diamond_{ < k} \varphi \colonequals \neg \Diamond_{\geq k} \varphi$, 
$\Box_{\geq k} \varphi \colonequals \neg \Box_{< k} \varphi$,
$\Diamond_{=k} \varphi \colonequals \Diamond_{\geq k} \varphi \land \Diamond_{< k+1} \varphi $, $\Box_{=k} \varphi \colonequals \Box_{\geq k} \varphi \land \Box_{< k+1} \varphi $, $\Diamond_{=0} \varphi \colonequals \Diamond_{< 1} \varphi$ and $\Box_{=0} \varphi \colonequals \Box_{< 1} \varphi $. The abbreviations $\langle E \rangle \varphi$, $\BoxG \varphi$, $\BoxG_{\geq k} \varphi$, $\langle E \rangle_{=k} \varphi$, $\BoxG_{=k} \varphi$, $\DiamondG_{=0} \varphi$ and $\BoxG_{=0} \varphi$ are defined analogously. 
} 
Moreover, when we exclude $\cT$ from the grammar, we have the set of \textbf{$\Pi$-formulae of graded modal logic with the (counting) global modality} (or $\GGML$). 
Now, assume that $\{X_1, \ldots, X_k\} = \cT$ is a set of $k$ distinct schema variables.
A $(\Pi, \cT)$-\textbf{program} $\Lambda$ of $\GGMSC$ consists of two lists of \textbf{rules}
\[
\begin{aligned}
    &X_1 (0) \colonminus \varphi_1\qquad\qquad &&X_1 \colonminus \psi_1 \\
    &\vdots &&\vdots \\
    &X_k (0) \colonminus \varphi_k\qquad\qquad &&X_k \colonminus \psi_k \\
\end{aligned}
\]
where each $\varphi_i$ is a $\Pi$-formula of $\GGML$ and each $\psi_i$ is a $(\Pi, \cT)$-schema of $\GGMSC$.
Moreover, each program is also associated with a set of \textbf{accepting predicates} $\cA\subseteq \cT$.
Strings of the form $X_i (0) \colonminus \varphi_i$ are called \textbf{base rules}, while $X_i \colonminus \psi_i$ are called \textbf{induction rules}. The variable $X_i$ at the front of the rule is called the \textbf{head predicate} and the formulae $\varphi_i$ and $\psi_i$ are called the \textbf{bodies}.  
By omitting schema variables $\cT$, we let a $\Pi$-program of $\GGMSC$ refer to a $(\Pi, \cT)$-program of $\GGMSC$ for some $\cT \subseteq \mathrm{VAR}$.

The truth of $\Pi$-formulae $\varphi$ of $\GGML$ in a pointed $\Pi$-model $(M, w)$ (denoted by $M, w \models \varphi$) is defined as follows. The semantics for Boolean connectives, as well as for $\bot$ and $\top$, is the usual one, while for proposition symbols $p \in \Pi$, we define $M, w \models p$ iff $p \in V(w)$. For $\Diamond_{\geq k} \varphi$, we define $M, w \models \Diamond_{\geq k} \varphi$ iff $\abs{\{\,v \mid M, v \models \varphi, (w, v) \in R\,\}} \geq k$, and for $\Box_{< k} \varphi$, we define $M, w \models \Box_{< k} \varphi$ iff $\abs{\{\, v \mid M, v \not\models \varphi, (w, v) \in R\,\}} < k$.
Moreover, we define $M, w \models \langle E \rangle_{\geq k} \varphi$ iff $\abs{\{\, v \mid M, v \models \varphi, v \in W\,\}} \geq k$, and for $\BoxG_{< k} \varphi$, we define $M, w \models \BoxG_{< k} \varphi$ iff $\abs{\{\, v \mid M, v \not\models \varphi, v \in W\,\}} < k$.\footnote{Note that, by definition, $\Box_{<k}$ is the corresponding dual operator for $\Diamond_{\geq k}$, i.e., $\Box_{<k} \varphi$ is logically equivalent to $\neg \Diamond_{\geq k} \neg \varphi$. Analogously, $\BoxG_{<k} \varphi$ is logically equivalent to $\neg \DiamondG_{\geq k} \neg \varphi$.}

Now, we define the semantics for $\Lambda$.
First, we define the \textbf{$n$th induction formula $X^n_i$} (w.r.t. $\Lambda$) (or the induction formula of $X_i$ in \textbf{round} $n \in \N$) for each head predicate $X_i$ recursively as follows.
    We define $X^0_i \colonequals \varphi_i$.
    %
    The $\GGML$-formula $X^{n+1}_i$ is obtained from $\psi_i$ by replacing each variable $X_j$ by $X_j^n$.
%
%
Analogously, given a $(\Pi, \cT)$-schema $\varphi$, we let $\varphi^{k}$ denote the $\GGML$-formula obtained from $\varphi$ by replacing each variable $X_i$ by $X_i^k$.

Now, we define $M, w \models \Lambda$ and say that $(M, w)$ is \textbf{accepted} by $\Lambda$ iff there is an $n \in \N$ such that $M, w \models X^n$ for some accepting predicate $X$. 
Given two classes, $\bbA$ and $\bbB$, of pointed $\Pi$-models, we say that $\Lambda$ (resp. a $\Pi$-formula $\varphi$) \textbf{separates $\bbA$ from $\bbB$} 
if every $(A, w) \in \bbA$ is accepted by $\Lambda$ (resp. $A, w \models \varphi$) and every $(B, w) \in \bbB$ is not accepted 
by $\Lambda$ (resp. $B, w \not\models \varphi$).

Next, we define important fragments of $\GGMSC$ and $\GGML$.
A $(\Pi, \cT)$-program of \textbf{graded modal substitution calculus} (or $\GMSC$) \cite{ahvonen2024logicalcharacterizationsrecurrentgraph} is a $(\Pi, \cT)$-program of $\GGMSC$ that does not contain global diamonds $\DiamondG_{\geq k}$ or boxes $\BoxG_{< k}$. Similarly, a $(\Pi, \cT)$-program of \textbf{modal substitution calculus} (or $\MSC$) is a program of $\GMSC$ that can only contain diamonds of the type $\Diamond$ or boxes of the type $\Box$. Analogously, we define the set of $\Pi$-formulae of \textbf{graded modal logic} (or $\GML$) and the set of $\Pi$-formulae of \textbf{modal logic} (or $\ML$).
Programs of \textbf{substitution calculus} (or $\mathrm{SC}$) do not include any occurrences of diamonds or boxes; they are interpreted over models of propositional logic, i.e., Kripke models without an accessibility relation and with a single node.

\begin{example}\label{example: centre-point}
    A pointed model $(M,w)$ has the \textbf{center-point property} if there exists $n \in \N$ such that each walk starting from $w$ leads in exactly $n$ steps to a node $v$ that has no out-neighbors. 
    The program $X(0) \colonminus \Box \bot$, $X \colonminus \Diamond X \land \Box X$ of $\MSC$ accepts precisely the pointed models that have the center-point property.  As already noted in \cite{Kuusisto13}, this property is not expressible in $\MSO$.
\end{example}

The \textbf{size} $\abs{\Lambda}$ of a $\Pi$-program $\Lambda$ of $\GGMSC$ is the number of occurrences of proposition symbols, schema variables, constant symbols, negations and the connectives $\lor$ and $\land$ in the bodies of $\Lambda$, plus the sum of the counting thresholds $k$ over all occurrences of diamonds and boxes $\Diamond_{\geq k}$, $\Box_{<k}$, $\DiamondG_{\geq k}$ and $\BoxG_{<k}$ in the bodies of $\Lambda$.
For instance, each occurrence of $\Diamond_{\geq 1}$ contributes $1$ to the size, while each occurrence of $\Diamond_{\geq 3}$ contributes $3$.

A program is in \textbf{negation normal form} if the only negated subschemata are negated proposition symbols or schema variables.
A program is in \textbf{strong negation normal form} if the only negated subschemata are negated proposition symbols.
The following lemma proves that each program can be translated into an equivalent program (with only a linear size increase) that is in strong negation normal form.

\begin{lemma}\label{lem: strong negation normal form}
    Given a $\Pi$-program $\Lambda$ of $\GGMSC$ of size $n$, there exists an equivalent 
    program of $\GGMSC$ in strong negation normal form of size $\ordo(n)$.
\end{lemma}
\begin{proof}
    We may assume that $\Lambda$ is in negation normal form, since it is easy to obtain an equivalent program in negation normal form with only a linear increase in size.
    We construct an equivalent program $\Lambda_d$ that is in strong negation normal form as follows. 
    For each head predicate $X$ of $\Lambda$ with the rules $X(0) \colonminus \varphi$ and $X \colonminus \psi$, we introduce a fresh head predicate $X_d$ with the rules $X_d(0) \colonminus \varphi_d$ and $X_d \colonminus \psi_d$ defined as follows. The formula $\varphi_d$ is obtained by transforming $\neg\varphi$ into negation normal form, and $\psi_d$ is obtained by transforming $\neg\psi$ into negation normal form and then replacing each $\neg Y$ by $Y_d$.
    Finally, we modify each original rule $X \colonminus \psi$ in $\Lambda$ by replacing each $\neg Y$ by $Y_d$.

    It is clear that $\Lambda_d$ is in strong negation normal form and the size is linear in the size of $\Lambda$.
    Now, by a routine induction, it is easy to show that for each pointed $\Pi$-model $(M, w)$ and for each head predicate $X$ that appears in $\Lambda$, the following holds:
    $M, w \models X^n$ w.r.t. $\Lambda$ iff $M, w \models X^n$ w.r.t. $\Lambda_d$,
    and
    $M, w \not\models X^n$ w.r.t. $\Lambda$ iff $M, w \models X_d^n$ w.r.t. $\Lambda_d$. 
\end{proof}

Lastly, we point out that each $\GGMSC$-schema can be interpreted as a $\GGML$-formula when a Kripke model is associated with an interpretation of the variables.
A \textbf{Kripke model over $(\Pi, \cT)$} (or $(\Pi, \cT)$-model) is a pair $M_g \colonequals (M, g)$, where $g \colon W \to \wp(\cT)$ is a function (called a \textbf{labeled tuple}) and $M$ is a $\Pi$-model.
The truth of a $(\Pi, \cT)$-schema $\psi$ in a pointed $(\Pi, \cT)$-model $(M_g, w)$ is defined as follows. We define $M_g, w \models X$ iff $X \in g(w)$. The rest of the semantics for proposition symbols, constant symbols, connectives, diamonds and boxes is defined analogously to $\GGML$. 
A labeled tuple $f$ is \textbf{suitable} for a $\Pi$-model $M$ if the domain of $f$ is the domain of $M$.
Note that in each round $n \in \N$, in each $\Pi$-model $M = (W, R, V)$, a $(\Pi, \cT)$-program $\Lambda$ induces a labeled tuple $g_n \colon W \to \wp(\cT)$ called the \textbf{global configuration} in round $n$ defined as follows. For each $w \in W$, we define $g_n(w) = \{\, X \mid M, w \models X^n \,\}$.

\section{Semantic games}

In this section, we define two new game-theoretic semantics for $\GGMSC$ and its variants. Informally, both semantics are based on semantic games which are played by two players, Eloise and Abelard, where for a given $\GGMSC$-program $\Lambda$ and a pointed model $(M, w)$, Eloise tries to show that $(M, w)$ is accepted by $\Lambda$ and Abelard opposes this. There are two main differences between the semantics. 
The first game-theoretic semantics is an extension of the standard game-theoretic semantics for $\GGML$ and is based on semantic games. In each such game, the current game position stores the current node, subformula and iteration round being verified. The second game-theoretic semantics is based on games in which a game position is simply a labeled tuple over the domain of the input model. Both game-theoretic semantics define co-inductive semantics for $\GGMSC$.
The correctness of these game-theoretic semantics is formally proved in Theorem \ref{thrm: semantic game} and Theorem \ref{thrm: global semantic game}. Moreover, we study an asynchronous variant of the first game.

\subsection{Standard semantic game}

Given a pointed $\Pi$-model $(M, w) = ((W, R, V), w)$ and 
a $\Pi$-program $\Lambda$ of $\GGMSC$,
the \textbf{semantic game $\cG(M, w, \Lambda)$} is defined as follows. The game has two players, \textbf{Abelard} and \textbf{Eloise}. 
The \textbf{positions} of the game $\cG(M, w, \Lambda)$ are tuples $(\bV, v, \varphi, k)$, where 
    %
    %
    $\bV \in \{\El, \Ab\}$ is the current \textbf{verifier} and the other player is the \textbf{falsifier},
    %
    %
    $v$ is a node in $M$,
    %
    %
    %
    $\varphi$ is a subschema of $\Lambda$ and
    %
    %
    $k \in \N$ is the \textbf{iteration round}.
Intuitively, a position $(\bV, v, \varphi, k)$ corresponds to the following claim:
\[
\text{``The player $\bV$ can verify $\varphi$ at $v$ in $k$ iteration rounds.''}
\]
An \textbf{initial position} is a position $(\El, w, \varphi, k)$, where $\varphi$ is either the body of the base rule or the body of the induction rule of an accepting predicate.

A play of the game $\cG(M, w, \Lambda)$ begins from an initial position that is chosen by Eloise from the set of initial positions. Moreover, if the set of initial positions of the game is empty (i.e., $\Lambda$ does not have accepting predicates), then Eloise automatically loses and Abelard wins. So, strictly speaking, there is a ``starting position'' where Eloise chooses an initial position.

The \textbf{rules} of the game are defined as follows.
\begin{enumerate}
    \item In a position $(\bV, v, \top, \ell)$ (resp. in a position $(\bV, v, \bot, \ell)$), the game ends and the verifier wins (resp. the falsifier wins).
    \item In a position $(\bV, v, p, \ell)$, where $p \in \Pi$, the game ends. The verifier wins if $p \in V(v)$. Otherwise, the falsifier wins.
    \item In a position $(\bV, v, \neg \psi, \ell)$, the game continues from the position $(\bV', v, \psi, \ell)$, where $\bV' \in \{\El,\Ab\} \setminus \{\bV\}$.
    \item In a position $(\bV, v, \psi \land \theta, \ell)$, the falsifier chooses a conjunct $\chi \in \{\psi, \theta\}$ and the game continues from the position $(\bV, v, \chi, \ell)$. 
    \item In a position $(\bV, v, \psi \lor \theta, \ell)$,
    the verifier chooses a disjunct $\chi \in \{\psi, \theta\}$ and the game continues from the position $(\bV, v, \chi, \ell)$.
    \item In a position $(\bV, v, \Diamond_{\geq k} \psi, \ell)$, the verifier $\bV$ chooses a set $\{u_1, \ldots, u_k\}$ of $k$ distinct out-neighbors of $v$, then the falsifier chooses a node $u \in \{u_1, \ldots, u_k\}$ and the game continues from $(\bV, u, \psi, \ell)$. If the verifier cannot choose $k$ out-neighbors of $v$, then the falsifier wins.
    \item In a position $(\bV, v, \Box_{< k} \psi, \ell)$, the game continues as in $(\bV, v, \Diamond_{\geq k} \psi, \ell)$, but the roles of the verifier and the falsifier are switched.
    \item In a position $(\bV, v, \DiamondG_{\geq k} \psi, \ell)$, the verifier $\bV$ chooses a set $\{u_1, \ldots, u_k\}$ of $k$ distinct nodes from $W$, then the falsifier chooses a node $u \in \{u_1, \ldots, u_k\}$ and the game continues from $(\bV, u, \psi, \ell)$. If the verifier is unable to choose $k$ nodes, the falsifier wins.
    \item In a position $(\bV, v, \BoxG_{< k} \psi, \ell)$, the game continues as in $(\bV, v, \DiamondG_{\geq k} \psi, \ell)$, but the roles of the verifier and the falsifier are switched.
    \item In a position $(\bV, v, X, \ell)$, the game continues as follows. If $\ell > 0$, then the game continues from the position $(\bV, v, \psi, \ell-1)$, where $\psi$ is the body of the induction rule of $X$. If $\ell = 0$, then the game continues from the position $(\bV, v, \theta, 0)$, where $\theta$ is the body of the base rule of $X$. 
\end{enumerate}
Note that, if a position $(\bV, v, \theta, 0)$ is reached during a play, then the play will end at a position $(\bV', v', \varphi, 0)$, where $\varphi \in \Pi \cup \{\top, \bot\}$, since the base rules are just $\GGML$-formulae. 
We write $M, w \Vdash \Lambda$ iff Eloise has a \textbf{winning strategy} in $\cG(M,w, \Lambda)$. Furthermore, it is trivial to show that if Eloise has a winning strategy in $\cG(M,w, \Lambda)$ starting from a position $(\El, w, \varphi, k)$, where $\varphi$ is a formula of $\GGML$, then for all $\ell \in \N$, Eloise has a winning strategy in $\cG(M,w, \Lambda)$ starting from a position $(\El, w, \varphi, \ell)$.

\begin{remark}
It is easy to obtain game-theoretic semantics for $\GGMSC$-schemata as follows. Given a $( \Pi, \cT ) $-schema $\psi$ of $\GGMSC$ and a pointed $( \Pi, \cT ) $-model $( M_g, w ) $, the \textbf{semantic game} $\cG(M_g, w, \psi)$ is played like a semantic game for $\GGMSC$-programs, but the game consists of a single initial position $(\El, w, \psi, 0)$ and variables are handled as follows. In a position $( \bV, v, X, 0 ) $ (resp. in a position $( \bV, v, \neg X, 0)$) the game ends and the current verifier $\bV$ wins if $X \in g(v)$ (resp. if $X \notin g(v)$), and otherwise, the falsifier wins. Thus, in these games, $0$ can be omitted from all the game positions.
Furthermore, we can write $M_g, w \Vdash \psi$ iff Eloise has a winning strategy in $\cG(M_g, w, \psi)$ and in the case where $\psi$ is a formula of $\GGML$, we may omit $g$.
\end{remark}

Now we prove the correctness of our semantic games.
\begin{theorem}\label{thrm: semantic game}
    For each pointed $\Pi$-model $(M, w)$ and $\Pi$-program $\Lambda$ of $\GGMSC$,  
    \[
    M, w \models \Lambda \iff M, w \Vdash \Lambda.
    \]
\end{theorem}
\begin{proof}
    We prove a more general result from which the claim follows. Given a subschema $\psi$ of $\Lambda$, we prove by induction on $k \in \N$ that
    $M, w \models \psi^k$ iff Eloise has a winning strategy in $\cG(M, w, \Lambda)$ starting from the position $(\El, w, \psi, k)$.

    We prove the base case for $k = 0$ by induction on the structure of $\psi$. First, assume that $\psi$ is a $\GGML$-formula.
    \begin{itemize}
        \item If $\psi \colonequals p$, where $p \in \Pi$, then $p^0 = p$, and thus the claim holds trivially. Also, the cases for negated proposition symbols, Boolean connectives, $\top$ and $\bot$ are trivial. 
        \item Assume that $\psi \colonequals \Diamond_{\geq m} \theta$. Now, 
        $M, w \models (\Diamond_{\geq m} \theta)^0$ iff there are at least $m$ distinct successors 
        $\{u_1, \ldots, u_m\}$ of $w$ such that $M, u_i \models \theta^0$ for all $i \in [m]$. By the induction hypothesis, Eloise has a winning strategy in $\cG(M, w, \Lambda)$ starting from $(\El, u_i, \theta, 0)$; this is equivalent to Eloise having a winning strategy in $\cG(M, w, \Lambda)$ starting from the position $(\El, w, \Diamond_{\geq m} \theta, 0)$. Other diamonds and boxes are handled analogously.
    \end{itemize}
    Now, the proof is straightforward to extend to $\GGMSC$-schemata by a routine induction.
    Assume that $\psi \colonequals Y$ for some head predicate $Y$ of $\Lambda$ and let $\varphi_Y$ denote the body of its base rule. Now, by the induction hypothesis, we have $M, w \models Y^{0}$ iff $M, w \models \varphi_Y^{0}$ iff Eloise has a winning strategy in the game $\cG(M, w, \Lambda)$ starting from the position $(\El, w, \varphi_Y, 0)$. 

    Assume that the induction hypothesis holds for $0 \leq \ell < k$. We prove the claim for $k$. 
    %
    %
    %
        %
        %
        The cases of literals, constant symbols, Boolean connectives, diamonds, and boxes are handled similarly to the case $k=0$.
        %
        %
        %
        Assume that $\psi \colonequals Y$ for some head predicate $Y$ of $\Lambda$ and let $\psi_Y$ denote the body of its induction rule. Now, by the induction hypothesis, we have $M, w \models Y^{k}$ iff $M, w \models \psi_Y^{k-1}$ iff Eloise has a winning strategy in the game $\cG(M, w, \Lambda)$ starting from the position $(\El, w, \psi_Y, k-1)$. 
        %

    Now, we can prove the desired theorem.
    Indeed, $M, w \models \Lambda$ iff there exists a $k \in \N$ and an accepting predicate $X$ of $\Lambda$ such that $M, w \models X^k$. Moreover, by the result above, this is equivalent to the fact that Eloise has a winning strategy in $\cG(M, w, \Lambda)$ starting from $(\El, w, X, k)$ iff Eloise has a winning strategy in $\cG(M, w, \Lambda)$.
\end{proof}
\begin{remark}
From the semantic games for $\GGMSC$, it is straightforward to obtain the corresponding semantic games for $\GMSC$, $\MSC$ and $\SC$, and to prove the corresponding result of Theorem \ref{thrm: semantic game} for these logics.    
\end{remark}

\subsection{Global semantic game}

The \textbf{global semantic game} $\cG^*(M, w, \Lambda)$ is played over a pointed $\Pi$-model $(M,w) = ((W, R, V), w)$ and a $(\Pi, \cT)$-program $\Lambda$ of $\GGMSC$. Again, the game has two players, Abelard and Eloise, and Eloise tries to show that $(M,w)$ is accepted by $\Lambda$ and Abelard opposes this. 
A \textbf{position} of the game is simply a labeled tuple $f \colon W \to \wp( \cT ) $.

Intuitively, Eloise tries to show that she can start from a global configuration of $\Lambda$ over $M$, where $w$ is accepted, and then move ``backwards'' to the initial global configuration of $\Lambda$ over $M$. 
The role of Abelard is to check that each global configuration given by Eloise is valid w.r.t. the rules of $\Lambda$.

More formally, the set of \textbf{initial positions} of the game consists of the labeled tuples $f$ such that for at least one accepting predicate $X$ of $\Lambda$, we have $X \in f(w)$. At the beginning of the game, Eloise chooses an initial position from the set of initial positions. If the set of initial positions is empty (i.e., $\Lambda$ does not have any accepting predicates), then Abelard wins. Again, strictly speaking, there is a ``starting position'' where Eloise chooses an initial position.

The \textbf{rules} of $\cG^* (M, w, \Lambda)$ are defined as follows. 
In each position $f$ of the game, Eloise first declares whether $f$ is the final position of the game.
\begin{enumerate}
    \item If Eloise declares that $f$ is the final position of the game, then the game continues as follows.
    %
        Abelard chooses a node $v \in W$ and a head predicate $X$ of $\Lambda$. Let $\varphi$ denote the body of the base rule of $X$.
        %
        %
        Then Eloise wins if the following holds: $X \in f(v)$ iff Eloise has a winning strategy in $\cG(M, v, \varphi)$. 
        %

    %
    %
    %
    %
    \item If Eloise declares that $f$ is not the final position of the game, then the game continues as follows. Eloise gives a labeled tuple $g \in \wp(\cT)^W$, then Abelard can decide to \textbf{challenge} $g$, or not. If Abelard does not challenge $g$, then the game continues from $g$. If Abelard challenges $g$, then the game continues as follows. 
    %
    %
    %
        Abelard chooses a node $v \in W$ and a head predicate $X$ of $\Lambda$. Let $\psi$ denote the body of the induction rule of $X$.
        %
        %
        Then Eloise wins if the following holds: $X \in f(v)$ iff Eloise has a winning strategy in $\cG(M_g, v, \psi)$. 
        %
        %
    %
    %
\end{enumerate}
If the play continues for infinitely many rounds, then Abelard wins. 
We write $M, w \Vvdash \Lambda$ iff Eloise has a winning strategy in $\cG^* (M, w, \Lambda)$.

Next, we prove the correctness of our global semantic game.
\begin{theorem}\label{thrm: global semantic game}
For each pointed $\Pi$-model $(M, w)$ and $\Pi$-program $\Lambda$ of $\GGMSC$,
\[
M, w \models \Lambda \iff M, w \Vvdash \Lambda.
\]
\end{theorem}
\begin{proof}
Assume that $M, w \models \Lambda$ and $M = (W, R, V)$. Let $k \in \N$ be the smallest round in which $M, w \models X^k$ for some accepting predicate $X$. Let $(g_0, \ldots, g_k)$ be the sequence of global configurations of $\Lambda$ in $M$ in the rounds $0, \ldots, k$. We construct a winning strategy for Eloise, where she starts by choosing $g_k$ as the initial position of the game and during the game in each position $g_i$ she picks a position $g_{i-1}$, where $i \in [k]$. Note that $g_k$ is an initial position since $X \in g_k(w)$. After reaching $g_0$, Eloise declares that $g_0$ is the final position of the game. Clearly, if Abelard does not challenge Eloise at any position, Eloise wins. On the other hand, if Abelard challenges $g_i$ for $i \in [0; k-1]$, then he will always lose for the following reason. 
Let $v \in W$ and let $Y$ be a head predicate of $\Lambda$. We let $\psi$ denote the body of the induction rule of $Y$, and we let $\varphi$ denote the body of the base rule of $Y$.
Now, we have $Y \in g_{i+1}(v)$ iff $M, v \models \psi^i$ iff $M_{g_i}, v \models \psi$ iff Eloise has a winning strategy in $\cG(M_{g_i}, v, \psi)$. Lastly, we also have $Y \in g_{0}(v)$ iff $M, v \models \varphi$.

For the converse direction, assume that $M, w \Vvdash \Lambda$, i.e., Eloise has a winning strategy $\sigma$ in $\cG^*(M, w, \Lambda)$. Let $f_k, \ldots, f_0$ enumerate the positions induced by $\sigma$ over $\cG^*(M, w, \Lambda)$ in the case where Abelard does not challenge Eloise in any round, where $f_k$ is an initial position and $f_i$ is followed by $f_{i - 1}$ during the game. Let $g_k, \ldots, g_0$ enumerate the global configurations of $\Lambda$ in $M$ from round $k$ to round $0$. 
We show by induction on $i \in [0; k]$ that $f_i = g_i$, and then the claim follows. The case $f_0 = g_0$ is trivial. 
Assume that $f_j = g_j$ for every $j < i$; we show that $f_i = g_i$ also holds. 
Let $X$ be a head predicate of $\Lambda$, $\psi$ the body of its induction rule, and $u \in W$. 
Now, we have $M_{g_{i-1}}, u \models \psi$ iff $M, u \models \psi^{i-1}$ iff $M, u \models X^{i}$ iff $X \in g_i(u)$ by the definition of global configurations.
Moreover, $M_{f_{i-1}}, u \models \psi$ iff Eloise has a winning strategy in $\cG(M_{f_{i-1}}, u, \psi)$ iff $X \in f_i(u)$.
By the induction hypothesis, $f_{i-1} = g_{i-1}$, and thus $X \in g_i(u)$ iff $X \in f_i(u)$, i.e., $f_i = g_i$.
\end{proof}

\begin{remark}
Again, it is straightforward to define the corresponding global semantic games for $\GMSC$, $\MSC$ and $\SC$ from the global semantic games of $\GGMSC$, and to obtain the corresponding result of Theorem \ref{thrm: global semantic game} for these logics.    
\end{remark}

\subsection{Asynchronous semantics}

In this section, we define asynchronous game-theoretic semantics for $\GGMSC$.
Let $(M, w)$ be a pointed $\Pi$-model and $\Lambda$ a $(\Pi, \cT)$-program of $\GGMSC$.

The \textbf{asynchronous semantic game $\cA\cG(M, w, \Lambda)$ of $\GGMSC$} is defined as follows. The game is defined analogously to $\cG(M, w, \Lambda)$, except that the game positions are tuples of the form $(\bV, v, \varphi)$ instead of $(\bV, v, \varphi, k)$, i.e., the game positions do not record the current iteration round. Furthermore, the set of initial positions is the set of tuples of the form $(\El, w, \varphi)$, where $\varphi$ is the body of the induction rule or the base rule of an accepting predicate. The rules of the game are defined analogously to those of $\cG(M, w, \Lambda)$, except that the variables are handled as follows:
\begin{itemize}
    \item In a position $(\bV, v, X)$, the game continues as follows. Let $\varphi_X$ and $\psi_X$ denote the body of the base rule and the body of the induction rule of $X$, respectively. The verifier chooses a formula $\chi$ from the set $\{\varphi_X, \psi_X\}$ and the game continues from the position $(\bV, v, \chi)$.
\end{itemize}
That is, when a variable is under verification, the current verifier can choose whether the variable is iterated or not. Notice that the game can continue for infinitely many rounds, and if that happens, then \emph{neither} player wins.

\textbf{Asynchronous $\GGMSC$} (or simply $\GGMSC^A$) is the set of programs of $\GGMSC$, where the acceptance of each program is based on its asynchronous semantic games, i.e., a $\Pi$-program $\Lambda$ of $\GGMSC^A$ accepts a pointed $\Pi$-model $(M, w)$ iff Eloise has a winning strategy in $\cA\cG(M, w, \Lambda)$. Analogously, we define $\GMSC^A$, $\MSC^A$ and $\SC^A$.

\section{Formula size game}

In this section, we define a formula size game for $\GGMSC$ (inspired by the formula size game defined for the modal $\mu$-calculus in \cite{mu-calculus-formula-size}) and prove that the game characterizes the logical equivalence of classes of pointed models up to programs of $\GGMSC$ of a given size (cf. Theorem \ref{thrm: uniform fs-game characterization}). We start by introducing auxiliary notions and notations, and then we formally define the game and consider its properties. 

\subsection{Syntax forest}

Informally, the syntax forest of a $\GGMSC$-program contains the syntax tree for each body of each rule and additional back edges which point from each variable to the corresponding bodies of its rules. 
These syntax forests are illustrated in Example \ref{example: syntax forest}.

We start by defining the concepts of trees and forests. Given a non-empty set $L$ of labels, a \textbf{node-labeled directed tree} (over $L$) is a tuple $(V, E, \lambda)$, where $V$ is a non-empty set of \textbf{nodes}, $\lambda \colon V \to L$ is a \textbf{labeling function} and $E \subseteq V \times V$ is a set of \textbf{edges} such that the following holds. There is a node $v \in V$ called the \textbf{root} such that for every node $u \in V \setminus \{v\}$ there is a unique \textbf{directed walk} from $v$ to $u$, i.e., a sequence of nodes $v_1, v_2, \ldots, v_k$, where $k > 1$, $v_1 = v$, $v_k = u$ and $(v_i, v_{i+1}) \in E$ for every $i \in [k-1]$. A \textbf{directed path} is a directed walk in which all nodes are distinct.
A \textbf{node-labeled directed forest} $(V, E, \lambda)$ is a disjoint union of node-labeled trees. 
We may also associate multiple node-labeling functions with trees and forests instead of one.
A node-labeled directed forest $(V, E, \lambda)$ \textbf{with back edges} $B \subseteq V \times V$ is a tuple $(V, E, B, \lambda)$, where for every $(v, u) \in B$, there is no directed walk from $v$ to $u$ through edges $E$.   

From now on, we may omit the word \emph{directed} in the concepts of directed trees, directed walks and so on above, since we do not consider non-directed trees, walks and so on. 

Now, we may define the \textbf{syntax tree} of a schema. Let $\psi$ be a $(\Pi, \cT)$-schema of $\GGMSC$. The syntax tree of $\psi$ is a node-labeled tree $T_\psi = (V_\psi, E_\psi, \lambda_\psi)$ defined as follows. 
The set $V_\psi$ consists of the occurrences of the subschemata of $\psi$, the relation $E_\psi$ corresponds to the subschema relation between the subschemata of $\psi$, and the function
\[
\lambda_\psi \colon V_\psi \to \Pi \cup \cT \cup \{\top, \bot, \neg, \land, \lor\} \cup \bigcup_{k \in \Z_+} \{\Diamond_{\geq k}, \Box_{< k}, \DiamondG_{\geq k}, \BoxG_{< k}\}
\]
is a \textbf{labeling function} (over $(\Pi, \cT)$ that labels each node in $V_\psi$ with its main connective, or with a symbol in $\Pi \cup \cT \cup \{\bot, \top\}$ respectively.\footnote{More formally, $V_\psi$ consists of a set of sequences of subschemata defined as follows. \textbf{(1)} $(\psi) \in V_\psi$, 
\textbf{(2)} if $(\psi_1, \ldots, \psi_m) \in V_\psi$ and $\psi_m \colonequals \varphi_1 \land \varphi_2$, then $(\psi_1, \ldots, \psi_m, \varphi_1) \in V_\psi$ and $(\psi_1, \ldots, \psi_m, \varphi_2) \in V_\psi$, \textbf{(3)} if $(\psi_1, \ldots, \psi_m) \in V_\psi$ and $\psi_m \colonequals  \Diamond_{\geq \ell} \varphi$, then $(\psi_1, \ldots, \psi_m, \varphi) \in V_\psi$. The cases for $\neg$, $\lor$, $\Box_{< \ell}$, $\DiamondG_{\geq \ell}$ and $\BoxG_{< \ell}$ are defined analogously. Now, $E_\psi$ is defined as follows: if $\vec{\psi}\colonequals (\psi_1, \ldots, \psi_m) \in V_\psi$ and $\vec{\psi'} \colonequals (\psi_1, \ldots,\psi_{m}, \psi_{m+1}) \in  V_\psi$, then $(\vec{\psi}, \vec{\psi'}) \in E_\psi$, and there are no other edges. Now, for each $\vec{\psi} = (\psi_1, \ldots,\psi_m) \in V_\psi$, we define $\lambda_\psi(\vec{\psi})$ as follows. If $\psi_m$ is a formula of the form $\varphi \in \Pi \cup \cT \cup \{\bot, \top\}$, then $\lambda_\psi(\vec{\psi}) = \varphi$. On the other hand, if $\psi_m$ is a formula of the form $ \varphi_1 \land \varphi_2$, then $\lambda_\psi(\vec{\psi}) = \land$. The cases for $\neg$, $\lor$, diamonds and boxes are defined analogously to $\land$. However, any isomorphic forest with back edges to the syntax forest of a program also suffices.}

Let $\mathsf{base}$ and $\mathsf{iter}$ be new constant symbols.
Let $\cT$ be a finite set of schema variables and $\Pi$ a set of proposition symbols. Given a $(\Pi, \cT)$-program $\Lambda$, its \textbf{syntax forest} is a node-labeled forest $(V_\Lambda, E_\Lambda, B_{\Lambda}, \rho_\Lambda, \lambda_\Lambda)$ consisting of back edges $B_\Lambda$ and two labeling functions $\rho_\Lambda \colon V_\Lambda \to \cT \times \{\mathsf{base}, \mathsf{iter} \}$ and $\lambda_\Lambda$, which are defined as follows.
\begin{itemize}
    \item 
    For each head predicate $X$ in $\Lambda$, we let $\varphi_X$ and $\psi_X$ denote the bodies of its base rule and induction rule, respectively. Now, $(V_\Lambda, E_\Lambda, \lambda_\Lambda)$ is the disjoint union of node-labeled trees $\bigcup_{X \in \cT} (V_{\psi_X}, E_{\psi_X}, \lambda_{\psi_X}) \cup (V_{\varphi_X}, E_{\varphi_X}, \lambda_{\varphi_X})$. 
    \item
    For each $v \in V_{\varphi_X}$ and $u \in V_{\psi_X}$, we have $\rho_\Lambda(v) = (X, \mathsf{base})$ and $\rho_\Lambda(u) = (X, \mathsf{iter})$. 
    \item 
    $B_{\Lambda} \subseteq V_\Lambda \times V_\Lambda$ is the set of back edges defined as follows. For each $w \in V_\Lambda$ such that $\lambda_\Lambda(w) \in \cT$ and for each $u \in V_\Lambda$ that is the root of a tree and $\rho_\Lambda(u) \in  \{\lambda_\Lambda(w)\} \times \{\mathsf{base}, \mathsf{iter} \}$, we define $(w, u) \in B_\Lambda$. 
\end{itemize}
A partial syntax forest represents a subprogram of a full program. 
Now, let $S = \bigcup_{k \in \Z_+} \{\Diamond_{\geq k}, \Box_{< k}, \DiamondG_{\geq k}, \BoxG_{< k}\}$ and $S_m = \{\Diamond_{\geq m}, \Box_{< m}, \DiamondG_{\geq m}, \BoxG_{< m}\}$.
The size of a (partial) syntax forest $\cF = (V, E, B, \rho, \lambda)$ is defined by
\[
\abs{\cF} \colonequals \abs{V \setminus V_S} + \sum_{v \in V_S, \lambda(v) \in S_m} m,
\]
where $V_S = \{\, u \in V \mid \lambda(u) \in S\,\}$.  
It is not hard to see that if $\cF_\Lambda$ is the syntax forest of a program $\Lambda$, then $\abs{\Lambda} = \abs{\cF_\Lambda}$. In other words, the size of the (partial) syntax forest corresponds to the size of a (partial) program that it represents.

Let $\Lambda$ be a $\Pi$-program of $\GGMSC$ and let $\cF_\Lambda = (V_\Lambda, E_\Lambda, B_\Lambda, \rho_\Lambda, \lambda_\Lambda)$ be its syntax forest. Given a node $u \in V_\Lambda$, we define the \textbf{$u$-subformula} $\Lambda_u$ of $\Lambda$ as follows.
\begin{itemize}
    \item If $\lambda_\Lambda(u) \in \Pi \cup \cT \cup \{\top, \bot\}$, then $\Lambda_u = \lambda_\Lambda(u)$.
    \item If $\lambda_\Lambda(u) = * \in \{\lor, \land\}$ and $u_1, u_2$ are the successors of $u$, then $\Lambda_u = (\Lambda_{u_1} * \Lambda_{u_2})$.
    \item If $\lambda_\Lambda(u) = \star \in \{\neg \} \cup \bigcup_{k \in \Z_+} \{\Diamond_{\geq k}, \Box_{< k}, \DiamondG_{\geq k}, \BoxG_{< k}\}$ and $u'$ is the successor of $u$, then $\Lambda_u = \star \Lambda_{u'}$.
\end{itemize}
Since $\Lambda_u$ is a schema, $\Lambda^i_u$ denotes the $\GGML$-formula obtained from $\Lambda_u$ by replacing each variable $X$ by $X^i$.

\begin{example}\label{example: syntax forest}
In Figure \ref{fig:syntax forest}, we illustrate the syntax forest of a program.
\begin{figure}[ht]
\centering
\caption{Below, on the left we have a program $\Lambda$ with two base rules (in red) and two induction rules (in blue) and on the right we have its syntax forest $\cF_\Lambda$. The two blue trees correspond to the bodies of the induction rules while the red ones correspond to the bodies of the base rules. Back edges are drawn as dotted edges. The size of the program and its syntax forest is $15$. If $v$ is the node labeled by $\Diamond_{\geq 2}$ in $\cF_\Lambda$, then $\Lambda_v = \Diamond_{\geq 2} X$.}
\label{fig:syntax forest}
\begin{minipage}[h]{0.45\textwidth}
\begin{tcolorbox}[colback=yellow!5!white]
Program:
\begin{align*}
  &\tre{X (0) \colonminus \neg p} &&\tb{X \colonminus Y \land \Diamond_{\geq 2} X} \\
  &\tre{Y (0) \colonminus r \lor q} && \tb{Y \colonminus \Box_{< 3} \neg Y}.
\end{align*}
\end{tcolorbox}
\end{minipage}
\hspace{0.5em}
\begin{minipage}[h]{0.45\textwidth}
\begin{tikzpicture}[scale=0.6, every node/.style={scale=0.6}, nodes={draw, circle}, <-]

\node[fill=blue!20] (X1) {$\land$};
\node[below left of=X1, node distance=2.8cm, fill=blue!20, double, thick](X2) {$Y$};
\node[below of=X1, node distance=1.5cm, fill=blue!20](X3) {$\Diamond_{\geq 2}$};
\node[below of=X3, node distance=1.5cm, fill=blue!20, double, thick](X4) {$X$};

\node[right of =X1, node distance=2.5cm, fill=red!20](X01) {$\neg$};
\node[below of =X01, node distance=2cm, fill=red!20](X02) {$p$};

\node[left of =X1, node distance=4cm, fill=blue!20] (Y1) {$\Box_{< 3}$};
\node[below of =Y1, node distance=1.5cm, fill=blue!20, thick](Y2) {$\neg$};
\node[below of =Y2, node distance=1.5cm, fill=blue!20, double, thick](Y3) {$Y$};

\node[left of =Y1, node distance=2.5cm, fill=red!20] (Y01) {$\lor$};
\node[below left of =Y01, node distance=2.8cm, fill=red!20](Y02) {$r$};
\node[below of =Y01, node distance=2cm, fill=red!20](Y03) {$q$};

\path [-stealth, very thick]
(X1) edge node [draw=none] {} (X2)
(X1) edge node [draw=none] {} (X3)
(X3) edge node [draw=none] {} (X4)
(Y1) edge node [draw=none] {} (Y2)
(Y2) edge node [draw=none] {} (Y3)
(X01) edge node [draw=none] {} (X02)
(Y01) edge node [draw=none] {} (Y02)
(Y01) edge node [draw=none] {} (Y03);

\path [-stealth, thick]
(X4) edge [dotted, bend right=20] node [draw=none] {} (X01)
(Y3) edge [dotted] node [draw=none] {} (Y01)
(X2) edge [dotted, bend right=80] node [draw=none] {} (Y01);

\path [-stealth, thick]
(X4) edge [bend right=57, dotted] node [draw=none] {} (X1)
(X2) edge [dotted] node [draw=none] {} (Y1)
(Y3) edge [bend left=44, dotted] node [draw=none] {} (Y1);
\end{tikzpicture}   
\end{minipage}
\end{figure}
\end{example}

\subsection{Clocked models and syntactic sugar}

In this section, we define key notions and notations related to clocked models, which intuitively are models associated with information about how long they can be ``iterated''. We also define some syntactic sugar.

Given a pointed $\Pi$-model $(M, w)$ and an $\ell \in \N$, the $\ell$\textbf{-clocked $\Pi$-model of $(M, w)$} is a tuple $(M, w, \ell)$. We simply say that a Kripke model is clocked if it is an $\ell$-clocked model for some $\ell \in \N$.  
Intuitively, $\ell$ tells how many times $(M, w)$ can be ``iterated''. 
Let $\bbA$ be a class of pointed Kripke models.
We let $\mathrm{CM}_\ell(\bbA) \colonequals \{\, (A, w, \ell) \mid (A, w) \in \bbA \,\}$ denote the class of $\ell$-clocked models obtained from $\bbA$.
Furthermore, we let $\CM_*(\bbA)$ denote the class of all $\ell$-clocked models obtained from $\bbA$ for every $\ell \in \N$.

We now define some syntactic sugar. 
Let $\cA \subseteq \CM_*(\bbA)$ be a class of clocked models. The class of \textbf{iterated models obtained from} $\cA$ is 
\[
    \mathrm{iter}(\cA) = \{\, (M, w, \ell-1) \mid (M, w, \ell) \in \cA,\, \ell \geq 1 \,\},
\]
and the class of \textbf{initialized models obtained from} $\cA$ is
\[
    \mathrm{init}(\cA) = \{\, (M, w, 0) \mid (M, w, \ell) \in \cA, \ell = 0 \,\}.
\]
Moreover, the class of \textbf{successor models obtained from} $\cA$ is
\[
\Box \cA \colonequals \{\, ((W, R, V), w', \ell) \mid ((W, R, V), w, \ell) \in  \cA, (w, w') \in R \,\}. 
\]
Similarly, the class of \textbf{pointed models obtained from} $\cA$ is
\[
\BoxG \cA \colonequals \{\, ((W, R, V), w', \ell) \mid ((W, R, V), w, \ell) \in  \cA, w' \in W \,\}.
\]

An \textbf{$m$-successor function over} $\cA$ is a function $f \colon \cA \to \wp(\Box \cA)$ such that for every $(A,w, \ell) \in \cA$, we have $f(A, w, \ell) \subseteq \Box \{(A, w, \ell)\}$ and $\abs{f(A, w, \ell)} = m$. Intuitively, an $m$-successor function assigns $m$ successor models to each clocked model. Moreover, we let $\Diamond_f \cA = \bigcup_{(A, w, \ell) \in \cA} f(A, w, \ell)$. Analogously, an \textbf{$m$-global function} over $\cA$ is a function $g \colon \cA \to \wp(\BoxG \cA)$ such that for every $(A, w, \ell) \in \cA$, we have $g(A, w, \ell) \subseteq \BoxG \{(A, w, \ell)\}$ and $\abs{g(A, w, \ell)} = m$. Moreover, we also define $\DiamondG_g \cA = \bigcup_{(A, w, \ell) \in \cA} g(A, w, \ell)$. These definitions generalize to partial $m$-successor functions $f \colon \cA \rightharpoonup \wp(\Box \cA)$ and partial $m$-global functions $g \colon \cA \rightharpoonup \wp(\BoxG \cA)$ in a natural way.

\subsection{Definition of the game and its applications}

In this section, we define the formula size game for $\GGMSC$, played by Samson and Delilah. 
We begin with an informal idea of the game and its main results.
Given a $k \in \N$ and two classes, $\bbA$ and $\bbB$, of pointed $\Pi$-models, in the game $\mathrm{FS}^\Pi_k(\bbA, \bbB)$, Samson tries to construct the syntax forest of a $\Pi$-program of $\GGMSC$ of size at most $k$ that separates $\bbA$ from $\bbB$. Intuitively, the game verifies all possible semantic games at once over the constructed program and the given models.
Samson's strategy is uniform in the game if he always constructs the same program regardless of what Delilah does.
Theorem \ref{thrm: uniform fs-game characterization} shows that the game characterizes the logical equivalence of classes of pointed models of a given program size, but requires uniform winning strategies for Samson. 
However, Theorem \ref{thrm: finite uniform non-uniform characterization} shows that, over finite models, uniform winning strategies for Samson are not necessary.

For the rest of this section, we fix an arbitrary set $\Pi$ of proposition symbols, a $k \in \N$ and two classes,  $\bbA$ and $\bbB$, of pointed $\Pi$-models.

Now, we formally define the formula size game $\mathrm{FS}^{\Pi}_k(\bbA, \bbB)$ as follows. 
A \textbf{position} of the game is a tuple $(\cF, U, \mathrm{left}, \mathrm{right})$, where 
\begin{itemize}
    \item 
    $\cF = (V, E, B, \rho, \lambda)$ is a finite non-empty forest associated with back edges $B$, a partial function
    \[
    \lambda \colon V \rightharpoonup \Pi \cup \mathrm{VAR} \cup \{\top, \bot, \neg, \land, \lor\} \cup \bigcup_{m \in \Z_+} \{\Diamond_{\geq m}, \Box_{< m}, \DiamondG_{\geq m}, \BoxG_{< m}\}
    \] 
    and a function $\rho \colon V \to  \mathrm{VAR} \times \{\mathsf{base}, \mathsf{iter} \}$,
    \item 
    $U \subseteq V$ is a set of nodes of $\cF$, which are, intuitively, the currently active nodes,
    \item 
    $\mathrm{left} \colon V \to \wp(\CM_*(\bbA)) \cup \wp(\CM_*(\bbB))$ and $\mathrm{right} \colon V \to \wp(\CM_*(\bbA)) \cup \wp(\CM_*(\bbB))$ are functions that both assign to each node a \emph{finite} set of clocked models.
\end{itemize}
Moreover, the number of \textbf{resources} used in the position is $\abs{\cF}$.
The set of \textbf{initial positions} of the game consists of positions of the form 
$((U_0, \emptyset, \emptyset, \rho_0, \emptyset), U_0, \leftf_0, \rightf_0)$, where $U_0$ is a finite non-empty set of nodes.

At the start of every play of the game, Samson gives a set of nodes $U_0$ and a function $\rho_0 \colon U_0 \to \mathrm{VAR}\times \{\mathsf{base}, \mathsf{iter}\}$ such that if for some $v \in U_0$, we have $\rho_0(v) = (X, s)$, then there is a node $u \in U_0$ such that $\rho_0(u) = (X, s')$, where $s' \in \{\mathsf{base}, \mathsf{iter}\} \setminus \{s\}$.
Then Delilah chooses a finite subset $\bbA' \subseteq \bbA$ and for every pointed model $(A, w) \in \bbA'$, Samson gives an integer $\ell_{(A,w)} \in \N$, which forms a set $\cA = \{\, (A, w, \ell_{(A,w)}) \mid (A, w) \in \bbA' \,\}$ of clocked models. 
After that, Delilah chooses a finite subset $\cB \subseteq \CM_*(\bbB)$.
Then Samson gives a function $\leftf_0 \colon U_0 \to \wp(\cA)$ such that $\cA = \bigcup_{v \in U_0}\leftf_0(v)$. Lastly, $\cB$ and $\rho_0$ induce a function $\rightf_0 \colon U_0 \to \wp(\cB)$ such that $\rightf_0(v) = \cB $ if $\rho_0(v) = (X, \mathsf{iter})$, and $\rightf_0(v) = \mathrm{init}(\cB)$ if $\rho_0(v) = (X, \mathsf{base})$.
Then the initial position of the play is 
$((U_0, \emptyset, \emptyset, \rho_0, \emptyset), U_0, \leftf_0, \rightf_0)$.

The \textbf{rules} of the game are defined as follows.
Assume that the position in a play of the game is $P = (\cF, U, \mathrm{left}, \mathrm{right})$ with     $\abs{\cF} = r$, and
    %
    %
    $\cF = (V, E, B, \rho, \lambda)$. 
    %
    %

The position that follows $P$ is denoted below by $P' = (\cF', U', \mathrm{left}', \mathrm{right}')$, and $\cF'$ is denoted by $(V', E', B', \rho', \lambda')$.
Samson loses if $\abs{\cF} > k$, or for some $v \in V$, there is $(M,w,\ell) \in \leftf(v)$ with $\ell >0$ and $\rho(v) = (X, \mathsf{base})$ for some $X$, or he cannot make the choices required by the move. Moreover, if in the definition of a move below we do not explicitly define a component of $P'$, then the component is the same as in $P$.

Before Samson decides which move is played, Delilah chooses a node $v$ from $U$ and the move is played at $v$. We let 
    %
    %
    $\leftf(v) = \cL$ and $\rightf(v) = \cR$.
    %
    %
If $v \not\in \dom(\lambda)$, then Samson has the option to make one of the following moves. 
\begin{itemize}
    \item \textbf{$\neg\,$-move}: Let $v'$ be a fresh node not in the domain $V$. Now the position $P'$ that follows $P$ is intuitively obtained by swapping the sets $\cL$ and $\cR$. Formally, $P'$ is defined as follows: $V' = V \cup \{v'\}$, 
    $E' = E \cup \{(v, v')\}$, $U' = (U \setminus \{v\}) \cup \{v'\}$, 
    $\rho' = \rho[\rho(v)/v']$,
    $\lambda' = \lambda[\neg/v]$, 
    $\mathrm{left}' = \mathrm{left}[\cR/v', \emptyset/v]$, and $\mathrm{right}' = \mathrm{right}[\cL/v', \emptyset/v]$. 
    \item \textbf{$\lor$-move}: Intuitively, Samson splits the set $\cL$.
    Formally, Samson gives two sets $\cL_1, \cL_2 \subseteq \cL$ such that $\cL_1 \cup \cL_2 = \cL$. 
    Let $v_1$ and $v_2$ be fresh nodes not in the domain $V$. Now the position $P'$ following $P$ is defined as follows: $V' = V \cup \{v_1, v_2\}$, 
    $E' = E \cup \{(v, v_1), (v, v_2)\}$, $U' = (U \setminus \{v\}) \cup \{v_1, v_2\}$, 
    $\rho' = \rho[\rho(v)/v_1, \rho(v)/v_2]$,
    $\lambda' = \lambda[\lor/v]$, 
    $\mathrm{left}' = \mathrm{left}[\cL_1/v_1, \cL_2/v_2, \emptyset/v]$, 
    $\mathrm{right}' = \mathrm{right}[\cR/v_1, \cR/v_2, \emptyset/v]$. %
    \item \textbf{$\land$-move}: The move is identical to the $\lor$-move except that the node $v$ is labeled with $\land$ and the roles of $\cL$ and $\cR$ are switched as follows: Samson gives two sets $\cR_1, \cR_2 \subseteq \cR$ such that $\cR_1 \cup \cR_2 = \cR$. Then we define $\mathrm{left}' = \mathrm{left}[\cL/v_1, \cL/v_2, \emptyset/v]$ and $\mathrm{right}' = \mathrm{right}[\cR_1/v_1, \cR_2/v_2, \emptyset/v]$.
    \item \textbf{$\Diamond_{\geq m}$-move}: Intuitively, Samson provides $m$-successor models for every model in $\cL$, which Delilah can challenge by selecting a finite subset, while Delilah provides $m$-successor models only for a chosen subset of models in $\cR$, which Samson must then respond to by selecting finite, non-empty successor sets. 
    
    Formally, Samson gives an $m$-successor function $f$ for $\cL$ and Delilah gives a \emph{partial} $m$-successor function $g$ for $\cR$. 
    Then Delilah chooses a finite subset $\cL' \subseteq \Diamond_{f} \cL$. For every $(N, u, \ell) \in \dom(g)$, Samson chooses a \emph{non-empty} finite subset $\cR_{(N, u, \ell)} \subseteq \Diamond_{g} \{(N, u, \ell)\}$, and we set $\cR' = \bigcup_{(N, u, \ell)\in \dom(g)} \cR_{(N, u, \ell)}$.

    Let $v'$ be a fresh node not in $V$. The position $P'$ that follows $P$ is defined as follows:
    $U' = (U \setminus \{v\}) \cup \{v'\}$,
    $V' = V \cup \{v'\}$, 
    $E' = E \cup \{(v, v')\}$, 
    $\rho' = \rho[\rho(v)/v']$,
    $\lambda' = \lambda[\Diamond_{\geq m}/v]$,
    $\mathrm{left}' = \mathrm{left}[\cL' /v', \emptyset/v]$, $\mathrm{right}' = \mathrm{right}[\cR'/v', \emptyset/v]$. 
    \item \textbf{$\Box_{< m}$-move}: The move is identical to the $\Diamond_{\geq m}$-move except that $v$ is labeled with $\Box_{< m}$ and the roles of $\cL$ and $\cR$ are switched as follows. Samson gives an $m$-successor function $f$ for $\cR$ and Delilah gives a \emph{partial} $m$-successor function $g$ for $\cL$. 
    Then Delilah chooses a finite subset $\cR' \subseteq \Diamond_{f} \cR$. For every $(M, w, \ell) \in \dom(g)$, Samson chooses a \emph{non-empty} finite subset $\cL_{(M, w, \ell)} \subseteq \Diamond_{g} \{(M, w, \ell)\}$, and we set $\cL' = \bigcup_{(M, w, \ell)\in \dom(g)} \cL_{(M, w, \ell)}$. Then we define $\mathrm{left}' = \mathrm{left}[\cL' /v', \emptyset/v]$ and $\mathrm{right}' = \mathrm{right}[\cR'/v', \emptyset/v]$.
    \item \textbf{$\DiamondG_{\geq m}$-move}: Identical to the $\Diamond_{\geq m}$-move except that $v$ is labeled with $\DiamondG_{\geq m}$, Samson gives an $m$-global function $f$ over $\cL$ instead of an $m$-successor function over $\cL$ and Delilah gives a partial $m$-global function $h$ over $\cR$. 
    \item \textbf{$\BoxG_{< m}$-move}: Identical to the $\Box_{< m}$-move except that $v$ is labeled with $\BoxG_{< m}$, Samson gives an $m$-global function $f$ over $\cR$ instead of an $m$-successor function over $\cR$ and Delilah gives a partial $m$-global function $h$ over $\cL$.
    \item \textbf{Sig-move}: Samson chooses a symbol $\varphi \in \Pi \cup \{\top, \bot\}$.
    The position $P'$ following $P$ is defined as follows: $\lambda' = \lambda[\varphi/v]$.
    If $\varphi$ separates $\cL$ from $\cR$, then Samson wins.
    Otherwise, Delilah wins. 
    \item \textbf{$X$-move}: 
    Intuitively, Samson picks a variable $X$, and Delilah decides whether to challenge. If she challenges, a fresh base rule for $X$ is constructed, considering only models with zero clocks. If not, the game continues (or begins) constructing the induction rule for $X$, clocks are updated, and models with zero clocks are stored in the current node. 
    
    Formally, Samson chooses a variable $X \in \mathrm{VAR}$. 
    If $\rho(v) = (Y,\mathsf{base})$ for some $Y \in \mathrm{VAR}$, 
    then Samson loses.
    Delilah can choose to \textbf{challenge} Samson.

    First assume that Delilah does not challenge Samson. If $\mathrm{iter}(\cL) = \mathrm{iter}(\cR) = \emptyset$, then Samson wins.
    Assume that there is a root $v' \in V$ such that $\rho(v') = (X,\mathsf{iter})$. If such a root does not exist, then we let $v'$ denote a fresh node.
    The position $P'$ is defined as follows: 
    \begin{itemize}
        \item 
        $U' = (U \setminus \{v\}) \cup \{v'\}$,
        \item 
        $V' = V\cup \{v'\}$, $B' = B \cup \{(v, v')\}$,
        \item 
        $\rho' = \rho[(X, \mathsf{iter})/v']$,
        \item 
        $\lambda' = \lambda[X/v]$,
        \item 
        $\leftf' = \leftf[\mathrm{iter}(\cL) \cup \leftf(v') /v', \mathrm{init}(\cL)/v]$ (omitting $\leftf(v')$ if $v'$ is a fresh node), 
        \item 
        $\rightf' = \rightf[\mathrm{iter}(\cR) \cup \rightf(v') /v', \mathrm{init}(\cR)/v]$ (omitting $\rightf(v')$ if $v'$ is a fresh node). 
    \end{itemize}

    Assume that Delilah challenges Samson. 
    Assume that there is a root $v' \in V$ such that  $\rho(v') = (X,\mathsf{base})$. If such a root does not exist, then we let $v'$ denote a fresh node. 
    The position $P'$ is defined as follows: 
    \begin{itemize}
        \item $U' = (U \setminus \{v\}) \cup \{v'\}$,
        \item $V' = V \cup \{v'\}$, $B' = B \cup \{(v, v')\}$,
        \item $\rho' = \rho[(X, \mathsf{base})/v']$,
        \item $\lambda' = \lambda[X/v]$,
        \item $\leftf' = \leftf[\mathrm{init}(\cL)\cup\leftf(v') /v', \emptyset/v]$ (omitting $\leftf(v')$ if $v'$ is a fresh node),
        \item $\rightf' = \rightf[\mathrm{init}(\cR)\cup\rightf(v')/v', \emptyset/v]$ (omitting $\rightf(v')$ if $v'$ is a fresh node).
    \end{itemize}
\end{itemize}

If $v \in \dom(\lambda)$, then Samson has to play a move according to the label given by $\lambda(v)$.  
\begin{itemize}
    \item $\lambda(v) = \neg$: Let $v'$ be the successor of $v$. The position $P'$ is defined as follows: $U' = (U \setminus \{v\}) \cup \{v'\}$, 
    $\mathrm{left}' = \mathrm{left}[\cR \cup \mathrm{left}(v')/v', \emptyset/v]$, and $\mathrm{right}' = \mathrm{right}[\cL \cup \mathrm{right}(v')/v', \emptyset/v]$.
    \item $\lambda(v) = \lor$: 
    Samson gives two sets $\cL_1, \cL_2 \subseteq \cL$ such that $\cL_1 \cup \cL_2 = \cL$. 
    Let $v_1$ and $v_2$ be the successors of $v$. Now the position $P'$ that follows $P$ is defined as follows: $U' = (U \setminus \{v\}) \cup \{v_1, v_2\}$,  
    \begin{itemize}
        \item $\mathrm{left}' = \mathrm{left}[\cL_1 \cup \leftf(v_1)/v_1, \cL_2 \cup \leftf(v_2)/v_2, \emptyset/v]$ and
        \item $\mathrm{right}' = \mathrm{right}[\cR \cup \rightf(v_1)/v_1, \cR \cup \rightf(v_2)/v_2, \emptyset/v]$. 
    \end{itemize}
    \item $\lambda(v) = \land$:
    The case is identical to $\lambda(v) = \lor$ except that the roles of $\cL$ and $\cR$ are switched analogously to the unlabeled case.
    \item $\lambda(v) = \Diamond_{\geq m}$:
    Samson gives an $m$-successor function $f$ for $\cL$ and Delilah gives a \emph{partial} $m$-successor function $h$ for $\cR$. 
    Then Delilah chooses a finite subset $\cL' \subseteq \Diamond_{f} \cL$. For every $(N, u, \ell) \in \dom(h)$, Samson chooses a \emph{non-empty} finite subset $\cR_{(N, u, \ell)} \subseteq \Diamond_{h} \{(N, u, \ell)\}$, and we set $\cR' = \bigcup_{(N, u, \ell)\in \dom(h)} \cR_{(N, u, \ell)}$.  
    
    Let $v'$ be the successor of $v$. The position $P'$ that follows $P$ is defined as follows: 
    \begin{itemize}
        \item $U' = (U \setminus \{v\}) \cup \{v'\}$,
        \item $\mathrm{left}' = \mathrm{left}[\cL' \cup \leftf(v') /v', \emptyset/v]$ and 
        \item $\mathrm{right}' = \mathrm{right}[\cR' \cup \rightf(v') /v', \emptyset/v]$. 
    \end{itemize}
    \item  $\lambda(v) = \Box_{< m}$: The case is identical to the case $\lambda(v) = \Diamond_{\geq m}$ except that the roles of $\cL$ and $\cR$ are switched analogously to the unlabeled case.
    \item The cases $\lambda(v) = \DiamondG_{\geq m}$ and $\lambda(v) = \BoxG_{<m}$ are obtained analogously from the unlabeled cases.
    \item $\lambda(v) \in \mathrm{VAR}$: 
    This is similar to the unlabeled $X$-move.
\end{itemize}
Note that Samson might not be able to perform a move that he chose, e.g., a $\Diamond_{\geq m}$-move if there is a model in $\cL$ that does not have $m$ successor models.

Now, we have defined the formula size game and can start to study its properties. 
We first present a proposition which shows that every play of the formula size game ends in a finite number of steps.
\begin{proposition}\label{prop: finite game tree}
    Every play of the game $\mathrm{FS}^\Pi_k(\bbA, \bbB)$ is finite.
\end{proposition}
\begin{proof}
    Suppose, for the sake of contradiction, that some play of the game continues indefinitely. Thus, during the play of the game there must be a variable $X$ that has been iterated indefinitely. Therefore, there is a position $(\cF, U, \leftf, \rightf)$ of the play where an $X$-move is played at a node $v \in U$ such that $\mathrm{iter}(\leftf(v)) = \mathrm{iter}(\rightf(v)) = \emptyset$, in which case Samson wins.
\end{proof}

Next, we define the notion of uniform strategies. Informally, Samson's strategy is uniform if he has a program of $\GGMSC$ in his mind and during the game he constructs the syntax forest of that program, regardless of how Delilah plays.

Formally, let $\Lambda$ be a $\Pi$-program of $\GGMSC$ and let $\cF_\Lambda = (V_\Lambda, E_\Lambda, B_\Lambda, \rho_\Lambda, \lambda_\Lambda)$ be its syntax forest.
Now, let 
$P = (\cF, U, \mathrm{left}, \mathrm{right})$ be a position of a play of $\mathrm{FS}^\Pi_k(\bbA, \bbB)$, where $\cF = (V, E, B, \rho, \lambda)$. A function $f \colon V \to V_\Lambda$ is a \textbf{position embedding (w.r.t. $\Lambda$ and $P$)} if it satisfies the following conditions.
\begin{itemize}
    \item If $u$ is the root of a tree in $\cF$, then $f(u)$ is the root of a tree in $\cF_\Lambda$.
    \item $f$ is an embedding, i.e., satisfies the following properties.
    \begin{itemize}
        \item $f$ is an injection.
        \item For every $u, u' \in V$, $(u, u') \in S$ iff $(f(u), f(u')) \in S_\Lambda$, where $S \in \{E, B\}$.
        \item For every $u \in \dom(\rho)$, $\rho(u) = \rho_\Lambda(f(u))$.
        \item For every $u \in \dom(\lambda)$, $\lambda(u) = \lambda_\Lambda(f(u))$.
    \end{itemize} 
\end{itemize}
Note that, by the definition of a position embedding, $\abs{\cF} \leq \abs{\cF_\Lambda}$.
Let $\sigma$ be a strategy for Samson in $\mathrm{FS}^\Pi_k(\bbA, \bbB)$.
We say that $\sigma$ is \textbf{uniform} (w.r.t. a $\Pi$-program $\Lambda$ of $\GGMSC$ of size at most $k$) if, in every play consistent with $\sigma$, there is a position embedding w.r.t. $\Lambda$ in every position of the play, and the initial position is of the form 
\[
((U_0, \emptyset, \emptyset, \rho_0, \emptyset), U_0, \leftf_0, \rightf_0), 
\]
where $U_0$ consists of two nodes $u, u'$ for each accepting predicate $Y$ of $\Lambda$ such that $\rho_0(u) = (Y, \mathsf{base})$ and $\rho_0(u') = (Y, \mathsf{iter})$.

Now, we shall prove that the formula size game (w.r.t. uniform strategies) characterizes the logical equivalence of $\GGMSC$-programs of a given program size.
\begin{theorem}\label{thrm: uniform fs-game characterization}
    The following claims are equivalent.
    \begin{enumerate}
        \item Samson has a uniform winning strategy in $\mathrm{FS}^\Pi_k(\bbA, \bbB)$. 
        \item There is a $\Pi$-program of $\GGMSC$ of size at most $k$ that separates $\bbA$ from $\bbB$.
    \end{enumerate}
\end{theorem}
\begin{proof}
    ``$1 \Rightarrow 2$''
    Assume that Samson has a uniform winning strategy $\sigma$ in $\mathrm{FS}^\Pi_k(\bbA, \bbB)$ w.r.t. a $\Pi$-program $\Lambda$ of $\GGMSC$. 
    We let $\cF_\Lambda = (V_\Lambda, E_\Lambda, B_\Lambda, \rho_\Lambda, \lambda_\Lambda)$ denote the syntax forest of $\Lambda$.
    By Proposition \ref{prop: finite game tree}, every play of the game is finite and thus the game tree is well-founded.
    Therefore, we may prove by induction on the positions of the game tree induced by $\sigma$ (starting from the leaves) that the following condition holds in every position $P = (\cF, U, \leftf, \rightf)$ and at every node $v \in U$.
    \begin{equation}\tag{$*$}\label{induction_hypo_1}
    \begin{split}
    M, w \models \Lambda_{g(v)}^\ell &\text{ for each } (M, w, \ell) \in \cL, \\
    N, u \not\models \Lambda_{g(v)}^\ell &\text{ for each } (N, u, \ell ) \in \cR, \\
    \end{split}
    \end{equation}
    where $\leftf(v) = \cL$, $\rightf(v) = \cR$ and $g$ is a position embedding w.r.t. $\Lambda$ and $P$.

    First assume that $v \notin \dom(\lambda)$.
    \begin{itemize}
        \item $\lambda_\Lambda(g(v)) = \varphi \in \Pi \cup \{\top, \bot\}$: Since the game tree is induced by $\sigma$, the next move of Samson is the Sig-move choosing the symbol $\varphi$. Since $\sigma$ is a winning strategy, $\varphi$ separates $\cL$ from $\cR$. 
        Thus, Condition (\ref{induction_hypo_1}) holds in the position $P$.
        \item $\lambda_\Lambda(g(v)) = \neg$: Let $s'$ be the successor of $g(v)$. By the induction hypothesis, Condition~(\ref{induction_hypo_1}) holds in the positions following $P$. Therefore, for each $(M, w, \ell) \in \cL$, $M, w \not\models \Lambda^\ell_{s'}$ and for each $(N, u, \ell) \in \cR$, $N, u \models \Lambda^\ell_{s'}$. Since $\Lambda^\ell_{g(v)} = \neg \Lambda^\ell_{s'}$, Condition (\ref{induction_hypo_1}) also holds for the node $v$ in the position $P$. 
        \item $\lambda_\Lambda(g(v)) = \lor$: Let $s_1$ and $s_2$ be the successors of $g(v)$. Let $\cL_1, \cL_2 \subseteq \cL$ be the selections of Samson according to $\sigma$. 
        By the induction hypothesis, Condition (\ref{induction_hypo_1}) holds in the positions following $P$ w.r.t. the selections of Samson.
        Therefore, for every $i \in \{1, 2\}$ and for each $(M, w, \ell) \in \cL_i$, $M, w \models \Lambda^\ell_{s_i}$. Also, for each $(N, u, \ell) \in \cR$, $N, u \not\models \Lambda^\ell_{s_i}$. 
        Since $\cL_1 \cup \cL_2 = \cL$ and $\Lambda_{g(v)}^\ell = \Lambda_{s_1}^\ell \lor \Lambda_{s_2}^\ell$, for each $(M, w, \ell) \in \cL$ and $(N, u, \ell) \in \cR$, we have $M, w \models \Lambda^\ell_{g(v)}$ and $N, u \not\models \Lambda^\ell_{g(v)}$. Thus, Condition (\ref{induction_hypo_1}) holds for the node $v$ in the position $P$.
        \item $\lambda_\Lambda(g(v)) = \Diamond_{\geq m}$: 
        Let $s'$ be the successor of $g(v)$. Let $f$ be the $m$-successor function for $\cL$ selected by Samson according to $\sigma$.
        Let $h$ be a partial $m$-successor function for $\cR$ given by Delilah. 
        By the induction hypothesis, Condition (\ref{induction_hypo_1}) holds in every position following $P$: no matter which subset $\cL' \subseteq \Diamond_{f} \cL$ Delilah chooses, and 
        for each $(N, u, \ell) \in \dom(h)$, Samson can choose a non-empty finite subset $\cR_{(N, u, \ell)} \subseteq \Diamond_{h} \{(N, u, \ell)\}$ according to $\sigma$. 
        Since $f$ is an $m$-successor function, for each $(M, w, \ell) \in \cL$ there are at least $m$ models $(M, w', \ell) \in \Diamond_f\cL$ such that $w'$ is a successor of $w$ and $M, w' \models \Lambda^\ell_{s'}$. Thus, $M, w \models \Diamond_{\geq m} \Lambda^\ell_{s'}$, i.e., $M, w \models \Lambda^\ell_{g(v)}$. Similarly, we can show that for each $(N, u, \ell) \in \cR$ it holds that $N, u \not\models \Lambda^\ell_{g(v)}$. Therefore, Condition (\ref{induction_hypo_1}) holds for the node $v$ in the position $P$.
        \item $\lambda_\Lambda(g(v)) = X \in \mathrm{VAR}$: 
        There are two cases: either Delilah challenges Samson or not.
        First assume that Delilah challenges Samson. Let $s'$ be the successor of $g(v)$ over back edges $B_\Lambda$ such that $\rho_\Lambda(s') = (X, \mathsf{base})$. By the induction hypothesis, Condition (\ref{induction_hypo_1}) holds in the position following $P$, i.e., 
        for every $(M, w, 0) \in \mathrm{init}(\cL) $, $M, w \models \Lambda^0_{s'}$ and for every $(N, u, 0) \in \mathrm{init}(\cR)$, 
        $N, u \not\models \Lambda^0_{s'}$. Therefore, for every $(M, w, 0) \in \cL$ and $(N, u, 0) \in \cR$, $M, w \models \Lambda^0_{g(v)}$ and $N, u \not\models \Lambda^0_{g(v)}$. 

        Assume that Delilah does not challenge Samson. Let $s''$ be the successor of $g(v)$ over back edges $B_\Lambda$ such that $\rho_\Lambda(s'') = (X, \mathsf{iter})$. By the induction hypothesis, Condition (\ref{induction_hypo_1}) holds in the position following $P$, i.e., for every $\ell \geq 1$, and for every $(M, w, \ell-1) \in \mathrm{iter}(\cL)$, $M, w \models \Lambda^{\ell-1}_{s''}$ and for every $(N, u, \ell-1) \in \mathrm{iter}(\cR)$, $N, u \not\models \Lambda^{\ell-1}_{s''}$. Therefore, for every $(M, w, \ell) \in \cL$ and $(N, u, \ell) \in \cR$, we have $M, w \models \Lambda^\ell_{g(v)}$ and $N, u \not\models \Lambda^\ell_{g(v)}$. 
        
        Therefore, Condition (\ref{induction_hypo_1}) holds for the node $v$ in the position $P$.
    \end{itemize}
    The cases where $\lambda_\Lambda(g(v))$ is $\land$, $\Box_{< m}$, $\DiamondG_{\geq m}$ or $\BoxG_{< m}$ are handled analogously.
    Moreover, the cases where $v \in \dom(\lambda)$ are proved similarly to unlabeled moves.

    Now, consider the beginning of the game, Samson first gives a set of nodes $U_0$ and a function $\rho_0 \colon U_0 \to \mathrm{VAR} \times \{\mathsf{base}, \mathsf{iter} \}$, where $U_0$ consists of two nodes $u$ and $u'$ for each accepting predicate $Y$ of $\Lambda$ such that $\rho_0(u) = (Y, \mathsf{base})$ and $\rho_0(u') = (Y, \mathsf{iter})$. Then Delilah chooses a finite subset $\bbA' \subseteq \bbA$. 
    By the induction hypothesis, and since $\sigma$ is a uniform winning strategy, Samson can choose an integer $\ell_{(A,w)} \in \N$ for every $(A, w) \in \bbA'$, and a suitable function $\leftf_0$ such that Samson wins, no matter which finite subset $\cB \subseteq \CM_*(\bbB)$ Delilah chooses.
    That is, for each $\bbA' \subseteq \bbA$, every $(A, w) \in \bbA'$ is accepted (in round $0$ or in round $\ell_{(A,w)} + 1$) by an accepting predicate $Y$ of $\Lambda$, and for each $\bbB' \subseteq \bbB$, for every $(B, w) \in \bbB'$, there is no round $\ell_{(B,w)} \in \N$ such that $\Lambda$ accepts $(B, w)$.
    Therefore, $\Lambda$ separates $\bbA$ from $\bbB$.

    ``$2 \Rightarrow 1$''
    Let $\Lambda$ be a $\Pi$-program of $\GGMSC$ of size at most $k$ that separates $\bbA$ from $\bbB$, and let $\cF_\Lambda = (V_\Lambda, E_\Lambda, B_\Lambda, \rho_\Lambda, \lambda_\Lambda)$ be its syntax forest.
    
    We build a uniform winning strategy w.r.t. $\Lambda$ for Samson and show by induction that in every position $P = (\cF, U, \leftf, \rightf)$ and for every $v \in U$ the following holds.
    \begin{equation}\tag{$**$}\label{induction_hypo_2}
    \begin{split}
    M, w \models \Lambda_{g(v)}^\ell &\text{ for each } (M, w, \ell) \in \cL, \\
    N, u \not\models \Lambda_{g(v)}^\ell &\text{ for each } (N, u, \ell) \in \cR, \\
    \end{split}
    \end{equation}
    where $\leftf(v) = \cL$, $\rightf(v) = \cR$ and $g$ is a position embedding w.r.t. $\Lambda$. 
    In the proof we let $P' = (\cF', U', \leftf', \rightf' )$ denote the position following $P$ and we let $g'$ denote a position embedding in the position $P'$.

    In the starting position, Delilah begins by choosing a subset $\bbA' \subseteq \bbA$. 
    Then Samson chooses $\ell_{(A,w)} \colonequals \max\{r_{(A,w)} - 1, 0\}$ for every $(A,w) \in \bbA'$, where $r_{(A,w)}$ is the smallest round in which $\Lambda$ accepts $(A,w)$. Now, let $\cA = \{\, (A, w, \ell_{(A,w)}) \mid (A, w) \in \bbA' \,\}$ and let $\cV = \{Y_1, \ldots, Y_m\}$ be the set of accepting predicates of $\Lambda$ containing precisely $m$ distinct predicates. For each $i \in [m]$, we define $\cA^{\mathsf{base}}_i = \{\, (A, w, 0) \in \cA \mid A, w \models Y_i^{0} \,\}$ and $\cA^{\mathsf{iter}}_i = \{\, (A, w, \ell) \in \cA \mid A, w \models Y_i^{\ell+1} \,\}$.
    Samson chooses $2m$ distinct nodes $U_0 = \{v^0_1, v^1_1, \ldots, v^0_{m}, v^1_{m}\}$ and defines the following functions.
    \begin{itemize}
        \item A function $\rho_0 \colon U_0 \to \cV \times \{\mathsf{base}, \mathsf{iter}\}$ such that $\rho_0(v^0_i) = (Y_i, \mathsf{base})$ and $\rho_0(v^1_i) = (Y_i, \mathsf{iter})$.
        \item A function $\leftf_0 \colon U_0 \to \wp(\cA)$ such that $\leftf_0(v^0_i) = \cA^{\mathsf{base}}_i$ and $\leftf_0(v^1_i) = \cA^{\mathsf{iter}}_i$.
    \end{itemize}
    Now, we see that for all $(B, w, \ell_{(B,w)}) \in \CM_*(\bbB)$, $(B, w)$ is not accepted by $\Lambda$ in round $\ell_{(B,w)}$.
    Thus, no matter which subset $\cB \subseteq \CM_*(\bbB)$ Delilah chooses, in each initial position 
    \[
    ((U_0, \emptyset, \emptyset, \rho_0, \emptyset), U_0, \leftf_0, \rightf_0), 
    \]
    where $\rightf_0$ is induced by $\cB$ and $\rho_0$,
    Condition (\ref{induction_hypo_2}) holds.

    Assume that Condition (\ref{induction_hypo_2}) holds in the current position $P$. We prove that Condition (\ref{induction_hypo_2}) holds in the position following $P$.
    
    First assume that $v \notin \dom(\lambda)$.
    \begin{itemize}
        \item $\lambda_\Lambda(g(v)) = \varphi \in \Pi \cup \{\top, \bot\}$: By the induction hypothesis, $\lambda_\Lambda(g(v))$ separates $\cL$ from $\cR$.
        \item $\lambda_\Lambda(g(v)) = \neg$: By the induction hypothesis, Condition (\ref{induction_hypo_2}) holds trivially in the position following $P$. 
        \item $\lambda_\Lambda(g(v)) = \lor$: Let $s_1$ and $s_2$ be the successors of $g(v)$. Samson splits the set $\cL$ as follows: $\cL_1 = \{\, (M, w, \ell) \in \cL \mid M, w \models \Lambda^\ell_{s_1} \,\}$ and $\cL_2 = \{\, (M, w, \ell) \in \cL \mid M, w \models \Lambda^\ell_{s_2} \,\}$. On the other hand, for every $(N, u, \ell) \in \cR$ and for every $i \in \{1,2\}$ it holds that $N, u \not\models \Lambda^\ell_{s_i}$, and since $\cR$ is not split, Condition (\ref{induction_hypo_2}) holds in the position following $P$. 
        Let $v_1$ and $v_2$ be fresh nodes not in $V$.
        To ensure uniformity, we set $g' = g[s_1/v_1, s_2/v_2]$. 
        \item $\lambda_\Lambda(g(v)) = \Diamond_{\geq m}$: Let $s$ be the successor of $g(v)$. 
        Now, $\Lambda_{g(v)} = \Diamond_{\geq m} \Lambda_s$.
        Therefore, for every $(M, w, \ell) \in \cL$ and for every $(N, u, \ell) \in \cR$, it holds that $M, w \models \Diamond_{\geq m} \Lambda^{\ell}_s$ and $N, u \not\models \Diamond_{\geq m} \Lambda^{\ell}_s$.
        Thus, for every $(M, w, \ell) \in \cL$, there are $m$ distinct successors $w_1, \ldots, w_m$ of $w$ such that for every $i \in [m]$, $M, w_i \models \Lambda_s^\ell$. Moreover, for every $(N, u, \ell) \in \cR$, for some $m' < m$, there are precisely $m'$ distinct successors $u_1, \ldots, u_{m'}$ of $u$ such that for every $i \in [m']$, $N, u_i \models \Lambda_s^\ell$.

        Thus, Samson can define an $m$-successor function $f$ for $\cL$ such that for every $(M, w, \ell) \in \cL$ and for every $(M, w', \ell) \in f(M, w, \ell)$, we have $M, w' \models \Lambda^\ell_s$.
        On the other hand, for every non-empty partial $m$-successor function $h$ for $\cR$, for each $(N, u, \ell) \in \dom(h)$, there is $(N, u', \ell) \in h(N, u, \ell)$ such that $N, u' \not\models \Lambda^\ell_s$.
        
        Therefore, Condition (\ref{induction_hypo_2}) holds in the next position.
        To ensure uniformity, we set $g' = g[s/v']$, where $v'$ is a fresh node. 
        \item 
        $\lambda_\Lambda(g(v)) = X \in \mathrm{VAR}$: By the induction hypothesis, Condition (\ref{induction_hypo_2}) holds in the current position. Let $\varphi_X$ be the body of the base rule of $X$ and let $\psi_X$ be the body of the induction rule of $X$. 

        Now, for every $(M, w, \ell) \in \cL$ with $\ell > 0$, we have $M, w \models \psi_X^{\ell-1}$. Also, for every $(N, u, \ell) \in \cR$ with $\ell > 0$, we have $N, u \not\models \psi_X^{\ell-1}$. Moreover, for every $(M, w, 0) \in \cL$, we have $M, w \models \varphi_X$ and for every $(N, u, 0) \in \cR$, we have $N, u \not\models \varphi_X$.

        Therefore, for every $(M, w, \ell) \in \mathrm{iter}(\cL)$, we have $M, w \models \psi_X^{\ell}$ and for every $(M, w, 0) \in \mathrm{init}(\cL)$, we have $M, w \models \varphi_X$.
        Analogously, for every $(N, u, \ell) \in \mathrm{iter}(\cR)$, we have $N, u \not\models \psi_X^{\ell}$ and for every $(N, u, 0) \in \mathrm{init}(\cR)$, we have $N, u \not\models \varphi_X$. Thus, Condition (\ref{induction_hypo_2}) holds in the next position.

        We can ensure uniformity as follows. Let $s$ and $s'$ be the roots of the trees of the induction rule and the base rule of $X$ in $\cF_\Lambda$, respectively, i.e., $\rho_\Lambda(s) = (X,\mathsf{iter})$ and $\rho_\Lambda(s') = (X,\mathsf{base})$. Moreover, let $v'$ (resp.\ $v''$) be the root of $\cF$ with $\rho(v') = (X,\mathsf{iter})$ (resp.\ $\rho(v'') = (X,\mathsf{base})$) if such a root exists, and a fresh node otherwise. If Delilah does not challenge Samson, the play continues from $v'$. If $v'$ already exists, then $g(v') = s$, since $g$ maps roots to roots and preserves the labels given by $\rho$, and we set $g' = g$; otherwise, we set $g' = g[s/v']$. If Delilah challenges Samson, the play continues from $v''$, and analogously we set $g' = g$ if $v''$ already exists and $g' = g[s'/v'']$ otherwise.
    \end{itemize}
    The cases for dual operators are analogous, and the cases for $\DiamondG_{\geq m}$ and $\BoxG_{< m}$ are analogous to the local counting operators.

    If $v \in \dom(\lambda)$, then the arguments are similar. The only difference is that, in the case of disjunction, conjunction and variables, there might be models that are left at a node of the forest to wait for verification. We consider the case where, during a play of the game, we enter a node that already stores models. 

    $\lambda(v) = \lor$: Let $\cL_1, \cL_2 \subseteq \cL$ be the sets of models chosen by Samson, analogous to the unlabeled case. 
    Let $v_1$ and $v_2$ be the successors of $v$.
    By the induction hypothesis, Condition (\ref{induction_hypo_2}) holds for the sets $\leftf(v_i)$, $\rightf(v_i)$ and $\cR$, i.e., for every $i \in \{1,2\}$, for every $(M, w, \ell) \in \leftf(v_i)$, $M, w \models \Lambda^\ell_{g(v_i)}$ and also for every $(N, u, \ell) \in \rightf(v_i) \cup \cR$, $N, u \not\models \Lambda^\ell_{g(v_i)}$. 
    Therefore, for every $i \in \{1,2\}$ and every $(M, w, \ell) \in \cL_i \cup \leftf(v_i)$, we have $M, w \models \Lambda^\ell_{g(v_i)}$. Similarly, for every $i \in \{1,2\}$ and every $(N, u, \ell) \in \cR \cup \rightf(v_i)$, we have $N, u \not\models \Lambda^\ell_{g(v_i)}$. Thus, Condition (\ref{induction_hypo_2}) holds in the next position.

    The cases where $\lambda(v) = \land$ or $\lambda(v) \in \mathrm{VAR}$ are proved similarly.
\end{proof}

Next, we prove that, over finite models, if Samson has a non-uniform strategy in a formula size game, then he also has a uniform strategy in the game, and vice versa.
\begin{lemma}\label{lem: finite models non-uniform equiv uniform}
    Let $\Pi$ be a finite set of proposition symbols.
    Let $\bbC$ and $\bbD$ be classes of finite pointed $\Pi$-models. The following claims are equivalent.
    \begin{enumerate}
        \item Samson has a uniform winning strategy in $\mathrm{FS}^\Pi_k(\bbC, \bbD)$.
        \item Samson has a non-uniform winning strategy in $\mathrm{FS}^\Pi_k(\bbC, \bbD)$. 
    \end{enumerate}
\end{lemma}
\begin{proof}
The direction ``$1 \Rightarrow 2$'' is trivial, so let us consider the direction ``$2 \Rightarrow 1$''.

Let $T$ be the subtree of the game tree of $\mathrm{FS}^\Pi_k(\bbC, \bbD)$ induced by the winning strategy $\sigma$ of Samson. 
Given a $\bbD' \subseteq \bbD$, we say that $\bbD'$ is \emph{fully clocked} in an initial position $P = (\cF, U, \leftf, \rightf)$ of $T$ if for each $v \in U$, $(D, w) \in \bbD'$ and $\ell \in [0; 2^{\abs{D}k}]$ we have $(D, w, \ell) \in \rightf(v)$, where $\abs{D}$ is the number of nodes in $D$.
Let $\Gamma_1, \ldots, \Gamma_K$ enumerate all $\Pi$-programs of $\GGMSC$ of size at most $k$ up to renaming of schema variables; this is possible since $\Pi$ is finite, and hence there are only finitely many programs of size at most $k$ up to renaming. Note that renaming schema variables does not change which pointed models a program accepts. In what follows, we identify programs that differ only in the names of their schema variables.

We show that for all finite subsets $\bbC'\subseteq \bbC$ and $\bbD' \subseteq \bbD$, there is a subtree $T_{unif}$ of $T$ that induces a uniform winning strategy for Samson in $\mathrm{FS}^\Pi_k(\bbC', \bbD')$. We construct $T_{unif}$ from $T$ as follows.
\begin{itemize}
    \item First, let $T_{prune}$ be a subtree of $T$ induced from the initial position $P$ of $T$ where the models in $\bbC'$ are clocked and the models in $\bbD'$ are fully clocked.
    \item 
    We then prune $T_{prune}$ recursively in the positions where an unlabeled move is played, starting from the root. If, in the current position $P = (\cF, U, \leftf, \rightf)$, a move is played at an unlabeled node $v \in U$, then we remove all the branches from $T_{prune}$ in which $v$ is labeled with a different symbol than in the position $P$.
    After that, we look at the successors of $v$ and continue this process until $T_{prune}$ can no longer be pruned.
\end{itemize}
After $T_{prune}$ is pruned, it is clear that there is a program $\Gamma_i$ such that for each position in $T_{prune}$ there is a position embedding to $\Gamma_i$.  

Now, we use the constructed $T_{prune}$ to induce a uniform winning strategy for Samson in $\mathrm{FS}^\Pi_k(\bbC', \bbD')$ w.r.t. $\Gamma_i$ as follows. The tree $T_{prune}$ already induces a \emph{partial} uniform winning strategy for Samson starting from the initial position, where $\bbD'$ is fully clocked. Next, we extend $T_{prune}$ into a tree $T_{unif}$ such that it induces a full uniform winning strategy for Samson in $\mathrm{FS}^\Pi_k(\bbC', \bbD')$ \emph{starting from the initial position of} $T_{prune}$. 
Assume that there is a position $P = (\cF, U, \leftf, \rightf)$ in $T_{prune}$ such that for some unlabeled node $v \in U$ chosen by Delilah, a strategy for Samson is not described. By the construction of $T_{prune}$, there is a position $P_v = (\cF_v, U_v, \leftf_v, \rightf_v)$ in $T_{prune}$ where the node $v$ is labeled for the first time. Now, in the position $P$, Samson mimics the move played in $P_v$. This is possible because $P_v$ is the position where $v$ is labeled for the first time, and thus, $\leftf(v) \subseteq \leftf_v(v)$ and $\rightf(v) \subseteq \rightf_v(v)$. It is clear that if Samson mimics the strategy played in $P_v$, then there is a position embedding to $\Gamma_i$. 
The case where $v$ is labeled is analogous, but in that case Samson mimics a move played in a position $P_v = (\cF_v, U_v, \leftf_v, \rightf_v)$ with $\leftf(v) \subseteq \leftf_v(v)$ and $\rightf(v) \subseteq \rightf_v(v)$. This process can be continued until we cannot extend $T_{prune}$ any further, and as a result we obtain the tree $T_{unif}$.

Now, $T_{unif}$ induces a uniform winning strategy $\sigma_{unif}$ starting from its initial position, since by the construction above there is a program $\Gamma \in \{\Gamma_1, \ldots, \Gamma_K\}$ such that for each position in $T_{unif}$ there is a position embedding to $\Gamma$. By using $\sigma_{unif}$, we can show, analogously to the proof of Theorem \ref{thrm: uniform fs-game characterization}, that the following condition holds in every position $P = (\cF, U, \leftf, \rightf)$ following the initial position of $T_{unif}$, at every node $v \in U$:
\begin{equation*}
\begin{split}
M, w \models \Gamma_{g(v)}^\ell &\text{ for each } (M, w, \ell) \in \cL, \\
N, u \not\models \Gamma_{g(v)}^\ell &\text{ for each } (N, u, \ell ) \in \cR, \\
\end{split}
\end{equation*}
where $\leftf(v) = \cL$, $\rightf(v) = \cR$ and $g$ is a position embedding w.r.t. $\Gamma$ and $P$. 
Now, the set $\bbD'$ is fully clocked in the initial position of $T_{unif}$, and $T_{unif}$ induces a uniform winning strategy starting from its initial position w.r.t. the program $\Gamma$, and so $\Gamma$ separates $\bbC'$ from $\bbD'$. To see this, note that since $\bbD'$ is fully clocked in the initial position of $T_{unif}$, every possible global configuration for every pointed model in $\bbD'$ is ``verified'' during the game starting from the initial position. Note that with infinite models, we could not fully clock $\bbD'$.

Now, for all finite subsets $\bbC' \subseteq \bbC$ and $\bbD' \subseteq \bbD$, 
there is a program $\Gamma$ that separates $\bbC'$ from $\bbD'$, and we say that $\Gamma$ is \emph{realized} from $\sigma$.
Next, we show that one of the programs that are realized from $\sigma$ must separate $\bbC$ from $\bbD$.
Let $\Lambda_1, \ldots, \Lambda_m$ enumerate all the programs that are realized from $\sigma$. Aiming for a contradiction, suppose that none of the programs $\Lambda_1, \ldots, \Lambda_m$ separates $\bbC$ from $\bbD$. Therefore, for each $\Lambda_i$, we can choose a model $(M_i, w_i) \in \bbC \cup \bbD$ such that either $(M_i, w_i) \in \bbC$ and $M_i, w_i \not\models \Lambda_i$, or $(M_i, w_i) \in \bbD$ and $M_i, w_i \models \Lambda_i$.

Now, consider the finite sets $\bbC' \colonequals \{\,(M_i, w_i) \mid i \in [m] \,\} \cap \bbC$ and $\bbD' \colonequals \{\, (M_i, w_i) \mid i \in [m] \,\} \cap \bbD$. Samson has a winning strategy in $\mathrm{FS}^\Pi_k( \bbC', \bbD')$ by the observation above. 
Thus, there is a program $\Lambda_i$ that separates $\bbC'$ from $\bbD'$, which is a contradiction. Hence, there must be a program that separates $\bbC$ from $\bbD$.
\end{proof}

From Theorem \ref{thrm: uniform fs-game characterization} and Lemma \ref{lem: finite models non-uniform equiv uniform}, the following theorem follows, which informally extends the game characterization for non-uniform strategies over finite models.

\begin{theorem}\label{thrm: finite uniform non-uniform characterization}
    Let $\Pi$ be a finite set of proposition symbols.
    Let $\bbC$ and $\bbD$ be classes of finite pointed $\Pi$-models. The following claims are equivalent.
    \begin{enumerate}
        \item Samson has a uniform winning strategy in $\mathrm{FS}^\Pi_k(\bbC, \bbD)$.
        \item Samson has a non-uniform winning strategy in $\mathrm{FS}^\Pi_k(\bbC, \bbD)$. 
        \item There is a $\Pi$-program of $\GGMSC$ of size at most $k$ that separates $\bbC$ from $\bbD$.
    \end{enumerate}
\end{theorem}

\begin{remark}
For $\GMSC$, $\MSC$ and $\SC$, it is trivial to obtain an analogous result for both Theorem \ref{thrm: uniform fs-game characterization} and Theorem \ref{thrm: finite uniform non-uniform characterization}. The formula size game is also trivial to extend to 
asynchronous variants as follows. In asynchronous variants, clocks act as binary flags: $1$ allows iteration, $0$ disables it. Moreover, whenever an $X$-move is played in a position $(\cF, U, \leftf, \rightf)$ at a node $v \in U$, Samson chooses which flags are set to 0 in each model of $\leftf(v)$, and Delilah does the same for each model of $\rightf(v)$.
\end{remark}

Next, we recall some notions related to bisimulations.
Global counting bisimilarity is defined as follows. Given two pointed $\Pi$-models, $((W, R, V), w)$ and $((W', R', V'), w')$, we say that they are \textbf{globally counting bisimilar} if there is a relation $Z$ (called a global counting bisimulation) satisfying the following conditions: 
\begin{itemize}
    \item $(w, w') \in Z$.
    \item \emph{atomic}: For all $(v, v') \in Z$, $V(v) = V'(v')$. 
    \item \emph{local forth}: If $(v, v') \in Z$, then for all $k \in \N$, for each $k$ distinct out-neighbors $v_1, \ldots, v_k \in W$ of $v$, there are $k$ distinct out-neighbors $v'_1, \ldots, v'_k \in W'$ of $v'$ such that $(v_1, v'_1), \ldots, (v_k, v'_k) \in Z$.
    \item \emph{local back}: If $(v, v') \in Z$, then for all $k \in \N$, for each $k$ distinct out-neighbors $v'_1, \ldots, v'_k \in W'$ of $v'$, there are $k$ distinct out-neighbors $v_1, \ldots, v_k \in W$ of $v$ such that $(v_1, v'_1), \ldots, (v_k, v'_k) \in Z$.
    \item \emph{global forth}: For every $k \in \N$, and for every $k$ distinct nodes $v_1, \ldots, v_k \in W$, there are $k$ distinct nodes $v'_1, \ldots, v'_k \in W'$ such that $(v_1, v'_1), \ldots, (v_k, v'_k) \in Z$.
    \item \emph{global back}: For every $k \in \N$, and for every $k$ distinct nodes $v'_1, \ldots, v'_k \in W'$, there are $k$ distinct nodes $v_1, \ldots, v_k \in W$ such that $(v_1, v'_1), \ldots, (v_k, v'_k) \in Z$.
\end{itemize}

The corresponding bisimulation game is played between two players, $\I$ and $\II$, over two pointed $\Pi$-models $((W, R, V), w)$ and $((W', R', V'), w')$. The set of positions of the game consists of pairs $(v, v') \in W \times W'$ and the initial position of the game is $(w, w')$. In each position $(v, v')$, $\I$ can play one of the moves below:
\begin{itemize}
    \item \emph{atomic-move}: The game ends, and $\II$ wins, if $V(v) = V'(v')$; otherwise, $\I$ wins. 
    \item \emph{local-move}: $\I$ picks a node $u \in \{v, v'\}$ and chooses a finite non-empty subset of out-neighbors of $u$, and then $\II$ responds by choosing a subset of out-neighbors of the other node $u' \in \{v, v'\} \setminus \{u\}$ of the same size. After that, $\I$ chooses a node from the set given by $\II$, and $\II$ chooses a node from the set given by $\I$; then the game continues from the pair induced by these choices. 
    \item \emph{global-move}: $\I$ picks a node $u \in \{v, v'\}$ and chooses a finite non-empty subset of nodes from its model, and then $\II$ responds by choosing a set of nodes from the model of the other node $u' \in \{v, v'\} \setminus \{u\}$ of the same size. After that, $\I$ chooses a node from the set given by $\II$, and $\II$ chooses a node from the set given by $\I$; then the game continues from the pair induced by these choices. 
\end{itemize}
If a player cannot make a choice required by a move, then the player loses. It is not hard to show that the corresponding bisimulation game characterizes global counting bisimilarity, i.e., two pointed models are globally counting bisimilar iff the player $\II$ has a winning strategy in the corresponding bisimulation game. 
Moreover, given an $n \in \N$, we say that $(M, w)$ and $(N, u)$ are \textbf{globally counting $n$-bisimilar} if $\II$ has a winning strategy in the variant of the corresponding bisimulation game in which $\I$ can play at most $n$ local-moves or global-moves, while atomic-moves can be played at any time. In particular, globally counting $0$-bisimilar pointed models satisfy the same proposition symbols.

Now, we can prove the following lemma, 
which intuitively states that Delilah has a winning strategy if there are two globally counting $n$-bisimilar pointed models (one on the left and another on the right) in the current position, where $n$ is sufficiently large w.r.t. the clocks of the bisimilar models and the number of diamonds and boxes available.
\begin{lemma}\label{lemma: finite-bisimilar-delilah}
    Let $P = (\cF, U, \mathrm{left}, \mathrm{right})$ be a position in $\mathrm{FS}^\Pi_k(\bbA, \bbB)$ with $\abs{\cF} \leq k$, and let $K = k - \abs{\cF} + M$, where $M$ is the number of (local and global) diamonds and boxes appearing in $\cF$.
    If there is a node $v \in U$ and clocked models $(A, w_A, \ell) \in \leftf(v)$ and $(B, w_B, \ell) \in \rightf(v)$ such that $(A, w_A)$ and $(B, w_B)$ are globally counting $n$-bisimilar for some $n \geq K(\ell+1)$, then Delilah has a winning strategy from position $P$.
\end{lemma}
\begin{proof}
We call the two clocked models that Delilah follows the \emph{tracked models}. Initially, these are $(A, w_A, \ell)$ and $(B, w_B, \ell)$. At every position, one of them is on the left and the other on the right of the \emph{tracked node}, which is a node of $U$. In every position, Delilah chooses the tracked node and plays as follows.
\begin{itemize}
    \item \emph{$\neg$-move:} The tracked models switch sides at the successor, which becomes the tracked node.
    \item \emph{$\lor$-move (resp. $\land$-move):} Samson's split puts the tracked model on the left (resp. right) into some $\cL_i$ (resp. $\cR_i$), and the other tracked model is copied to both successors. The successor $v_i$ becomes the tracked node.
    \item \emph{$\Diamond_{\geq m}$-move:} Let $(C, c, \ell')$ and $(D, e, \ell')$ be the tracked models on the left and on the right, respectively. They are globally counting $n'$-bisimilar for some $n' \geq 1$ (see below). Let $f(C, c, \ell') = \{(C, c_i, \ell') \mid i \in [m]\}$ be given by Samson. Let $\{e_1, \ldots, e_m\}$ be the answer of $\II$ to the local-move of $\I$ choosing $\{c_1, \ldots, c_m\}$. Delilah gives $h(D, e, \ell') = \{(D, e_j, \ell') \mid j \in [m]\}$ and chooses $\cL' = \Diamond_f \cL$. Whichever $(D, e_j, \ell')$ Samson keeps, the strategy of $\II$ gives some $c_i$ such that $(C, c_i)$ and $(D, e_j)$ are globally counting $(n'-1)$-bisimilar. The clocked models $(C, c_i, \ell')$ and $(D, e_j, \ell')$ become the new tracked models. The $\Box_{< m}$-move is symmetric, and the $\DiamondG_{\geq m}$- and $\BoxG_{< m}$-moves are analogous, using global-moves of $\I$. If the move makes $\abs{\cF} > k$, Samson loses, and Delilah's answer is irrelevant.
    \item \emph{$X$-move:} If the clock of the tracked models is positive, Delilah does not challenge Samson, and both tracked models continue with their clock decreased by one. Otherwise, she challenges him, and both continue with clock $0$ in a tree of a base rule.
\end{itemize}
Samson cannot win by an $X$-move: Delilah challenges if the clock is $0$, and otherwise $\mathrm{iter}(\leftf(v))$ contains a tracked model. Nor can he win a Sig-move: the tracked models are always globally counting $0$-bisimilar (see below), and thus they satisfy the same proposition symbols. Since every play is finite by Proposition \ref{prop: finite game tree}, Delilah wins.

It remains to show that at most $n$ diamond and box moves are played at the tracked node. Each such move uses one round of the bisimulation game, so this guarantees that every such move is played with $n' \geq 1$ remaining rounds, and that the tracked models remain at least $0$-bisimilar until the end of the play.

Consider a play starting from $P$, and let $\cF_{end}$ be the forest of its last position with $\abs{\cF_{end}} \leq k$. We first bound the number of diamonds and boxes in $\cF_{end}$. The forest $\cF$ contains $M$ of them, and each diamond or box that is added during the play costs at least one resource. Since $\abs{\cF_{end}} \leq k$, at most $k - \abs{\cF}$ of them are added, and hence $\cF_{end}$ contains at most $k - \abs{\cF} + M = K$ diamonds and boxes.

Next, we divide the play into \emph{segments} separated by the $X$-moves played at the tracked node. Within a segment, the tracked node only moves along edges, and thus it follows a directed path within a single tree of $\cF_{end}$. Such a path visits each node at most once, so each segment passes at most $K$ diamonds and boxes. At most $\ell + 1$ $X$-moves are played at the tracked node: at most $\ell$ without a challenge, each of which decreases the clock by one, and possibly a final one with a challenge. Hence, there are at most $\ell + 2$ segments. All segments before the challenge lie in trees of induction rules, since an $X$-move without a challenge leads to a tree of an induction rule, and the segment after the challenge lies in a tree of a base rule. (If the tracked node starts in a tree of a base rule, then no $X$-moves occur at all, since base rules contain no schema variables, and the play consists of a single segment.) The last segment before the challenge and the segment after it lie in disjoint trees, so they pass distinct diamonds and boxes, and together they pass at most $K$ of them. The remaining at most $\ell$ segments pass at most $K$ each. Therefore, at most $\ell K + K = K(\ell + 1) \leq n$ diamond and box moves are played at the tracked node.
\end{proof}

\begin{remark}\label{remark: bisimilar-delilah}
Counting bisimilarity is defined analogously to the global counting bisimilarity, except that the global forth and back conditions are omitted. Non-counting bisimilarity, in turn, is obtained from the counting bisimilarity by restricting $k$ to $k \leq 1$ in the local forth and local back conditions.
Lemma \ref{lemma: finite-bisimilar-delilah} trivially generalizes to $\GMSC$ and the counting bisimilarity, and $\MSC$ and the non-counting bisimilarity. 
\end{remark}

Next, we demonstrate the power of the formula size game and show for a fragment of $\GGMSC$ that a certain reachability property is not expressible.

We say that a program of $\GGMSC$ is in \textbf{strong positive normal form} if it is in strong negation normal form, $\land$-symbols do not appear in induction rules and it does not contain any boxes of the form $\BoxG_{< k}$.
As an example, consider a variant of the center-point property (cf. Example \ref{example: centre-point}), which is satisfied by a pointed model $(M, w)$ iff ``there exists \(n \in \N\) such that from every out-neighbor of \(w\) there is a walk of length \(n\) to a node that does not have any out-neighbors''. 
This property can be expressed by the following $\MSC$-program in strong positive normal form: $X (0) \colonminus \Box \bot$, $X \colonminus \Diamond X$, $Y (0) \colonminus \bot$, $Y \colonminus \Box X$, where $Y$ is an accepting predicate. 
By using an argument similar to that in the proof of Proposition 6 in \cite{Kuusisto13}, it is easy to show that this variant of the center-point property is not expressible in $\MSO$. Assume the contrary, that this property is expressible by an $\MSO$-formula $\varphi$. Using $\varphi$, we can define an $\MSO$-formula $\psi$ which defines the following language $L = \{\, 0^n10^n \in \{0,1\}^* \mid n \in \N \,\}$.
Now, $L$ is regular since it is definable in $\MSO$, but by using the pumping lemma for regular languages, it is easy to show that $L$ is actually not a regular language, which leads to a contradiction. So programs of $\mathrm{GGMSC}$ in strong positive normal form have also orthogonal expressive power with $\mathrm{MSO}$.

\begin{remark}\label{remark: positive normal form}
    Lemma \ref{lemma: finite-bisimilar-delilah} can be generalized to $\GGMSC$-programs in strong positive normal form as follows.  
    Let $P =(\cF, U, \mathrm{left}, \mathrm{right})$ be a position in the game $\mathrm{FS}^\Pi_k(\bbA, \bbB)$ restricted to programs in strong positive normal form. Let $K = k-\abs{\cF} + M$, where $M$ is the number of diamonds and boxes appearing in $\cF$.
    Below, \emph{strong positive global counting bisimilarity} refers to a global counting bisimilarity in which the global back condition is omitted.
    Now, if there is a node $v \in U$ and clocked models $(A, w_A, \ell) \in \leftf(v)$ and $(B, w_B, \ell) \in \rightf(v)$ such that $(A, w_A)$ and $(B, w_B)$ are \emph{strong positive} globally counting $n$-bisimilar for some $n \geq K(\ell+1)$, then Delilah has a winning strategy from position $P$.
\end{remark}

The \textbf{prime reachability property} contains all pointed Kripke models $(M, w)$ where there is a node $u$ with zero out-neighbors and there is a directed path from $w$ to $u$ in $M$ such that the length of the path is a prime number.
The following proposition shows that the prime reachability property is not definable by any program of $\GGMSC$ that is in strong positive normal form.

\begin{proposition}\label{prop: formulasize game demo}
    There is no $\Pi$-program of $\GGMSC$ in strong positive normal form that defines the prime reachability property.
\end{proposition}
\begin{proof}
    Assume, for the sake of contradiction, that there is a $\Pi$-program $\Lambda$ of $\GGMSC$ in strong positive normal form that defines prime reachability. Let $\mathscr{P}$ denote the set of prime numbers. For each $n \in \Z_+$, we let $p_n$ denote the $n$th prime number.

    Before we start describing a strategy for Delilah, we introduce some auxiliary notions and notations.
    Let $k$ be the size of $\Lambda$, and we may assume w.l.o.g. that $k \geq 2$.
    For each $d \in [k-1]$, we let $\ell_{d,k} \in \N$ denote the smallest number $\ell$ such that $p_k + d\ell \in \N \setminus \primes$.
    Now, let 
    \[
    H = \{\, h \in \N \setminus \primes \mid h = p_k + d\ell_{d,k}, d \in [k-1]\,\}.
    \]
    A model $(W, R, V)$ is \emph{path-shaped} if $R$ forms a directed path over $W$. We let $P_t = (W, R, V)$ denote the path-shaped model where $W = \{w_t, \ldots, w_0\}$, $R$ forms the directed path $w_t, \ldots, w_0$, and $V(v) = \emptyset$ for every $v \in W$.

    We define $\bbA = \{ ( P_{p_k}, w_{p_k} ) \}$ and $\bbB = \{\, (P_h, w_h) \mid h \in H \,\}$. 
    Intuitively, $d$ corresponds to the number of local diamonds and boxes in a cycle that Samson may construct during the game, and $n$ (used below) to the number of schema variables in that cycle.

    Let $P = (\cF, U, \leftf, \rightf)$ be a position in a formula size game. 
    Given a node $v$ from $\cF$, we let $\mathrm{reach}(\cF,v)$ denote the set of all the nodes reachable from $v$.\footnote{In a forest, a node $u$ is reachable from a node $v$ if there is a directed walk from $v$ to $u$ through edges or back edges.}
    We define
    \[
    \leftf(\cF, v) \colonequals \bigcup_{u \in \mathrm{reach}(\cF, v)} \leftf(u), \qquad \rightf(\cF, v) \colonequals \bigcup_{u \in \mathrm{reach}(\cF, v)} \rightf(u).
    \]
    We say that a clocked model $(A, w, \ell) \in \leftf(v) \cup \leftf(\cF, v)$ is \textbf{left-relevant} (w.r.t. $v$) if $\ell$ is the least clock of the pointed model $(A, w)$ in this union, i.e., $\ell = \min\{\, \ell' \mid (A, w, \ell') \in \leftf(v) \cup \leftf(\cF, v) \,\}$. \textbf{Right-relevance} (w.r.t. $v$) is defined analogously via $\rightf(v) \cup \rightf(\cF, v)$. Moreover, if additionally $(A, w, \ell) \in \leftf(v')$ for a node $v' \in \{v\} \cup \mathrm{reach}(\cF, v)$, we say that $(A, w, \ell)$ is left-relevant \emph{at} $v'$ (w.r.t. $v$), and right-relevance at $v'$ is defined analogously. A (directed) \textbf{cycle} in a forest is a walk through edges or back edges that starts and ends at the same node.

    Now, we may describe a winning strategy for Delilah in $\mathrm{FS}^\Pi_k(\bbA, \bbB)$ restricted to strong positive normal form and then we conclude that there is no program of $\GGMSC$ in strong positive normal form which defines the prime reachability property. Below, every position and move refers to this restricted game.

    Before we formally describe an invariant that Delilah tries to maintain during the game $\mathrm{FS}^\Pi_k(\bbA, \bbB)$, we describe the invariant informally.
    At the beginning of the game, Delilah uses the set $H$ to choose suitable clocks for the models in $\bbB$ such that she can direct the game into a position such that there are bisimilar clocked models on the ``left'' and ``right''. During the game, Delilah follows a left-relevant model and retains all the clocked models on the ``right'' before a cycle is constructed within the partial subprogram of the game. After a cycle is constructed, Delilah uses the cycle to pump one of the clocked models on the right such that she ends up in a position where there are bisimilar clocked models on the left and right with the same clocks. There are also some (easy) pathological cases that Delilah has to take into account, e.g., if Samson chooses at the beginning of the game a ``bad clock'' for the model in $\bbA$.

    Next, with this intuition, we formally define the invariant that Delilah tries to maintain during the game in each position $P = (\cF, U, \leftf, \rightf)$. There exists a node $v \in U$ such that the following holds. 
    Below, we let $(P_{p_k}, w_{j}, \ell)$ denote a left-relevant model at a node $v'$ w.r.t. $v$, and we let $\leftf(v) = \cL$ and $\rightf(v) = \cR$.
    %
    \begin{enumerate}
        \item If $v = v'$, then one of the below holds.
        \begin{enumerate}
            \item For all $d, n \in [k-1]$ with $d+n \leq k$, there is a clocked model $(P_h, w_i, \ell + n\ell_{d,k}) \in \cR$ where $h = p_k + d\ell_{d,k}$ and $h - i = p_k - j$ with $j > 1$.
            \item There is a clocked model $(P_h, w_j, \ell)$ in $\cR$, for some $h \in H$.
        \end{enumerate}
        \item If $v' \neq v$, then $v'$ is reachable from $v$ with a walk that consists of $d$ diamonds or boxes and $n$ variables, for some $d, n \in \N$, and there is a clocked model $(P_h, w_{j+d}, \ell+n)$ in $\cR$.    
    \end{enumerate}
    If there \emph{does not exist} a left-relevant model for any node in $U$, then the following holds. There is a node $v$ of $\cF$ and a clocked model $(P_h, w_i, \ell') \in \cR$ in a cycle such that there is a walk $\bw$ starting from $v$ such that one of the following holds. 
    \begin{enumerate}[resume]
        \item The walk $\bw$ consists of at most $\ell'$ variables and at least one diamond, 
        \item or $\bw$ consists of $i+1$ boxes and fewer than $\ell'$ variables.
    \end{enumerate}
    
    Now, we can describe a winning strategy for Delilah in $\mathrm{FS}^{\Pi}_k( \bbA, \bbB)$.
    \begin{itemize}
        \item At the beginning of the game, Delilah chooses $\bbA$ and then Samson chooses a clock $\ell$ for the model $(P_{p_k}, w_{p_k})$. If $(\ell+1) k \leq p_k$, then, by Remark \ref{remark: positive normal form}, Delilah wins by choosing $(P_{p_k +1}, w_{p_k+1}, \ell)$. 
        Thus, we suppose that $(\ell+1) k > p_k$, and in that case Delilah chooses the set $\cB$ of clocked models that consists of, for each $d \in [k-1]$ and each $n \in \N$ with $d + n \leq k$, the clocked model $(P_h, w_h, \ell + n\ell_{d,k})$, where $h = p_k + d\ell_{d,k}$.
        Clearly, in the starting position of the game Condition $1(a)$ of the invariant holds. 

        \item Now, we describe a winning strategy for Delilah in every position $(\cF, U, \leftf, \rightf)$ of the game $\mathrm{FS}^\Pi_k(\bbA, \bbB)$. In each position, Delilah chooses a \emph{relevant} node $v \in U$, i.e., a node which satisfies the invariant described above. It does not matter which relevant node Delilah chooses in each position. After the starting position, there is a relevant node that contains a clocked model for a model in $\bbA$. We assume that Condition $1(a)$ or $2$ of the invariant holds in the current position, and show that Condition $1(b)$, $3$, or $4$ will eventually hold, and then conclude how Delilah can win if one of those conditions holds.
        First, note that while Condition $1(a)$ or $2$ holds, the play stays within the trees of induction rules: Samson has to place the clocked models of the left at nodes of induction-rule trees, since a clocked model with a positive clock at a base-rule node loses the game for Samson immediately, and Delilah never challenges $X$-moves. In particular, since $\Lambda$ is in strong positive normal form, Samson cannot play $\land$-moves or $\BoxG_{<m}$-moves at the relevant nodes, and thus the moves considered below cover all the moves available to Samson. Moreover, Samson must eventually construct a cycle: since $\abs{\cF} \leq k$ and the clocked models on the left and right survive the positive bisimulation game for at least $p_k > k$ rounds, Samson cannot separate them without iterating schema variables through back edges.
        \begin{itemize}
            \item Assume that Samson plays a Sig-move. Clearly, Condition $2$ cannot hold in that situation (a Sig-move is played at an unlabeled node, from which no node is reachable), so Condition $1$ must hold and in that case Delilah clearly wins, since every valuation is empty and hence no symbol in $\Pi \cup \{\top, \bot\}$ separates the non-empty sets $\cL$ and $\cR$. 
            \item Assume that Samson plays a $\neg$-move. Since $\Lambda$ is in strong positive normal form, the successor of the $\neg$-labeled node can only be labeled by a Sig-move choosing a proposition symbol. Since every valuation is empty and both swapped sets are non-empty, no proposition symbol separates them, and thus Delilah wins by choosing the successor node next (the invariant is not needed in this case, as Delilah wins immediately).
            \item Assume that Samson plays a $\lor$-move. 
            \begin{enumerate}
                \item If the first condition of the invariant holds, then there are two cases: if $1 (a)$ holds and $\cF$ does not contain any cycle, then Delilah follows the left-relevant model. If $1 (a)$ holds and $v$ is in a cycle, then Delilah chooses a disjunct which stays in a cycle. Note that if Samson moves all the left-relevant models out of the cycle, then Condition $2$ holds instead of Condition $1(a)$ after the move is finished. (In that case the left-relevant model lies at a node $v'$ outside the cycle, while suitable clocked models of the right remain at a node $v$ of the cycle from which $v'$ is reachable, lagging behind the left-relevant model exactly by the diamonds, boxes and variables of the connecting walk.) The reason is that $\cB$ consists of a clocked model for each possible cycle that Samson could construct such that at least one of the clocked models in $\cB$ satisfies Condition $2$. 
                \item If Condition $2$ of the invariant holds, then Delilah chooses a disjunct which follows the walk to $v'$ which satisfies Condition $2$. Note that after the move is played, Condition $1(b)$ might hold.
            \end{enumerate}
            After a $\lor$-move is played, at least one of the disjuncts is a relevant node.
            \item Assume that Samson plays an $X$-move.
            \begin{enumerate}
                \item If Condition $1(a)$ of the invariant holds, then Delilah does not challenge Samson.
                \item If Condition $2$ of the invariant holds, then Delilah does not challenge Samson and after that Condition $2$ must hold.
            \end{enumerate}
            Note that after an $X$-move is played, there is a new relevant node which is the successor of $v$.
            \item Assume that Samson plays a $\Diamond_{\geq m}$-move (resp. a $\Box_{< m}$-move). Then Delilah moves every model (that she can) toward the dead end. 
            Note that Delilah might win the game if Condition $2$ holds and $\cL$ is empty. Note that after a box-move is played, either Condition $3$ or $4$ might hold.
            \item 
            Assume that Samson plays a $\DiamondG_{\geq m}$-move and gives an $m$-global function $g$ over $\cL$. First, suppose that Condition $1(a)$ holds, and let $(P_{p_k}, w_j, \ell)$ be a left-relevant model at the node $v$. Then, according to the set $g(P_{p_k}, w_j, \ell)$, Delilah gives a partial $m$-global function $f$ defined on every clocked model $(P_h, w_i, \ell'') \in \cR$ that satisfies Condition $1(a)$, i.e., $h - i = p_k - j$, such that $f(P_h, w_i, \ell'')$ contains for every clocked model $(P_{p_k}, w_{j + t}, \ell) \in g(P_{p_k}, w_j, \ell)$ the clocked model $(P_h, w_{i+t}, \ell'') \in f(P_h, w_i, \ell'')$, where $t \in \{-j, \ldots, p_k-j\}$; this is well-defined and each set $f(P_h, w_i, \ell'')$ has exactly $m$ elements, since $0 \leq h - p_k \leq i + t \leq h$. Since the offsets $h - (i+t) = p_k - (j+t)$ and the clocks are preserved, Condition $1(a)$ holds again in the following position. If Condition $2$ holds, then Delilah gives a partial function that maps each clocked model $(P_h, w_i, \ell') \in \cR$ to a set of $m$ distinct clocked models that contains $(P_h, w_i, \ell')$ itself. 
        \end{itemize}
    \end{itemize}
    Therefore, either Condition $1(b)$, $3$ or $4$ will eventually hold in some later position of the game.
    If Condition $1(b)$ holds, then, by Remark \ref{remark: positive normal form}, Delilah has a winning strategy from that position. If Condition $3$ or $4$ holds, then Samson will eventually lose if Delilah follows the walk that satisfies Condition $3$ or $4$: she repeatedly chooses the labeled nodes along the walk, where Samson's moves are determined by the labels; the variables of the walk decrement the clock of the tracked clocked model (which the bounds of Conditions $3$ and $4$ cover), and the diamonds and boxes decrement its index, until Condition $4$ holds and the walk ends at a $\Box_{<m}$-labeled node where the tracked model lies at the dead end $w_0$; there Samson cannot provide the required $m$-successor function for the right side and loses. Now, clearly, since Delilah has a winning strategy in $\mathrm{FS}^\Pi_k(\bbA, \bbB)$, she also has a winning strategy in $\mathrm{FS}^\Pi_k(\bbP, \overline{\bbP})$ restricted to strong positive normal form, where $\bbP$ is the class of pointed $\Pi$-models that have the prime reachability property and $\overline{\bbP}$ is the class of pointed $\Pi$-models that do not have this property; note that $\bbA \subseteq \bbP$ and $\bbB \subseteq \overline{\bbP}$, since the unique path from $w_t$ to the node $w_0$ without out-neighbors has length $t$, and $p_k$ is prime while every $h \in H$ is composite. Thus, the prime reachability property is not definable by any program of $\GGMSC$ in strong positive normal form.
\end{proof}
We believe that the prime reachability property is also not expressible in the full GGMSC. However, this seems non-trivial to prove for the following reasons. In the proof of Proposition~4.10, Delilah chooses the models on the right so that they survive the bisimulation game for more than $k$ rounds, which forces Samson to construct a cycle in the syntax forest (where $k$ is the upper bound for the size of the syntax forest). Moreover, she maintains an invariant ensuring that, no matter which cycle Samson constructs, she can ``pump'' at least one of the models on the right along this cycle until it is bisimilar to the model on the left and has the same clock. 
Next, we explain why this strategy breaks down in the full GGMSC. By Lemma~2.2, we may still assume that programs are in strong negation normal form, but pushing negations inward turns disjunctions into conjunctions and global diamonds into global boxes. Hence, Samson can play $\land$-moves in induction rules, which allow him to split the models on the right between two successors. In this way, the model that Delilah would later need for pumping along a cycle may be sent to a successor that is not on the cycle. Delilah could try to counter this by using several models on the left that Samson has to keep up with, but it is not clear how to do this. Moreover, Samson can play global box moves, which allow him to move the models on the right to arbitrary nodes of the underlying Kripke models, which destroys the alignment between the positions of the models on the left and on the right that Delilah's invariant relies on. We could not find a suitable invariant or winning strategy for Delilah that can handle these moves, which motivated us to study programs in strong positive normal form.

\section{Model checking, satisfiability and expressivity results}

In this section, we study the expressive power of variants of the logic $\MSC$ and the combined and data complexity of the model checking problem for these logics. First, we link the asynchronous variants to the corresponding variants of modal computation logic $\MCL$ and the modal $\mu$-calculus (cf. Theorem \ref{thrm: AMSC=MCL=mu-fragment}). Then we show that, over words, $\MSC$ (as well as $\GMSC$ and $\GGMSC$) has the same expressive power as deterministic linear tape-bounded Turing machines (cf. Theorem \ref{thrm: lbas equiv 2-way GMSC}). 
By applying these expressivity results, we obtain the following under the standard semantics (resp. under the asynchronous semantics): both the combined and the data complexity of the model checking problem for $\GGMSC$, $\GMSC$ and $\MSC$ are PSPACE-complete (resp. PTIME-complete); see Theorem \ref{thrm: pspace model checking} (resp. Theorem \ref{thrm: ptime model checking}).
For the logic $\SC$, we show that under the standard semantics (resp. under the asynchronous semantics) the combined complexity of the model checking problem for $\SC$ is PSPACE-complete (resp. PTIME-complete); see Theorem \ref{thrm: SC PSPACE} (resp. Theorem \ref{thrm: SC async PTIME}). As a corollary, we obtain that the satisfiability problem for $\SC$ is PSPACE-complete under the standard semantics and NP-complete for the asynchronous variant of $\SC$.

Lastly, we show that the logic $\MSC$ is ``universal'' in the following sense. There is a universal reduction that gives a reduction from any recursively enumerable problem to the model checking problem for $\MSC$ when we pad the input; see Theorem \ref{thrm: computable reduction}. In particular, Theorem~\ref{thrm: time complexity reduction} shows that if a given Turing machine has a superlinear time- or space-complexity function $f$, then the universal reduction gives a reduction to the model checking problem for $\MSC$ that runs in time $\ordo(f)$.

\subsection{Asynchronous variants meet computation logic and \texorpdfstring{$\mu$-calculus}{mu-calculus}}

In this section, we conclude that over finite models $\GGMSC^A$ has the same expressive power as \textbf{global graded modal computation logic} (or GGMCL) and \textbf{the $\mu$-fragment of the global graded modal $\mu$-calculus} (or $\mathrm{GGML}^\mu_1$). 
We also obtain an analogous result for $\GMSC^A$ and $\MSC^A$.  

First, we introduce the global graded modal $\mu$-calculus and its fragments, and then global graded modal computation logic.
The global graded modal $\mu$-calculus extends $\GGML$ with the least and greatest fixed-point operators $\mu X$ and $\nu X$, where $X$ is a variable. The semantics for these operators is the same as for the modal $\mu$-calculus; see \cite{DBLP:books/el/07/BBW2007} for more details.
The logic $\GGML^\mu_1$ is the fragment of the global graded modal $\mu$-calculus that does not allow $\nu$-operators and allows negations only in front of proposition symbols. Analogously, we define $\GML^\mu_1$ and $\ML^\mu_1$. 
Modal computation logic (or $\MCL$) was introduced in \cite{HELLA2022104882} and is a fragment of static computation logic (or SCL) that was recently studied in \cite{jaakkola2022firstorderlogicselfreference}. Furthermore, SCL is a fragment of computation logic that was introduced in \cite{Kuusisto_2014} and shown to be Turing-complete. Here, we also consider more general variants of $\MCL$.

Let $\mathrm{LAB} = \{\, L_i \mid i \in \N\,\}$ be a set of \textbf{label symbols} and $\mathrm{CS} = \{\, C_{L_i} \mid L_i \in \mathrm{LAB}\,\}$ the set of \textbf{claim symbols}.
The set of $\Pi$-formulae of \(\GGMCL\) is given by the grammar
\[
\varphi ::= \bot \mid \top \mid p \mid C_L \mid \neg \varphi \mid \varphi \lor \varphi \mid \varphi \land \varphi \mid \Diamond_{\geq k} \varphi \mid \Box_{< k} \varphi \mid \DiamondG_{\geq k} \varphi \mid \BoxG_{< k} \varphi \mid L\varphi,
\]
where \(p \in \Pi\) is a proposition symbol, $k \in \Z_+$, \(C_L\) is a claim symbol and \(L\) is a label symbol.

The \textbf{reference formula} of \(C_L\) in a formula \(\varphi\) of \(\GGMCL\), denoted by \(\Rf_\varphi(C_L)\), is the unique (if it exists) subformula occurrence \(L \psi\) of \(\varphi\) such that there is a directed path from \(L\psi\) to \(C_L\) in the syntax tree of \(\varphi\), and \(L\) does not occur strictly between \(L\psi\) and \(C_L\) on that path.

Let \(\varphi\) be a $\Pi$-formula of \(\GGMCL\) and \((M,w)\) a pointed $\Pi$-model, where \(M = (W,R,V)\). We define the \textbf{unbounded evaluation game} \(\cG_\infty(M,w,\varphi)\) as follows. The game has two players, Abelard and Eloise. A \textbf{position} of the game is a tuple \((\bV, v, \psi )\), where $\bV \in \{\El, \Ab\}$ is the current verifier (and the other player is the current falsifier), $\psi$ is a subformula of $\varphi$ and \(v \in W\).
    
The game begins from the \textbf{initial position} \((\El, w, \varphi)\) and it is then played according to the following \textbf{rules}.
\begin{itemize}
        \item In a position \((\bV, v, \chi)\), where \(\chi \in \{\bot, \top\}\), the play of the game ends and the current verifier wins if \(\chi\) is $\top$. Otherwise, the falsifier wins.
        \item In a position \((\bV, v, p)\), where \(p\) is a proposition symbol, the play of the game ends and the current verifier wins if \(p \in V(v)\). Otherwise, the falsifier wins.
        \item In a position \((\bV, v, \neg \psi)\), the game continues from the position \((\bV', v, \psi)\), where $\bV' \in \{\El, \Ab\}\setminus \{\bV\}$. 
        \item In a position \((\bV, v, \psi \lor \theta)\), the current verifier chooses whether the game continues from the position \((\bV, v, \psi)\) or \((\bV, v, \theta)\).
        \item In a position \((\bV, v, \psi \land \theta)\), the current falsifier chooses whether the game continues from the position \((\bV, v, \psi)\) or \((\bV, v, \theta)\).
        \item In a position \((\bV, v, \Diamond_{\geq k} \psi)\), the current verifier chooses $k$ distinct nodes \(u_1, \ldots, u_k\) such that \((v, u_i) \in R\) for all $i \in [k]$. Then the current falsifier chooses an $i \in [k]$ and the game continues from the position \((\bV, u_i, \psi)\).
        \item In a position \((\bV, v, \Box_{< k} \psi)\), the game continues as in \((\bV, v, \Diamond_{\geq k} \psi)\), but the roles of the verifier and the falsifier are switched.
        \item In a position \((\bV, v, \DiamondG_{\geq k} \psi)\), the current verifier chooses $k$ distinct nodes \(u_1, \ldots, u_k\) from $W$. Then the current falsifier chooses an $i \in [k]$ and the game continues from the position \((\bV, u_i, \psi)\).
        \item In a position \((\bV, v, \BoxG_{< k} \psi)\), the game continues as in \((\bV, v, \DiamondG_{\geq k} \psi)\), but the roles of the verifier and the falsifier are switched.
        \item In a position \((\bV, v, C_L)\), if \(\Rf_\varphi(C_L)\) exists, then the next position is \((\bV, v, \Rf_\varphi(C_L))\). Otherwise, the game stops and neither player wins.
        \item If the position is \((\bV, v, L\psi)\), then the next position is \((\bV, v, \psi)\).
\end{itemize}
If the play continues for infinitely many rounds, then neither player wins.
If Eloise has a winning strategy in \(\cG_\infty(M,w,\varphi)\), then we write \(M,w \Vdash_\infty \varphi\) and say that $(M, w)$ \textbf{satisfies} $\varphi$.

Analogously, we obtain \textbf{graded modal computation logic} ($\GMCL$) by excluding all the global diamonds $\DiamondG_{\geq k}$ and global boxes $\BoxG_{< k}$ from the syntax of $\GGMCL$. Similarly, the formulae of \textbf{modal computation logic} (or $\MCL$) are the formulae of $\GMCL$, which may contain only diamonds of the type $\Diamond$ or boxes of the type $\Box$.

Next, we define some notions related to expressive power.
Let $\cL_1$ and $\cL_2$ be logics. The notation $\cL_1 \leq \cL_2$ means that \textbf{$\cL_1$ is contained in $\cL_2$}, i.e., for each formula $\varphi_1$ of $\cL_1$, there is a formula $\varphi_2$ of $\cL_2$ that accepts (or is satisfied by) precisely the same pointed Kripke models as $\varphi_1$.
Moreover, $\cL_1 \equiv \cL_2$ means that $\cL_1$ and $\cL_2$ \textbf{have the same expressive power}, i.e., $\cL_1 \leq \cL_2$ and $\cL_2 \leq \cL_1$ hold.
We write $\cL_1 \lneq \cL_2$ iff $\cL_1 \leq \cL_2$ and $\cL_1 \not\equiv \cL_2$.
Lastly, $\cL_1 \leq^{\mathrm{fin}} \cL_2$, $\cL_1 \lneq^{\mathrm{fin}} \cL_2$ and $\cL_1 \equiv^{\mathrm{fin}} \cL_2$ denote the corresponding notions restricted to the class of finite pointed Kripke models.

Now, it is easy to obtain the following theorem from the previous results in \cite{Kuusisto13} and \cite{jaakkola2022firstorderlogicselfreference}.
Below, we let $* \in \{\epsilon, \mathrm{G}, \mathrm{GG}\}$, where $\epsilon$ denotes the empty string.
\begin{theorem}\label{thrm: AMSC=MCL=mu-fragment}
    The following holds: $*\MCL \equiv *\MSC^A$.
    Moreover, over finite Kripke models, we have that
    \[
    *\ML^\mu_1 \equiv^{\mathrm{fin}} *\MSC^A \equiv^{\mathrm{fin}} *\MCL \lneq^{\mathrm{fin}} *\MSC.
    \]
\end{theorem}
\begin{proof}
    With a proof analogous to the proof of Proposition $7$ in \cite{Kuusisto13}, it follows that $*\ML^\mu_1 \leq^{\mathrm{fin}} *\MSC^A$ and $*\ML^\mu_1 \lneq^{\mathrm{fin}} *\MSC$. 
    A formula $\varphi$ of $*\MCL$ is in \emph{strong negation normal form} if the only negated subformulae of $\varphi$ are proposition symbols.
    From the proof of Lemma 5.1 in \cite{jaakkola2022firstorderlogicselfreference}, it follows that each formula of $*\MCL$ can be translated into a formula of $*\MCL$ that is in strong negation normal form. Formulae of $*\MCL$ in strong negation normal form can be translated in a straightforward way into $*\ML^\mu_1$. Indeed, by viewing claim symbols as second-order variables, we can translate a given formula \(\varphi\) of \(*\MCL\) into an equivalent formula of \(*\ML^\mu_1\) simply by replacing each label symbol \(L\) with \(\mu C_L\). Claim symbols without a reference formula are replaced by $\bot$; this preserves equivalence, since in strong negation normal form Eloise is always the verifier and cannot win a play that reaches such a claim symbol.
    
    Finally, we show that each $(\Pi, \cT)$-program $\Lambda$ of $*\MSC^A$ translates into a $\Pi$-formula of $*\MCL$. 
    At first glance, this might seem like a straightforward translation by viewing schema variables as label symbols and identifying each pair of rules $X(0) \colonminus \varphi$, $X \colonminus \psi$ in $\Lambda$ as the corresponding $*\MCL$-formula $X(\varphi \lor \psi)$. Then, starting from the disjunction consisting of the corresponding $*\MCL$-formulae of accepting predicates, one could recursively substitute label symbols with the corresponding $*\MCL$-formula. 
    If, during this recursive process, a label symbol $X$ appears twice in the formula, we replace the most recently added label symbol $X$ with the claim symbol $C_X$. However, this is a naive translation, as we might obtain a formula in which a claim symbol does not have a reference formula! 
    
    For example, consider the program $\Gamma$ below
    \[
    \begin{aligned}
        &X (0) \colonminus \varphi_X &X \colonminus Y \land Z \\
        &Y (0) \colonminus \varphi_Y &Y \colonminus \Diamond X \\
        &Z (0) \colonminus \varphi_Z &Z \colonminus X \lor Y,
    \end{aligned}
    \]
    where $X$ is an accepting predicate.
    By using the naive translation described above, the resulting formula is $\theta_\Gamma \colonequals X( \varphi_X \lor ( Y (\varphi_Y \lor \Diamond C_X) \land Z ( \varphi_Z \lor C_X \lor C_Y)) ) $. Now, $\theta_\Gamma$ is not equivalent to $\Gamma$, since $C_Y$ does not have a reference formula.
    However, the formula 
    \[
    X( \varphi_X \lor ( Y (\varphi_Y \lor \Diamond C_X) \land Z ( \varphi_Z \lor C_X \lor Y_{Z} ( \varphi_{Y} \lor \Diamond C_X))) ),
    \]
    which is obtained from $\theta_\Gamma$ by replacing $C_Y$ with $Y_Z(\varphi_{Y} \lor \Diamond C_X)$, is equivalent to $\Gamma$.

    Now, we give a proper translation, inspired by the example above. If $\Lambda$ does not have any accepting predicates, then the corresponding $*\MCL$-formula is $\bot$. 
    If the set of accepting predicates $\cA$ of $\Lambda$ is non-empty, we recursively construct the corresponding formula $\varphi_\Lambda$ of $*\MCL$ as follows. During the construction, we will use label symbols of the form $\{\, X_\bu \mid \bu \in \cT^*, X \in \cT\,\}$, where $\cT^*$ denotes the set of finite strings over $\cT$, to make sure that claim symbols refer to the correct formula, and that we do not end up in a situation where a claim symbol does not have a reference formula. Moreover, before the construction is ready, we will use label symbols as atomic subformulae, and the construction is ready when there are no label symbols appearing as an atomic subformula. In other words, before we obtain the final formula, the formula under construction can be seen as a ``pseudoformula''. At the beginning, the formula $\varphi_\Lambda$ is $\bigvee_{X \in \cA} X_\epsilon$, where $\epsilon$ is the empty string. 
    Then we recursively update the formula $\varphi_\Lambda$ as follows. If there is a label symbol $X_\bu$ as an atom, we replace it with the formula $X_\bu( \varphi_X \lor \psi_{X_\bu})$, where $\varphi_X$ is the body of the base rule of $X$ and $\psi_{X_\bu}$ is obtained by replacing each predicate $Y$ appearing in the body of the induction rule of $X$---depending on the current formula $\varphi_\Lambda$---as follows:
    \begin{enumerate}
        \item If there is no label symbol of the form $Y_\bt$ in the syntax tree of $\varphi_\Lambda$, then we replace $Y$ with the label symbol $Y_\epsilon$, where $\epsilon$ is the empty string.
        \item If there is a label symbol of the form $Y_\bt$ in the syntax tree of $\varphi_\Lambda$ and there is no (directed) path from $Y_\bt$ to $X_\bu$, then we replace $Y$ with $Y_{\bu X}$.
        \item If none of the above holds, then there is a closest label symbol of the form $Y_\bt$ from which $X_\bu$ is reachable in the syntax tree of $\varphi_\Lambda$. In that case, we replace $Y$ with the claim symbol $C_{Y_\bt}$.
    \end{enumerate}
    By the last condition, this recursive construction will terminate after a finite number of steps.
    
    Lastly, it is easy to show that the formula $\varphi_\Lambda$ obtained from the construction is equivalent to $\Lambda$. First of all, note that the semantics of $*\MCL$ and $*\MSC^A$ are almost identical. Secondly, it is easy to construct a winning strategy for Eloise in $\cG_\infty(M, w, \varphi_\Lambda)$ if Eloise has a winning strategy in $\cA\cG(M, w, \Lambda)$, and vice versa, by simply interpreting each claim symbol $C_{X_\bu}$ as the predicate $X$ (and label symbols are trivial to handle). Then, by a routine induction, one can show that $\Lambda$ and $\varphi_\Lambda$ are equivalent.
\end{proof}

\subsection{Relating linear tape-bounded Turing machines and MSC}

In this section, we give a logical characterization of deterministic linear tape-bounded Turing machines via $\MSC$, $\GMSC$ and $\GGMSC$.
Informally, linear tape-bounded Turing machines are Turing machines in which the length of the tape is linear in the length of a given input string.

\subsubsection{Turing machines and tape-bounded Turing machines}

We begin by defining Turing machines and linear tape-bounded Turing machines.
First, we fix some constant symbols.
We let $\br$, $\bl$ and $\bs$ denote the \textbf{direction symbols} for ``right'', ``left'', and ``stay'', respectively.

A deterministic \textbf{Turing machine} (or $\mathrm{TM}$) is a tuple 
\[
T = (Q, \Gamma, \blacksquare, E_\bl, \Sigma, q_0, \delta, A, R),
\]
where
%
%
    %
    $Q$ is a finite non-empty set of \textbf{states},
    %
    $\Gamma$ is a finite \textbf{tape alphabet},
    %
    %
    $\blacksquare \in \Gamma$ is the \textbf{blank symbol}, 
    $E_\bl$ is the \textbf{left end marker},
    %
    %
    $\Sigma \subseteq \Gamma \setminus \{\blacksquare, E_\bl\}$ is an \textbf{input alphabet},
    %
    $q_0$ is an \textbf{initial state},
    %
    $A \subseteq Q$ is a set of (halting) \textbf{accepting states},
    $R \subseteq Q$ is a set of (halting) \textbf{rejecting states} such that $A \cap R = \emptyset$. 
%
Finally, 
\[
\delta \colon Q \times \Gamma \to Q \times \Gamma \times \{ \bl, \bs, \br\}
\]
is a \textbf{transition function} which is restricted such that it cannot print other symbols over the left end marker and cannot move to the left of the left end marker. In other words, for every $q \in Q$, we require that
$\delta(q, E_\bl) \in Q \times \{E_\bl\} \times \{\bs, \br\}$. The set $A \cup R$ forms the set of \textbf{halting states} of $T$.

Next, we define how $\TM$s compute over strings.
A \textbf{configuration} of a $\TM$ is a tuple consisting of the contents of the work tape, the current position of the head on the work tape, and the current state. The work tape is one-way infinite to the right. 
A \textbf{run} of $T$ over an input string $w \in \Sigma^*$ is a configuration sequence starting from the initial configuration defined as follows: the leftmost cell of the work tape contains $E_\bl$ followed by the input string and then infinitely many blank symbols, the head is on the cell immediately to the right of $E_\bl$ (i.e., on the first symbol of $w$ if $w$ is non-empty, and on a blank symbol otherwise), and the current state is $q_0$.
From the current configuration, new configurations are obtained as follows. Let $a$ be the symbol in the current position of the head on the work tape, and let $q$ be the current state. If $\delta(q, a) = (q', a', \bd)$, then the new configuration is a tuple consisting of the state $q'$, the head moved one step in the direction given by $\bd$ and $a$ replaced in the old position of the head by $a'$. We say that $T$ \textbf{accepts} (resp. \textbf{rejects}) a string $\bs \in \Sigma^*$ if it enters an accepting (resp. a rejecting) state during a run with the input $\bs$. If $T$ accepts or rejects a string during a run, then it \emph{halts}, meaning that the configuration is not updated anymore, i.e., the last configuration of such a run is the first configuration in which $T$ is in an accepting or rejecting state.  
Moreover, we say that $T$ recognizes a language $L \subseteq \Sigma^*$ if $L = \{\, \bs \in \Sigma^* \mid \text{$T$ accepts $\bs$}\}$. For more concepts of Turing machines, see \cite{papadimitriou2003computational}.

Given a $k \in \N$, a deterministic \textbf{$k$-bounded Turing machine} (or a linear tape-bounded $\mathrm{TM}$) is a tuple 
\[
T' = (k, Q, \Gamma, \blacksquare, E_\bl, E_\br, \Sigma, q_0, \delta, A, R),
\]
where
$(Q, \Gamma, \blacksquare, E_\bl, \Sigma, q_0, \delta, A, R)$ is a Turing machine, $E_\br \in \Gamma$, with $E_\br \notin \Sigma$, is the \textbf{right end marker}, and 
$\delta$ is restricted as follows (in addition to the restrictions in the definition of TMs): it cannot print other symbols over the right end marker and cannot move to the right of the right end marker. More formally, for every $q \in Q$, we require that
$\delta(q, E_\br) \in Q \times \{E_\br\} \times \{\bs, \bl\}$. 
Linear tape-bounded $\TM$s compute over strings analogously to $\TM$s, but the initial configuration of $T'$ over an input string $\bs = s_1 \cdots s_n \in \Sigma^*$ is defined as follows. The current state is $q_0$, the work tape consists of the string $E_\bl \bs \blacksquare^k E_\br$ followed by an infinite sequence of blank symbols, 
and the head of the work tape is on the cell immediately to the right of $E_\bl$.

\subsubsection{Languages accepted by programs}\label{section: language by MSC}

From now on, we do not make a distinction between proposition symbols and alphabet symbols. 
Let $p_{\blacksquare}$, $p_\bl$ and $p_\br$ be three distinct proposition symbols.
For this section, we fix an arbitrary alphabet $\Pi$ such that $\Pi \cap \{p_{\blacksquare}, p_\bl, p_\br\} = \emptyset$. We let $\Pi_+ \colonequals \Pi \cup \{p_{\blacksquare}, p_\bl, p_\br\}$.

We can make programs reject by associating them with rejecting predicates. A $(\Pi, \cT)$-program of $\GGMSC$ can be associated with a set of \textbf{rejecting predicates} $\cR \subseteq \cT$ that is disjoint from the set of accepting predicates of the program. Now, such a program accepts (resp. rejects) a pointed $\Pi$-model $(M, w)$ if during the computation there is an $n \in \N$ and an accepting predicate (resp. a rejecting predicate) $X$ such that $M, w \models X^n$ and there is no rejecting predicate (resp. accepting predicate) $Y$ and no $m \leq n$ with $M, w \models Y^m$.

Given an $\ell \in \N$ and a string $\bs \in \Pi^*$, its \textbf{$\ell$-extended word model} $M^\ell_\bs = ([0;\abs{\bs}+ \ell + 1], R, V)$ is a $\Pi_+$-model, where $R$ is the symmetric closure of the successor relation over integers in $[0;\abs{\bs} + \ell + 1]$ and for each $p \in \Pi$ and for each $i \in [\abs{\bs}]$, we define $V(i) = \{p\}$ if $\bs(i) = p$, and lastly for each $\abs{\bs} < j \leq \abs{\bs} + \ell$, we have $V(j) = \{p_\blacksquare\}$, $V(0) = \{p_\bl\}$ and $V(\abs{\bs}+\ell+1) = \{p_\br\}$.

Given a $\Pi_+$-program $\Lambda$ of $\GGMSC$, a $k \in \N$, and a string $\bs \in \Pi^*$, we say that $\Lambda$ \textbf{$k$-accepts (resp. $k$-rejects) $\bs$} if its pointed $k$-extended word model $(M^k_\bs, 1)$ is accepted (resp. rejected) by $\Lambda$.
Moreover, given a language $L \subseteq \Pi^*$, we say that $\Lambda$ \textbf{$k$-recognizes $L$} if $\Lambda$ precisely $k$-accepts the strings in $L$. On the other hand, we say that a language $L \subseteq \Pi^*$ is \textbf{recognized} by $\Lambda$ if there is a $k \in \N$ such that $\Lambda$ precisely $k$-accepts the strings in $L$.

\begin{example}
Given a word $ppqqpp \in \{p,q\}^*$, 
its $3$-extended word model is drawn below.
\begin{center}
\begin{tikzpicture}[scale=0.8, every node/.style={scale=0.8}, nodes={draw, circle}, <-]

\node (X0) {$p_\bl$};
\node[right of=X0, node distance=1.5cm, fill=green!20] (X1) {$p$};
\node[right of=X1, node distance=1.5cm, fill=green!20](X2) {$p$};
\node[right of=X2, node distance=1.5cm, fill=blue!20](X3) {$q$};
\node[right of=X3, node distance=1.5cm, fill=blue!20](X4) {$q$};
\node[right of=X4, node distance=1.5cm, fill=green!20](X5) {$p$};
\node[right of=X5, node distance=1.5cm, fill=green!20](X6) {$p$};
\node[right of=X6, node distance=1.5cm](X7) {$p_\blacksquare$};
\node[right of=X7, node distance=1.5cm](X8) {$p_\blacksquare$};
\node[right of=X8, node distance=1.5cm](X9) {$p_\blacksquare$};
\node[right of=X9, node distance=1.5cm](X10) {$p_\br$};

\path [-stealth, thick]
(X0) edge [bend left=40] node [draw=none] {} (X1)
(X1) edge [bend left=40] node [draw=none] {} (X2)
(X2) edge [bend left=40] node [draw=none] {} (X3)
(X3) edge [bend left=40] node [draw=none] {} (X4)
(X4) edge [bend left=40] node [draw=none] {} (X5)
(X5) edge [bend left=40] node [draw=none] {} (X6)
(X6) edge [bend left=40] node [draw=none] {} (X7)
(X7) edge [bend left=40] node [draw=none] {} (X8)
(X8) edge [bend left=40] node [draw=none] {} (X9)
(X9) edge [bend left=40] node [draw=none] {} (X10);
%

\path [-stealth, thick]
(X1) edge [bend left=40] node [draw=none] {} (X0)
(X2) edge [bend left=40] node [draw=none] {} (X1)
(X3) edge [bend left=40] node [draw=none] {} (X2)
(X4) edge [bend left=40] node [draw=none] {} (X3)
(X5) edge [bend left=40] node [draw=none] {} (X4)
(X6) edge [bend left=40] node [draw=none] {} (X5)
(X7) edge [bend left=40] node [draw=none] {} (X6)
(X8) edge [bend left=40] node [draw=none] {} (X7)
(X9) edge [bend left=40] node [draw=none] {} (X8)
(X10) edge [bend left=40] node [draw=none] {} (X9);
\end{tikzpicture}   
\end{center}
\end{example}

\subsubsection{Translations}

Now, we are ready to show that $\GGMSC$, $\GMSC$ and $\MSC$ recognize precisely the same languages as linear tape-bounded $\TM$s. In the theorem below, we also consider rejection and assume that programs can be associated with rejecting predicates. 

\begin{theorem}\label{thrm: lbas equiv 2-way GMSC}
Linear tape-bounded $\TM$s recognize precisely the same languages as $\GGMSC$, $\GMSC$ and $\MSC$.

More precisely, the following holds:
\begin{enumerate}
    \item Given a $\Pi_+$-program $\Lambda$ of $\GGMSC$ and $k \in \N$, there is a deterministic $k$-bounded $\TM$ $T_\Lambda$ such that the following holds: $\bw$ is $k$-accepted (resp. $k$-rejected) by $\Lambda$ iff $\bw$ is accepted (resp. rejected) by $T_\Lambda$. 
    \item Given a deterministic $k$-bounded $\TM$ $T$, there is a program $\Lambda_T$ of $\MSC$ such that the following holds: $\bw$ is accepted (resp. rejected) by $T$ iff $\bw$ is $k$-accepted (resp. $k$-rejected) by $\Lambda_T$.
\end{enumerate}
\end{theorem}
\begin{proof}
    First, let $\Lambda$ be a $(\Pi_+, \cT)$-program of $\GGMSC$ and $k \in \N$. 
    We can assume that the modal depth of each base rule is $0$ and the modal depth of each induction rule is at most $1$, since $\Lambda$ can be easily translated into an equivalent program in that form. An analogous result was shown for $\GMSC$ in \cite{ahvonen2024logical} (Lemma B.1) and for $\MSC$ in \cite{10.1093/logcom/exae087} (Theorem $5.4$).
    Furthermore, we can assume that proposition symbols do not appear in the bodies of the induction rules (since an equivalent program is easy to obtain). 
    
    Informally, we construct a linear tape-bounded $\TM$ $T_{\Lambda}$ that simulates the global configurations of $\Lambda$ in a periodic fashion. The machine $T_\Lambda$ has a corresponding position on the work tape for each node of the word model under simulation.  
    To compute the next global configuration of $\Lambda$ from its previous global configuration at a single node, $T_\Lambda$ scans which head predicates are true at the node and its out-neighbors. After the scan, $T_\Lambda$ writes to the corresponding position on the work tape which predicates are true. It is easy but tedious to implement $T_{\Lambda}$, so we only informally describe how $T_\Lambda$ works.
    \begin{enumerate}
        \item The machine $T_\Lambda$ is $k$-bounded with the tape alphabet $\wp(\cT) \cup \Pi \cup \{\blacksquare, E_\bl, E_\br\}$ and the machine has a single accepting state and a single rejecting state.
        \item The global configuration in round $n=0$ is computed as follows. The head of $T_{\Lambda}$ scans the whole work tape $E_\bl \bs \blacksquare^k E_\br$, where $\bs = s_1 \ldots s_m$, from left to right and marks during the scan which head predicates are true at each position as follows. If $s_i$ is the current alphabet symbol (or the blank symbol) where the head is, then it is replaced by $\cT_i$, where $\cT_i$ is the set of head predicates true at the node $i$ in the extended word model $M^k_\bs$ in round $0$. Since the modal depth of each base rule is zero, there is no need for scanning all the tape symbols to the left or right of each position. 
        However, note that $T_{\Lambda}$ cannot overwrite the left end marker or right end marker; thus, instead of rewriting the tape symbols in the position $0$ and position $\abs{\bs}+k+1$, $T_{\Lambda}$ keeps track of the current local configuration of these positions in its state.
        After updating the whole work tape, the head moves back to the left end marker. 
        \item 
        Similarly, to compute the global configuration of $\Lambda$ in round $n+1$, $T_{\Lambda}$ proceeds as follows. Assume that the head is at the left end marker, i.e., in the position $0$.
        \begin{itemize}
            \item 
            Let $m$ denote the greatest counting threshold that appears in $\Lambda$. Recall that a multiset $\cM(S)$ over a set $S$ is a function of the form $S \to \N$, and an $m$-bounded multiset $\cM^m(S)$ over $S$ is a multiset $f$ such that $f(x) \leq m$ for each $x \in S$. 
            The machine $T_{\Lambda}$ scans the whole content of the tape and records the $m$-bounded multiset of the tape symbols that appear on the tape (including the current local configurations in the first and last positions). After that, $T_\Lambda$ scans the content of the cell in position $1$. Then $T_\Lambda$ updates the current local configuration of the position $0$ in its state according to the old local configuration of the position $0$, the recorded multiset and the tape symbol in the position $1$; this is possible since the modal depth of each induction rule is at most $1$. At this point, $T_\Lambda$ also stores the old local configuration of the position $0$ in the state of $T_{\Lambda}$.    
            \item After updating the local configuration in the position $0$ in the state of $T_\Lambda$, it scans the content of the cell in the position $2$. After that, based on the multiset scanned in the previous step, the old local configuration in the position $0$, the tape symbol in the position $1$ and the tape symbol in the position $2$, $T_{\Lambda}$ updates the tape symbol in the position $1$ w.r.t. the induction rules of $\Lambda$. Again, we record the old tape symbol of the position $1$ in the state of $T_{\Lambda}$ and we can forget the old local configuration in positions $0$ and $2$.
            \item This process continues in an analogous way until we reach the right end marker and update its local configuration in the state of $T_\Lambda$. Note that we do not need to scan to the right of the right end marker. 
            \item After computing a new local configuration for each position, we have computed the global configuration in round $n+1$, and the head moves back to the left end marker to start a new cycle; when a new cycle starts, $T_\Lambda$ only needs to store the local configurations of the left and right end markers.
        \end{itemize}
        \item If the tape symbol of the position $1$ (i.e., the position to the right of the left end marker) includes an accepting predicate but no rejecting predicate (resp. a rejecting predicate but no accepting predicate), then $T_{\Lambda}$ enters an accepting state (resp. a rejecting state). If the tape symbol of the position $1$ includes both an accepting and a rejecting predicate, then $T_\Lambda$ keeps running and never halts. 
    \end{enumerate}
    Clearly, the constructed linear tape-bounded $\TM$ accepts (resp. rejects) $\bw$ iff $\bw$ is $k$-accepted (resp. $k$-rejected) by $\Lambda$.

    For the converse direction, let $T = (k, Q, \Gamma, \blacksquare, E_\bl, E_\br, \Pi, q_0, \delta, A, R)$ be a linear tape-bounded Turing machine. We may assume w.l.o.g. that $T$ moves its head to the cell immediately to the right of $E_\bl$ before entering a halting state; this does not change the language recognized by $T$.
    We will construct a $\Pi_+$-program $\Lambda_T$ of $\MSC$ that $k$-accepts (resp. $k$-rejects) each word accepted (resp. rejected) by $T$. 
    Informally, for each $a \in \Gamma$, the head predicate $X_a$ records the current tape symbol at the node and therefore we construct the program such that in each round precisely one of these predicates is true at each node. For each $q \in Q$, $X_q$ records the current state of the node and the position of the head, i.e., we construct the program such that in each round, precisely one of the predicates $X_q$ is true at a single node. 
    The head predicates $X_\bl$ and $X_\br$ record the positions to the left and right of the current position of the head. 

    It is easy to define rules for the described head predicates, but we will formally define them below.
    For each $a \in \Gamma$, we define a head predicate $X_a$ and the rules as follows:
    \[
    X_a (0) \colonminus \varphi_a \qquad X_a \colonminus \big( X_a \land \neg \bigvee_{q \in Q} X_q \big) \lor \bigvee_{\substack{(s, b) \in Q \times \Gamma \\ \delta(s, b) \in Q \times \{a\} \times \{\bl, \bs, \br\}}} (X_s \land X_b), 
    \]
    where $\varphi_a \colonequals p_a$ if $a \in \Pi_+$, and otherwise we define $\varphi_a \colonequals \bot$.
    For each state $q \in Q$, we define a head predicate $X_q$ as follows. For $X_{q_0}$, we set ${X_{q_0} (0) \colonminus \Diamond p_\bl}$ and for other states $q \in Q \setminus \{q_0\}$, we set $X_q (0) \colonminus \bot$. 
    Before we define the induction rules, we define for each $b \in \Gamma$ and for each $s \in Q$ the following auxiliary formulae, which are used to ``tell the head where to move'':
    \[
    \varphi_{(s, b, \bl)} \colonequals X_\bl \land \Diamond (X_s \land X_b), \qquad \varphi_{(s, b, \bs)} \colonequals (X_s \land X_b), \qquad \varphi_{(s, b, \br)} \colonequals X_\br \land \Diamond (X_s \land X_b). 
    \]
    Now, the induction rules for each $q \in Q$ are defined as follows:
    \[
    X_q \colonminus \bigvee_{\substack{(s, b, \bd) \in Q \times \Gamma \times \{\bl, \bs, \br\} \\ \delta(s,b) \in \{q\} \times \Gamma \times \{\bd\}}} \varphi_{(s, b, \bd)}. 
    \]
    Now, for each $\bd \in \{\bl, \bs, \br\}$, we let 
    \[
    \psi_\bd \colonequals \bigvee_{\substack{(s, b) \in Q \times \Gamma \\ \delta(s,b) \in Q \times \Gamma \times \{\bd\}}} \varphi_{(s, b, \bd)}, \qquad \chi_\bd \colonequals \bigvee_{\substack{(s, b) \in Q \times \Gamma \\ \delta(s, b) \in Q \times \Gamma \times \{\bd\}}} (X_s \land X_b)
    \]
    Lastly, we define
    \[
    \begin{aligned}
    &X_\bl (0) \colonminus p_\bl &&X_\bl \colonminus \overbrace{( X_\bl \land \Diamond \psi_\bs)}^{\text{stay}} \lor \overbrace{\chi_\br}^{\text{move to the right}} \lor \overbrace{\Big(\neg \bigvee_{q \in Q} X_q \land \Diamond \psi_\bl \Big)}^{\text{move to the left}} \\
    &X_\br (0) \colonminus \neg p_\bl \land  \Diamond\Diamond p_\bl &&X_\br \colonminus ( X_\br \land \Diamond \psi_\bs)  \lor \Big(\neg \bigvee_{q \in Q} X_q \land \Diamond \psi_\br \Big) \lor \chi_\bl.
    \end{aligned}
    \]
    Intuitively, in the induction rules for $X_\bl$ and $X_\br$, the first disjunct is true in the case where the head stays in the same position, the second disjunct handles the case where the head moves to the right, and the last disjunct takes care of the case where the head moves to the left.   
    
    The set of accepting predicates (resp. rejecting predicates) consists of the head predicates $X_q$, where $q \in A$ (resp. $q \in R$).
    By a routine induction, one can show that the constructed program $k$-accepts (resp. $k$-rejects) $\bw$ iff $\bw$ is accepted (resp. rejected) by $T$.
\end{proof}

\subsection{Model checking and satisfiability}\label{sec: model checking}

In this subsection, we study the computational complexity of the model checking problem for $\SC$, $\MSC$, $\GMSC$ and $\GGMSC$.

We start by considering asynchronous variants.

\begin{theorem}\label{thrm: ptime model checking}
    Both the combined and data complexity of the model checking problem for $\GGMSC^A$, $\GMSC^A$ and $\MSC^A$ are $\mathrm{PTIME}$-complete.
\end{theorem}
\begin{proof}
    As already shown in \cite{HELLA2022104882} in the proof of Proposition 8.5, the data complexity of the model checking problem for $\MCL$ is PTIME-hard. The proof in that paper is based on the proof of Proposition 8.2, which shows that the model checking problem for a fixed formula of the modal $\mu$-calculus is PTIME-hard. Therefore, by Theorem \ref{thrm: AMSC=MCL=mu-fragment}, it follows that the data complexity of the model checking problem for $\GGMSC^A$, $\GMSC^A$ and $\MSC^A$ is PTIME-hard, and thus the combined complexity of the model checking problem for these logics is also PTIME-hard.
    
    To see that model checking is in $\mathrm{PTIME}$, it is straightforward to simulate the asynchronous semantic game $\cA\cG(M, w, \Lambda)$ for a program $\Lambda$ of $\GGMSC^A$, $\GMSC^A$ or $\MSC^A$ by an alternating $\mathrm{LOGSPACE}$ machine; the machine simply keeps track of the current position of the game. 
    In more detail:
    \begin{itemize}
        \item The machine keeps track of the position of the game, i.e., the work tape keeps a pointer to the current node and another pointer to the current subformula of the program. 
        \item The machine simulates each rule in a standard way: informally, each existential state corresponds to the choice of the current verifier, while each universal state corresponds to the choice of the current falsifier. Also, each state tells whether Eloise is the current falsifier or verifier. Proposition symbols, schema variables, constant symbols and Boolean connectives are easy to handle. 
        \item The most involved part is to simulate counting diamonds and boxes. We go through how to simulate a subformula of the type $\Diamond_{\geq k} \varphi$, since the other cases are analogous. The machine goes through the out-neighbors of the current node one by one by using the implicit linear order of the out-neighbors given by the input model. For each out-neighbor of the current node, the verifier tells whether the node satisfies $\varphi$ or not; the machine stores (in binary) on the work tape the number of neighbors that the verifier claims satisfy $\varphi$. If the verifier claims that an out-neighbor $u$ satisfies $\varphi$, the current falsifier can challenge the current verifier and the machine proceeds to verify the formula $\varphi$ at $u$. If the falsifier does not challenge any of the claims made by the verifier and the machine runs out of out-neighbors, then the machine checks whether the number $m$ of nodes claimed to satisfy $\varphi$ is greater than or equal to $k$. After that, the machine may enter a halting state, which is an accepting state if $m \geq k$ and a rejecting state if $m < k$. Note that the number of out-neighbors is always bounded by the number of nodes in the input model, and therefore, since $m$ is encoded in binary, the space used on the work tape is always logarithmic.
    \end{itemize}
\end{proof}

Next, we consider programs with the standard semantics.

\begin{theorem}\label{thrm: pspace model checking}
    Both the combined and data complexity of the model checking problem for $\GGMSC$, $\GMSC$ and $\MSC$ are $\mathrm{PSPACE}$-complete. 
\end{theorem}

\begin{proof}
    We show that the data complexity of the model checking problem for $\MSC$ is PSPACE-hard by showing that the data complexity of the membership problem for linear tape-bounded $\TM$s is PSPACE-hard. Recall that the membership problem for a linear tape-bounded $\TM$ asks whether a given string is in the language accepted by the given linear tape-bounded $\TM$. It is a well-known fact that the combined complexity of the membership problem for linear tape-bounded $\TM$s is PSPACE-complete, but, to our knowledge, the data complexity of the membership problem for linear tape-bounded $\TM$s is not explicitly stated anywhere. Therefore, we give a reduction from the quantified Boolean formulae problem (or QBF) to the membership problem for linear tape-bounded $\TM$s with a fixed machine. 

    Let $\varphi = Q_1 x_1 \cdots Q_n x_n \psi$ be an instance of QBF, where $\psi$ is a quantifier-free Boolean formula. We can encode $\varphi$ into a string $\bw_\varphi$, which is linear in the size of $\varphi$ and consists of the following three parts. The first part $\bw_1$ is a binary string that stores the current interpretation of the variables that appear in $\psi$, the second part $\bw_2$ records $\varphi$, and the third part $\bw_3$ is used to evaluate $\psi$ under the interpretation expressed in $\bw_1$.
    Initially, $\bw_1$ is the zero string. The machine $T_{\mathrm{QBF}}$ repeatedly evaluates $\psi$ under the interpretation expressed in $\bw_1$ (using $\bw_2$ to reset $\bw_3$). 
    Intuitively, $T_{\mathrm{QBF}}$ evaluates the formula depth-first and stops early.
    We let $v$ denote the obtained truth value of $\psi$ under the interpretation $\bw_1$. Then the machine scans the indices $i = n, n-1, \ldots$ as follows (the type of the quantifier $Q_i$ is read from $\bw_2$). 
    \begin{itemize}
        \item If the bit of $x_i$ (in $\bw_1$) is $0$ and either $Q_i$ is universal and $v$ is true, or $Q_i$ is existential and $v$ is false, then the machine sets the bit of $x_i$ (in $\bw_1$) to $1$, sets the bits of $x_{i+1}, \ldots, x_n$ to $0$, and evaluates $\psi$ again (and proceeds to the index $n$). In every other case, $v$ is the truth value of $Q_i x_i \cdots Q_n x_n\, \psi$ under the interpretation of $x_1, \ldots, x_{i-1}$ expressed in $\bw_1$, and the scan proceeds to the index $i - 1$. 
        \item If the bit of $x_i$ is $1$ (in $\bw_1$), then the machine $T_{\mathrm{QBF}}$ proceeds to the index $i-1$. This is because the branch $x_i = 1$ is entered only when the branch $x_i = 0$ did not determine the value of $Q_i$. 
    \end{itemize}
     If the scan proceeds past the index $1$, then $v$ is the truth value of $\varphi$, and $T_{\mathrm{QBF}}$ accepts iff $v$ is true.

    We have now shown that the data complexity of the membership problem for linear tape-bounded $\TM$s is PSPACE-hard.
    Thus, by Theorem \ref{thrm: lbas equiv 2-way GMSC}, the data complexity of the model checking problem is also PSPACE-hard for $\MSC$, $\GMSC$ and $\GGMSC$. Furthermore, clearly the combined complexity of the model checking problem for $\GGMSC$ is in PSPACE: given a $\Pi$-program $\Lambda$ of $\GGMSC$ and a pointed $\Pi$-model $(M, w)$, we simulate $\Lambda$ in $(M, w)$ by tracking its global configurations at each node and accept if an accepting predicate becomes true at $w$. Note that a single round of the simulation can be computed in polynomial time. Therefore, since storing the global configuration requires only a polynomial amount of space and we only need to simulate \(\Lambda\) for at most an exponential number of rounds, this simulation can be done using polynomial space.
\end{proof}

We conclude this section with results on the combined complexity of \(\SC\) and its asynchronous variant \(\SC^A\). Note that in both cases, the data complexity of the model checking problem is trivial, because if we fix a program of \(\SC\), then we also fix the underlying set of proposition symbols, which means that there is only a fixed number of possible inputs. Considering related work on model checking $\SC$, in the paper \cite{10.1093/logcom/exag011}, the logic $\SC$ was linked to Boolean networks, and it is folklore that the reachability problem for Boolean networks is a PSPACE-complete problem; see, for example, \cite{goles2015pspacecompletenessmajorityautomatanetworks, perrot2024complexitybooleanautomatanetworks}. The reachability problem for Boolean networks asks the following: given an input configuration $\bi$ and a given configuration $\bc$, does the studied Boolean network reach the configuration $\bc$ with the input $\bi$? Thus, the PSPACE-completeness of the combined complexity of the model checking problem for $\SC$ (under the standard semantics) could be obtained via Boolean networks, but to keep the paper self-contained, we prove the result here without referring to Boolean networks. 
\begin{theorem}\label{thrm: SC async PTIME}
    The combined complexity of the model checking problem for $\SC^A$ is $\mathrm{PTIME}$-complete.
\end{theorem}
\begin{proof}
    The model checking problem for $\SC^A$ is in PTIME by Theorem \ref{thrm: ptime model checking}. We show that it is PTIME-hard by giving a LOGSPACE reduction from the Boolean circuit value problem to the model checking problem for $\SC^A$. Note that a translation from Boolean circuits to programs of $\SC$ was already (implicitly) given in \cite{dist_circ_mfcs}, but we go through the details here for $\SC^A$.
    
    A Boolean circuit (with one output gate) is an acyclic graph defined as follows. It contains exactly one node whose out-degree is zero, each node whose in-degree is zero is labeled with a variable from the set $\{\, x_i \mid i \in \N \,\}$, other nodes are labeled with $\land$ or $\lor$ symbols, and only nodes with in-degree one can be labeled with $\neg$. A Boolean circuit $C$ with $n$ distinct variables induces a Boolean function $f_C \colon \{0,1\}^n \to \{0,1\}$ in a natural way; see \cite{vollmer} for more details. Nodes of Boolean circuits are called gates.
    The Boolean circuit value problem asks: given a Boolean circuit $C$ with $n$ variables and an input $\bb \in \{0,1\}^n$ for its variables, does $C$ output $1$? 

    Now, we give a LOGSPACE reduction. 
    Given a Boolean circuit $C$ with $n$ variables $x_1, \ldots, x_n$, and an input $\bb \in \{0,1\}^n$, we construct a program $\Lambda_C$ of $\SC^A$ and a model $M_\bb$ such that $C$ outputs $1$ iff $\Lambda_C$ accepts $M_\bb$. For each variable $x_i$, we define a proposition symbol $p_i$ that acts as an input. The model $M_\bb$ is the model of propositional logic over the vocabulary $\{p_1, \ldots, p_n\}$, where $p_i \in V(w)$ iff the $i$th bit of $\bb$ is $1$. For each input gate $I$ labeled with $x_i$, we define corresponding rules $X_I (0) \colonminus p_i$, $X_I \colonminus p_i$. Then for each non-input gate $G$ with a label $\star \in \{\lor, \land\}$ and incoming gates $G_1, \ldots, G_k$, we define corresponding rules 
    \[
    X_G (0) \colonminus \bot,\qquad X_G \colonminus X_{G_1} \star \cdots \star X_{G_k}.
    \]
    Also, for each non-input gate $G'$ with the label $\neg$ and an incoming gate $G''$, we define $X_{G'} (0) \colonminus \bot$, $X_{G'} \colonminus \neg X_{G''}$. If $O$ is the output gate of $C$, then $X_O$ is the accepting predicate of $\Lambda_C$.
    These rules can clearly be constructed from the input in LOGSPACE. It is also easy to verify that the constructed program works.
\end{proof}

We then consider \(\SC\) programs with the standard semantics.

\begin{theorem}\label{thrm: SC PSPACE}
The combined complexity of the model checking problem for \(\mathrm{SC}\) is $\mathrm{PSPACE}$-complete.
\end{theorem}
\begin{proof}
The upper bound follows from Theorem \ref{thrm: pspace model checking}, since \(\mathrm{SC}\) is a fragment of \(\MSC\). For the lower bound, we note that the proof of Theorem \ref{thrm: pspace model checking} shows that the model checking problem for \(\MSC\) is PSPACE-hard already over extended words. Consider now a \((\Pi,\mathcal{T})\)-program \(\Lambda\) of \(\MSC\). We claim that given a finite extended word \(\bw\) of length \(n\), we can construct in polynomial time a model \(M_\bw\) and a program \(\Lambda_\bw\) of \(\mathrm{SC}\) such that \(M_\bw \models \Lambda_\bw\) iff \(\bw\) is accepted by \(\Lambda\). The idea is that for every \(p \in \Pi\) and \(X \in \mathcal{T}\), we introduce \(n\) fresh proposition symbols \(p_1,\dots,p_n\) and schema variables \(X_1,\dots,X_n\) respectively. Here \(p_i\) and \(X_i\) intuitively mean that ``\(p\) is true at the \(i\)th letter of \(\bw\)'' and ``\(X\) is true at the \(i\)th letter of \(\bw\)'', respectively. With this intuition, it is straightforward to construct \(M_{\bw}\) and \(\Lambda_\bw\) from \(\bw\) and \(\Lambda\). As an example, a rule of the form
\[X \colonminus p \lor (Y \land \Diamond X)\]
is mapped to
\begin{align*}
& X_1 \colonminus p_1 \lor (Y_1 \land X_2)\\
& X_2 \colonminus p_2 \lor (Y_2 \land (X_1 \lor X_3))\\
& X_3 \colonminus p_3 \lor (Y_3 \land (X_2 \lor X_4))\\
& \vdots \\
& X_n \colonminus p_n \lor (Y_n \land X_{n-1})\\
\end{align*}
Note that this results in only a polynomial blow-up, i.e., \(|\Lambda_{\bw}| = \mathcal{O}(|\bw||\Lambda|)\).
\end{proof}

We observe that the satisfiability problem for $\SC^A$ is $\mathrm{NP}$-complete.
The lower bound follows from the NP-hardness of propositional satisfiability, and the upper bound follows from the observation that, given an \(\SC^A\) program, we can guess an input bit string and verify whether the program accepts it or not in polynomial time (Theorem \ref{thrm: SC async PTIME}).
\begin{theorem}\label{thrm: SC async sat NP}
    The satisfiability problem for \(\SC^A\) is $\mathrm{NP}$-complete.
\end{theorem}

Likewise, the satisfiability problem for
$\SC$ is $\mathrm{PSPACE}$-complete. The lower bound follows easily from Theorem \ref{thrm: SC PSPACE} with the following reduction: from a given program $\Lambda$ of $\SC$ and an input model $M$, we build a program $\Lambda_M$ of $\SC$ that simulates $\Lambda$ over $M$, no matter what the input model for $\Lambda_M$ is. The upper bound is also easy to obtain. Given an \(\SC\) program, we can systematically enumerate all possible input bit strings using polynomial space and, for each input, verify whether the program accepts it. Since the combined complexity of model checking for \(\SC\) is in PSPACE (Theorem \ref{thrm: SC PSPACE}), the entire procedure can be carried out using polynomial space.
\begin{theorem}\label{thrm: SC sat PSPACE}
    The satisfiability problem for $\SC$ is $\mathrm{PSPACE}$-complete.
\end{theorem}

Concerning the satisfiability problem for $\MSC$ and its extensions, it follows from the results of \cite{Kuusisto13, KuusistoAntti2020Epfd} that it is undecidable for $\MSC$, $\GMSC$ and $\GGMSC$.

\subsection{Universal reduction for Turing machines}\label{sec: meta reduction}

In this section, we show that the logic $\MSC$ is Turing-complete in the following sense:
We provide a ``universal reduction'' $U$ such that, given as input a Turing machine $T$ and a time- or space-complexity function $s$ of $T$, $U$ gives an $s(n)$-time reduction from $T$ to the model checking problem for $\MSC$.
Moreover, if no time- or space-complexity function of $T$ is given as input to $U$, then $U$ still provides a computable reduction to the model checking problem for $\MSC$.

We begin by defining some notions of complexity functions.
We assume w.l.o.g. that all the Turing machines in this section use the same binary alphabet denoted by $\Pi$. 
For technical convenience, we define how a Turing machine produces an output. If a Turing machine enters a halting state (i.e., an accepting or a rejecting state) during a run with an input $\bw$, then it halts and produces an \textbf{output} which is the content of the tape in the last configuration of the run (excluding the symbol $E_\bl$ and the infinite string of blank symbols). Given a function $f \colon \N \to \N$, we say that it is a \textbf{proper complexity function} if $f$ is non-decreasing and the following holds: there is a Turing machine $T_f$ such that for any input $\bw$, the output of $T_f$ is $1^{f(\abs{\bw})}$, $T_f$ halts after $\ordo(\abs{\bw} + f(\abs{\bw}))$ steps, and $T_f$ uses $\ordo(f(\abs{\bw}))$ space.\footnote{A similar definition is used, for example, in \cite{papadimitriou2003computational}.} 
Such a function can be used as an upper bound for the amount of space or time used by a Turing machine for solving a problem. 
We say that a proper complexity function $s$ is a \textbf{space-complexity function} (resp. a \textbf{time-complexity function}) of a Turing machine $T$ if, for every input $\bw$, the head of $T$ visits at most $s(\abs{\bw})$ distinct cells of the work tape during its run with the input $\bw$ (resp. $T$ halts within $s(\abs{\bw})$ steps).

Next, we define some concepts related to the meta-reduction.
We let $T_{\MSC}$ denote the Turing machine which solves the model checking problem for $\MSC$, i.e., $T_\MSC$ takes as input a pair $((M,w), \Lambda)$, where $(M, w)$ is a pointed Kripke model and $\Lambda$ is a program of $\MSC$ (both over the same vocabulary).
The \textbf{meta-reduction} $\fm$ is a mapping that takes a Turing machine $T$ as an input and outputs $T_\MSC$ parametrized with the following program $\Lambda_T$ (with rejecting predicates).
The program $\Lambda_T$ is constructed from $T$ analogously to the proof of Theorem \ref{thrm: lbas equiv 2-way GMSC}, but with the following small modification: if the head simulated by the program reaches the right end marker (i.e., the node labeled by $p_\br$ in the extended word model; cf. Section \ref{section: language by MSC} for the definition of an extended word model), then the program stops updating its head predicates at each node. This can be done since, in each round, the position of the head is recorded in only one node, while the other nodes are ``inactive''.

Let $\cF(\Pi_+)$ denote the class of finite pointed Kripke models over $\Pi_+ \colonequals \Pi \cup \{p_\blacksquare, p_\bl, p_\br\}$, where $p_\blacksquare$, $p_\bl$ and $p_\br$ are defined as in Section \ref{section: language by MSC}.
The \textbf{input-reduction} $\mathfrak{i}$ is a mapping that takes a Turing machine $T$ as input and outputs a Turing machine $\mathfrak{i}_T$ that computes as follows with an input $\bw \in \Pi^*$. The machine $\mathfrak{i}_T$ simulates the run of $T$ with the input $\bw$, keeping track of the rightmost tape cell visited during the run. If $T$ halts, then $\mathfrak{i}_T$ outputs (an encoding of) the pointed $\ell$-extended word model $M^\ell_\bw$, where $\ell$ is chosen, based on the rightmost visited cell, such that the run of $T$ with the input $\bw$ fits in the tape $E_\bl \bw \blacksquare^{\ell} E_\br$. If $T$ does not halt with the input $\bw$, then neither does $\mathfrak{i}_T$. Note that the machine $\mathfrak{i}_T$ is obtained from $T$ in a computable way, and thus the mapping $\mathfrak{i}$ is computable.
Below, we write $\fm(T)$ \textbf{accepts} (resp. \textbf{rejects}) $\mathfrak{i}_T(\bw)$ if the Turing machine that first runs $\mathfrak{i}_T$ with the input $\bw$ and then runs $\fm(T)$ with the produced output accepts (resp. rejects). Note that if $T$ does not halt with the input $\bw$, then this composed machine does not halt either, and in particular it neither accepts nor rejects.

Now, it is straightforward to prove that we obtain the following ``universal reduction'' (essentially by using the same argument as in the proof of Theorem \ref{thrm: lbas equiv 2-way GMSC}).
\begin{theorem}\label{thrm: computable reduction}
    Informally, the universal reduction maps a Turing machine $T$ to the pair consisting of the machines $\mathfrak{i}_T$ and $\fm(T)$. With an input $\bw$, the machine $\mathfrak{i}_T$ simulates the run of $T$ and outputs the $\ell$-extended word model $M^\ell_\bw$, where $\ell$ is chosen based on the rightmost cell visited by $T$ with the input $\bw$. Then $\fm(T)$ checks whether $\Lambda_T$ accepts or rejects $(M^\ell_\bw, 1)$. The following theorem states that $T$ accepts (resp.\ rejects) $\bw$ iff $\fm(T)$ accepts (resp.\ rejects) $\mathfrak{i}_T(\bw)$, and that $\mathfrak{i}_T$ does not halt on $\bw$ iff $T$ does not. Importantly, both mappings $\mathfrak{i}$ and $\fm$ are computable.

    More precisely, for any Turing machine $T$, we have the following:
    \begin{enumerate}
        \item $T \text{ accepts } \bw \iff \fm(T) \text{ accepts } \mathfrak{i}_{T}(\bw)$,
        \item $T \text{ rejects } \bw \iff \fm(T)\text{ rejects } \mathfrak{i}_{T}(\bw)$,
        \item $T \text{ does not halt on } \bw \iff \mathfrak{i}_{T} \text{ does not halt on } \bw$.
    \end{enumerate}
\end{theorem}

In particular, the theorem above says that there is a computable reduction from any recursively enumerable problem to the model checking problem for $\MSC$.

Next, we consider reductions whose running time is bounded. Here, an \emph{$f$-time reduction} refers to a reduction that is computed by a Turing machine that halts within $\ordo(f(\abs{\bw}))$ steps with every input $\bw$.

Let $T$ be a Turing machine, and let $f$ be a time- or space-complexity function of $T$ that is superlinear, i.e., $f(n) \geq n$ for every $n \in \N$.\footnote{For simplicity, we restrict to superlinear functions since the characterization of Theorem \ref{thrm: lbas equiv 2-way GMSC} concerns linear tape-bounded Turing machines, and therefore it does not directly apply to sublinear functions.} Let $T_f$ be a Turing machine realizing $f$. We let $\mathfrak{i}_{T_f}$ denote the Turing machine that first applies the input-reduction to $T_f$ with the input $\bw$ and then extends the produced word model by one additional blank cell. We write $\fm(T)$ accepts (resp.\ rejects) $\mathfrak{i}_{T_f}(\bw)$ analogously to the case of $\mathfrak{i}_T$.

Since $T_f$ halts with every input, so does $\mathfrak{i}_{T_f}$. The output of $\mathfrak{i}_{T_f}$ with an input $\bw$ is an $\ell$-extended word model $M^\ell_\bw$ such that
$f(\abs{\bw}) + 1 \leq \abs{\bw} + \ell $ and $\abs{\bw} + \ell \in \ordo(f(\abs{\bw}))$.
The first bound holds since $T_f$ writes the string $1^{f(\abs{\bw})}$ on its tape and one blank cell is added. The second bound holds since $T_f$ uses $\ordo(f(\abs{\bw}))$ space and $f$ is superlinear.

The first bound guarantees that $M^\ell_\bw$ is large enough for simulating $T$. If $f$ is a space-complexity function of $T$, then the head of $T$ never moves beyond the cell $f(\abs{\bw})$. If $f$ is a time-complexity function of $T$, then the head of $T$ never moves beyond the cell $f(\abs{\bw}) + 1$, since $T$ makes at most $f(\abs{\bw})$ moves starting from the cell $1$. Hence, the head simulated by $\Lambda_T$ never reaches the right end marker, and, as in the proof of Theorem \ref{thrm: computable reduction}, $\fm(T)$ accepts (resp.\ rejects) $\mathfrak{i}_{T_f}(\bw)$ iff $T$ accepts (resp.\ rejects) $\bw$. Note that padding the word model further than necessary does not affect this, since the simulated head of $T$ then simply never reaches the right end marker. If $T$ does not halt with the input $\bw$, which is only possible if $f$ is a space-complexity function, then $\Lambda_T$ never accepts the produced model, and thus $\fm(T)$ rejects it.

The second bound determines the complexity of the reduction. Note that $\fm(T)$ does not depend on the input $\bw$: the reduction maps $\bw$ to the instance consisting of $\mathfrak{i}_{T_f}(\bw)$ and the program $\Lambda_T$, and $\Lambda_T$ is the same for all inputs once $T$ is fixed. Thus, for a fixed $T$, the program $\Lambda_T$ is fixed in the reduction, and it contributes only a constant amount of time, independent of $\abs{\bw}$. The complexity of the reduction is therefore determined by $\mathfrak{i}_{T_f}$. Encoding word models by their extended words $E_\bl \bw \blacksquare^\ell E_\br$, the second bound shows that the output can be written in $\ordo(f(\abs{\bw}))$ steps. Since $T_f$ halts within $\ordo(\abs{\bw} + f(\abs{\bw})) = \ordo(f(\abs{\bw}))$ steps (note that $f$ is superlinear), the machine $\mathfrak{i}_{T_f}$ halts within $\ordo(f(\abs{\bw}))$ steps. Altogether, we obtain the following theorem.
\begin{theorem}\label{thrm: time complexity reduction}
    Informally, if a superlinear time- or space-complexity function $f$ of $T$ is given, the universal reduction maps $T$ to the pair consisting of the machines $\mathfrak{i}_{T_f}$ and $\fm(T)$. With an input $\bw$, the machine $\mathfrak{i}_{T_f}$ runs the machine $T_f$ realizing $f$ and outputs an $\ell$-extended word model $M^\ell_\bw$ that is large enough to contain every run of $T$ within the given time or space bound. Unlike $\mathfrak{i}_T$, $\mathfrak{i}_{T_f}$ does not simulate $T$, and thus it always halts, within $\ordo(f(\abs{\bw}))$ steps. Then $\fm(T)$ checks whether $\Lambda_T$ accepts or rejects $(M^\ell_\bw, 1)$. The following theorem states that $T$ accepts $\bw$ iff $\fm(T)$ accepts $\mathfrak{i}_{T_f}(\bw)$, and that $T$ rejects or does not halt on $\bw$ iff $\fm(T)$ rejects $\mathfrak{i}_{T_f}(\bw)$. Rejection and non-termination are combined since $\mathfrak{i}_{T_f}$ always halts, and thus $\fm(T)$ always gives an answer.

    More precisely, let $T$ be a Turing machine, and let $f$ be a superlinear time- or space-complexity function of $T$. Then $\mathfrak{i}_{T_f}$ and $\fm(T)$ give an $f$-time reduction from $T$ to the model checking problem for $\MSC$, i.e., the following holds:
    \begin{enumerate}
        \item $T \text{ accepts } \bw \iff \fm(T) \text{ accepts } \mathfrak{i}_{T_f}(\bw)$,
        \item $T \text{ rejects or does not halt on } \bw \iff \fm(T) \text{ rejects } \mathfrak{i}_{T_f}(\bw)$.
    \end{enumerate}
\end{theorem}

\section{Conclusion}

We have introduced two new game-theoretic semantics for SC, MSC, GMSC, and GGMSC. We also studied asynchronous variants of these logics and showed that over finite models GGMSC with asynchronous semantics has the same expressive power as global graded modal computation logic and the $\mu$-fragment of the global graded modal $\mu$-calculus, and that an analogous result holds for GMSC and MSC. We also defined a formula size game for GGMSC and its variants, which characterizes the equivalence of classes of pointed Kripke models up to programs of a given size. 
We proved that, over words, GGMSC, GMSC, MSC, and linear tape-bounded Turing machines have the same expressive power. We also obtained some nice additional results, such as the fact that both the combined and data complexity of model checking for GGMSC, GMSC, and MSC are PSPACE-complete, while both the combined and data complexity for GGMSC, GMSC, and MSC with the asynchronous semantics are PTIME-complete.
We also showed that the combined complexity of the model checking problem for $\SC$ is PSPACE-complete, while for the asynchronous variant of $\SC$ it is PTIME-complete. As a corollary, we obtained that the satisfiability problem for $\SC$ is PSPACE-complete, and NP-complete for the asynchronous variant of $\SC$.
In Section~\ref{sec: meta reduction}, we provided a uniform method that reduces any recursively enumerable problem to the model checking problem for $\MSC$, provided that the input may be padded.

Future research directions involve studying the succinctness of SC, MSC, GMSC, and GGMSC by applying formula size games. Recall that many variants of MSC have recently been used to characterize several important classes of computational models, for example, neural networks. Thus, succinctness results for these logics could support the theoretical study of the succinctness of the classes they characterize. 
For instance, \cite{10.1093/logcom/exag011} characterizes recurrent neural networks with floating-point numbers and piecewise polynomial activation functions via SC and shows that each recurrent neural network can be translated into an equivalent program of SC with a polynomial blow-up. Moreover, for a fixed floating-point format and fixed parameters of the activation functions, the translation gives only a linear blow-up. One could then use, for example, the formula size game to show that a property $P$ is not expressible by any SC-program of size at most $s$ and then transfer this lower bound result to concern neural networks via the translation.

\subsection*{Acknowledgments}

The authors are listed in alphabetical order. Veeti Ahvonen was supported by the \emph{Vilho, Yrjö and Kalle Väisälä Foundation} of the Finnish Academy of Science and Letters. Antti Kuusisto was supported by the project \emph{Perspectives on computational logic}, funded by the Research Council of Finland, project number 369424. Antti Kuusisto was also supported by the Research Council of Finland projects \emph{Theory of computational logics} (grant numbers 352419, 352420, 353027, 324435, 328987) and \emph{Explaining AI via Logic} (XAILOG) (grant number 345612).

\bibliography{references}

\end{document}